\documentclass{article}
\usepackage{arxiv}
\usepackage[utf8]{inputenc} % allow utf-8 input
\usepackage[T1]{fontenc}    % use 8-bit T1 fonts
\usepackage{hyperref}       % hyperlinks
\usepackage{url}            % simple URL typesetting
\usepackage{booktabs}       % professional-quality tables
\usepackage{amsfonts}       % blackboard math symbols
\usepackage{nicefrac}       % compact symbols for 1/2, etc.
\usepackage{microtype}      % microtypography
\usepackage{lipsum}
\usepackage{graphicx}
\usepackage{graphicx,subfigure}
\usepackage{balance}  % for  \balance command ON LAST PAGE  (only there!)
\usepackage{booktabs} 
\usepackage{amsmath,amssymb,amsfonts,amsthm}
\usepackage{subfigure}
\usepackage{enumerate}
\usepackage[shortlabels]{enumitem}
\usepackage{xspace}
\usepackage{tikz}
\usepackage{multirow}
\usepackage{weiwAlgorithm}
\usepackage{array, makecell}
\usepackage{url}
\usepackage{diagbox}
\usepackage{textcomp}

\graphicspath{ {./figures/} }

\title{A Framework to Quantify Approximate Simulation on Graph Data}

\author{
 Xiaoshuang Chen\\
  University of New South Wales\\
  Sydney, Australia \\
  \texttt{xiaoshuang.chen@unsw.edu.au} \\
   \And
 Longbin Lai \\
  Alibaba Group\\
  Hangzhou, China \\
  \texttt{longbin.lailb@alibaba-inc.com} \\
  \And
 Lu Qin \\
  University of Technology Sydney\\
  Sydney, Australia \\
  \texttt{lu.qin@uts.edu.au} \\
  \And
  Xuemin Lin \\
  University of New South Wales\\
  Sydney, Australia \\
  \texttt{lxue@cse.unsw.edu.au} \\
  \And
  Boge Liu\\
  University of New South Wales\\
  Sydney, Australia \\
  \texttt{boge.liu@unsw.edu.au} \\
}

\newtheorem{definition}{Definition}

\newtheorem{theorem}{Theorem}
\newtheorem{example}{Example}

\newtheorem{corollary}{Corollary}
\newtheorem{remark}{Remark}

\long\def\comment#1{}
\newcommand{\stitle}[1]{\vspace{1ex}\noindent{{\bf#1}}}
\newcommand{\sstitle}[1]{\vspace{1ex}\noindent{\underline{\textit{#1}}}}
\newcommand{\ssstitle}[1]{\vspace{1ex}\noindent{\textit{#1}}}

\newcommand{\kw}[1]{{\ensuremath {\mathsf{#1}}}\xspace}

\newcommand{\Sim}{\kw{FSim}}
\newcommand{\lfunc}{\ell}
\newcommand{\inifunc}{\mathcal{L}}

\newcommand{\argmax}{\mathop{\mathrm{argmax}}}
\newcommand{\normopr}{\Omega}
\newcommand{\mappingopr}{\mathcal{M}}
\newcommand{\lparam}{(1-{w^+}-{w^-})}
\newcommand{\outparam}{w^+}
\newcommand{\inparam}{w^-}
\newcommand{\simu}{\kw{s}}
\newcommand{\dpsim}{\kw{dp}}
\newcommand{\bisim}{\kw{b}}
\newcommand{\bjsim}{\kw{bj}}
\newcommand{\ub}{\kw{ub}}
\newcommand{\mysim}{\rightsquigarrow}
\newcommand{\PCRW}{\kw{PCRW}}
\newcommand{\PathSim}{\kw{PathSim}}
\newcommand{\JoinSim}{\kw{JoinSim}}
\newcommand{\nSimGram}{\kw{nSimGram}}
\newcommand{\SimRank}{\kw{SimRank}}
\newcommand{\RoleSim}{\kw{RoleSim}}

\newcommand{\PathSims}{\kw{PathSim}\xspace}

\newcommand{\reffig}[1]{Figure~\ref{fig:#1}}
\newcommand{\refsec}[1]{Section~\ref{sec:#1}}
\newcommand{\reftab}[1]{Table~\ref{tab:#1}}
\newcommand{\refalg}[1]{Algorithm~\ref{alg:#1}}
\newcommand{\refeq}[1]{Equation~\ref{eq:#1}}
\newcommand{\refdef}[1]{Definition~\ref{def:#1}}
\newcommand{\refthm}[1]{Theorem~\ref{thm:#1}}

\newcommand{\refrem}[1]{Remark~\ref{rem:#1}}

\newcommand{\refex}[1]{Example~\ref{ex:#1}}

\newcommand{\mapcurr}{$H_{\kw{c}}$\xspace}

\newcommand{\topcaption}{%
	\setlength{\abovecaptionskip}{0.01cm}%
	\setlength{\belowcaptionskip}{0.01cm}%
	\caption}

\begin{document}
\maketitle

\newcommand\blfootnote[1]{%
  \begingroup
  \renewcommand\thefootnote{}\footnote{#1}%
  \addtocounter{footnote}{-1}%
  \endgroup
}

% \blfootnote{This paper has been accepted by ICDE 2021}

\begin{abstract}
Simulation and its variants (e.g., bisimulation and degree-preserving simulation) are useful in a wide spectrum of applications. However, all simulation variants are coarse ``yes-or-no'' indicators that simply confirm or refute whether one node simulates another, which limits the scope and power of their utility. Therefore, it is meaningful to develop a fractional $\chi$-simulation measure to quantify the degree to which one node simulates another by the simulation variant $\chi$. To this end, we first present several properties necessary for a fractional $\chi$-simulation measure. Then, we present $\Sim_\chi$, a general fractional $\chi$-simulation computation framework that can be configured to quantify the extent of all $\chi$-simulations. Comprehensive
experiments and real-world case studies show the measure to be effective and the computation framework to be efficient.
\end{abstract}

% keywords can be removed

\section{Introduction} \label{sec:introduction}
{\color{black}Consider two directed graphs $G_1$ and $G_2$ with labeled nodes from the sets $V_1$ and $V_2$, respectively. A \emph{simulation} \cite{DBLP:journals/pvldb/MaCFHW11} relation $R \subseteq V_1 \times V_2$ is a binary relation over $V_1$ and $V_2$. For each node pair $(u, v)$ in $R$ (namely, $u$ is simulated by $v$), each $u$'s out-neighbor\footnote{\color{black}A node $u'$ is an out-neighbor of $u$, if there is an outgoing edge from $u$ to $u'$ in $G$. Similarly, $u''$ is an in-neighbor of $u$, if an edge from $u''$ to $u$ presents.} is simulated by one of $v$'s out-neighbors, and the same applies to in-neighbors. An illustration of this concept is shown below.}

\begin{example} \label{ex:simulation_definition_examples}
As shown in \reffig{example_graphs}, node $u$ is simulated by node $v_2$, as they have the same label, and each $u$'s out-neighbor can be simulated by the same-label out-neighbor of $v_2$ ($u$ has no in-neighbors). Note that the two hexagonal nodes in $\mathcal{P}$ are simulated by the same hexagonal node in $\mathcal{G}_2$. Similarly, $u$ is simulated by $v_3$ and $v_4$. However, $u$ can not be simulated by $v_1$, as the pentagonal neighbor of $u$ cannot be simulated by any neighbor of $v_1$.
\end{example}
% \vspace{-0.5em}

% On the basis of simulation relation, there are other variants defined by adding more constraints. For example, \emph{bisimulation} \cite{milner1989communication} constrains the simulation relation to be symmetric, namely if $u$ is simulated by $v$, then $v$ must be simulated by $u$; and \emph{degree-preserving simulation} \cite{DBLP:journals/pvldb/SongGCW14} requires that two neighbors of $u$ cannot be simulated by the same neighbor of $v$.
{\color{black}The original definition of simulation put forward by Milner in 1971 \cite{DBLP:conf/ijcai/Milner71} only considered out-neighbors. But, in 2011, Ma et al. \cite{DBLP:journals/pvldb/MaCFHW11} revised the definition to consider in-neighbors, making it capture more topological information. Additionally, different variants of simulation have emerged over the years, each with its own constraint(s). For example, on the basis that $R$ is a simulation relation, \emph{bisimulation} \cite{milner1989communication} further requires that $R^{-1}$ is also a simulation, where $R^{-1}$ denotes the converse relation of $R$ (i.e., $R^{-1} = \{(v,u) | \forall (u, v) \in R\}$); and \emph{degree-preserving simulation} \cite{DBLP:journals/pvldb/SongGCW14} requires that two neighbors of $u$ cannot be simulated by the same neighbor of $v$.}

\stitle{Applications.} Simulation and its variants are important relations among nodes, and have been adopted in a wide range of applications. For example, simulation and degree-preserving simulation are shown to be effective in graph pattern matching \cite{DBLP:journals/pvldb/FanLMTWW10,DBLP:journals/pvldb/MaCFHW11,DBLP:journals/tods/MaCFHW14,DBLP:journals/pvldb/SongGCW14}, and {\color{black}a node in the data graph is considered to be a potential match for a node in the query graph if it simulates the query node. Bisimulation has been applied to compute RDF graph alignment \cite{DBLP:journals/pvldb/BunemanS16} and {\color{black}graph partition \cite{DBLP:conf/sigmod/HellingsFH12, DBLP:conf/sigmod/SchatzleNLP13,DBLP:conf/sac/HeeswijkFP16}}.} Generally, two nodes will be aligned or be placed in the same partition if they are in a bisimulation relation. Other applications include data retrieval \cite{DBLP:conf/vldb/Ramanan03}, {\color{black}graph compression \cite{DBLP:conf/sigmod/FanLWW12} and index construction \cite{DBLP:journals/is/FletcherGWGBP09,DBLP:conf/sigmod/KaushikBNK02,DBLP:journals/corr/abs-2003-03079}}, etc. %Note that the applications listed here are all fundamental computation problems in graph analysis, and each of them has wide real-life applications, e.g., pattern matching in

%Compared with the NP-complete subgraph isomorphism that is often used in graph pattern matching \cite{Bi2016, Lai2019}, simulation alternatives are more efficient to compute as they are polynomially solvable. We show how simulation \cite{DBLP:journals/pvldb/FanLMTWW10} facilitates graph pattern matching in the following example.

\comment{
\begin{example} \label{ex:pattern_matching_example}
In the Amazon co-purchasing graph $G$, books are the vertices labeled by their categories, and an edge from book A to B indicates that people are very likely to buy B when they buy A. Suppose a user wants to search for "Parenting" books that are connected by books of "Children", "Home" and "Health" (mutually), which can be modelled as a pattern graph $P$ as shown in \reffig{intro_query}. The simulation-based algorithm first considers all the nodes of the same label as the candidates (i.e. book 89985 and book 3004 in \reffig{intro_data_graph}). Then for each candidate $v$, it encloses a sub-graph $G_v$ (as highlighted) induced by the candidate nodes and all its one-hop neighbors. Thereafter, the algorithm computes simulation relation between $P$ and $G_v$, and includes such a candidate $v$ in the results while all vertices of $P$ present in the simulation relation. In \reffig{pattern_matching_example}, one can verify that the book 89985 is the only result.
\end{example}
}

\begin{figure}
    \centering
	\includegraphics[width=0.7\linewidth]{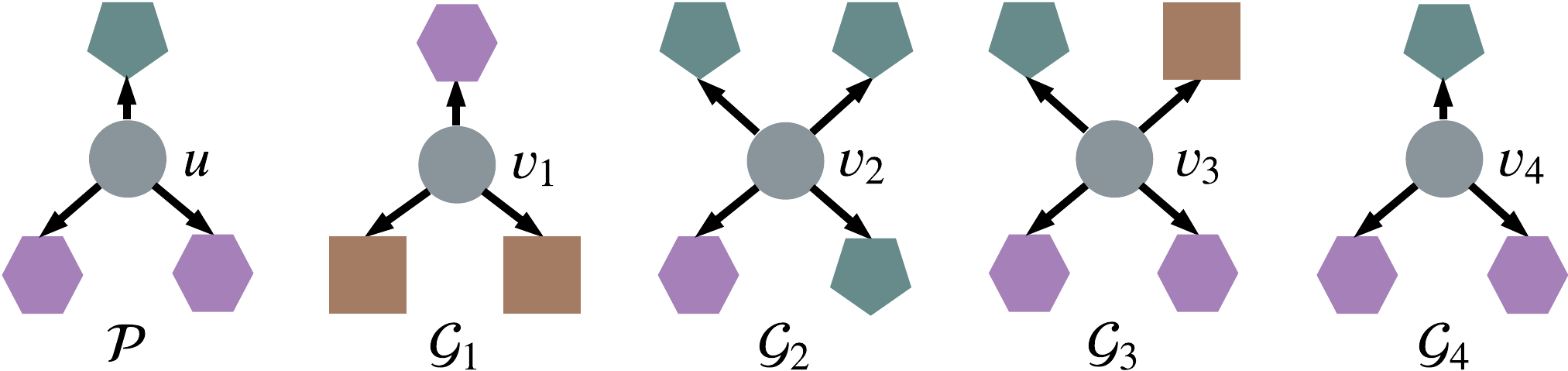} \label{fig:example_graphs}
    \topcaption{Example graphs. A node's shape denotes its label.} \label{fig:example_graphs}
    % \vspace{-1.8em}
\end{figure}

\stitle{Motivations.} {\color{black}Despite their numerous valuable uses, simulation and its variants are all coarse ``yes-or-no'' indicators. That is, simulation and its variants can only answer whether node $v$ can fully simulate node $u$; they cannot tell us whether $v$ might be able to partially or even very nearly (i.e., approximately) simulate $u$. This coarseness raises two practical issues.} First, there often exist some nodes that nearly simulate $u$ in real-world graphs, which either naturally present in the graphs or are consequences of data errors (a common issue by data noise and mistakes of collecting data). However, simulation and its variants cannot catch these nodes and cause loss of potential results. Second, the coarseness makes it inappropriate to apply simulation and its variants to applications that naturally require fine-grained evaluation, such as node similarity measurement. %coarse-grained indicators are only appropriate for coarse-grained evaluations. The ability to quantify the extent of simulation would open up a host of possibilities for using simulation in fine-grained evaluations, such as node similarity measurement.
\refex{pattern_matching_example} provides a real-life illustration of these issues.

\begin{example} \label{ex:pattern_matching_example}
We consider the application of simulation to determine whether or not a poster $A$ is simulated by another poster $B$ in terms of their design elements (e.g., color, layout, font, and structure). For example, when compared with the poster $P_1$ in \reffig{data_poster}, the candidate poster $P$ in \reffig{query_poster} only slightly differs in the font and font style. Hence, it is highly suspected as a case of plagiarism \cite{solo-wiki}. {\color{black}Nevertheless, due to a minor change of design elements, there is no exact simulation relation between posters $P$ and $P_1$, and thus exact simulation can not be used to discover such similarity.} As a result, it is more desirable to develop a mechanism to capture the similarity between two posters via the degree of approximate simulation (some fine-grained measurement), instead of simply using the exact simulation. 
\end{example}
% \vspace{-0.3em}

\begin{figure}
    \centering
    \subfigure[A poster]{ \includegraphics[width=0.095 \linewidth]{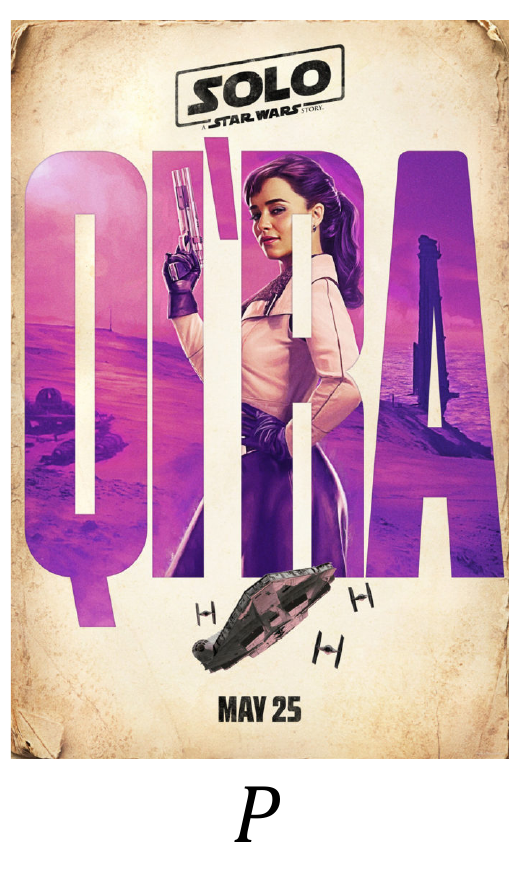} \label{fig:query_poster}}
    \hspace{1em}
    \subfigure[A database of existing posters]{
    \includegraphics[width=0.3\linewidth]{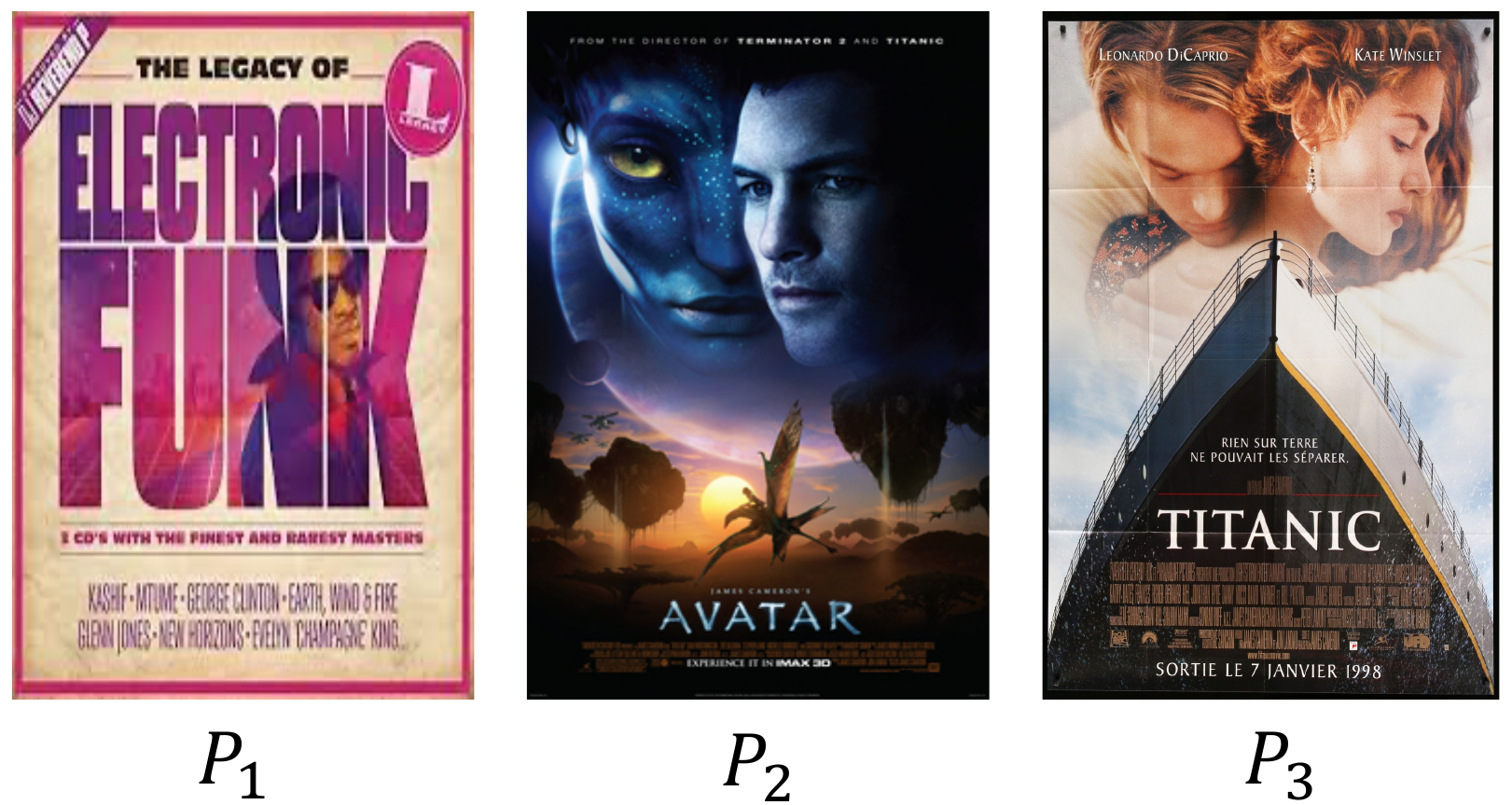} \label{fig:data_poster}
    }
    \hspace{-0.6em}
    \subfigure[Query graph]{ \includegraphics[width=0.16\linewidth]{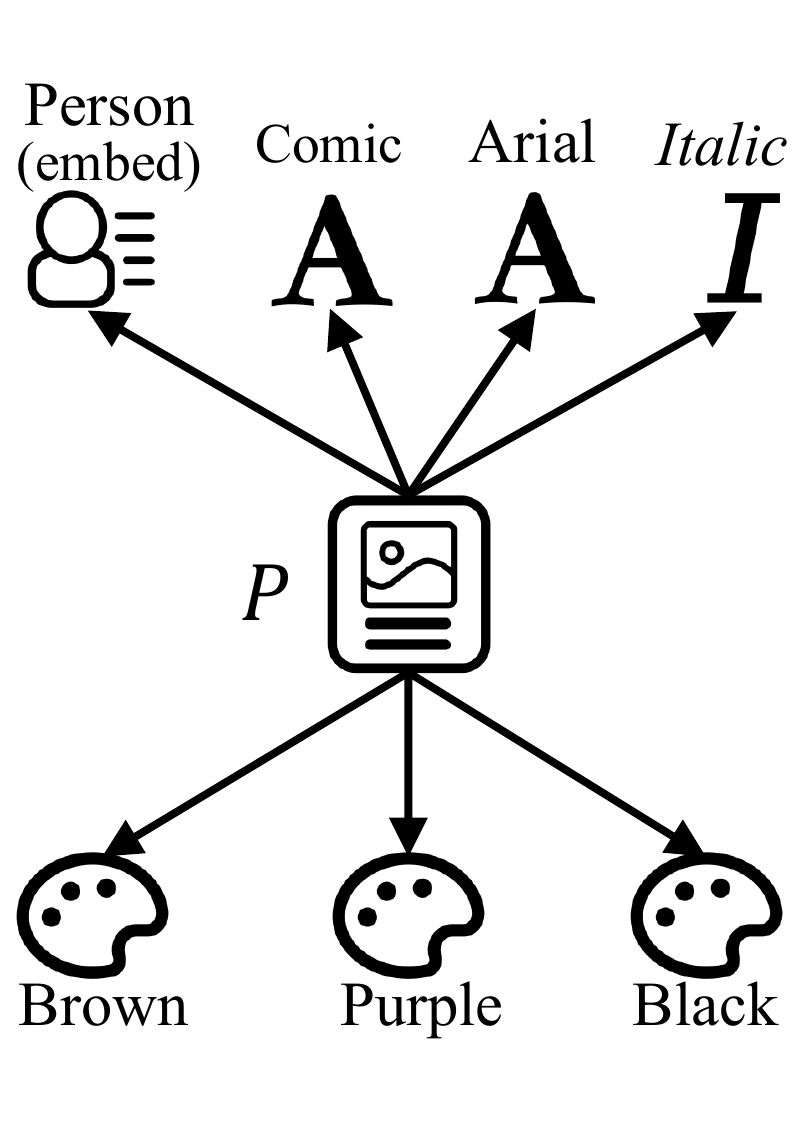} \label{fig:query_graph}}
    \hspace{0.2em}
    \subfigure[Data graph]{
    \includegraphics[width=0.3\linewidth]{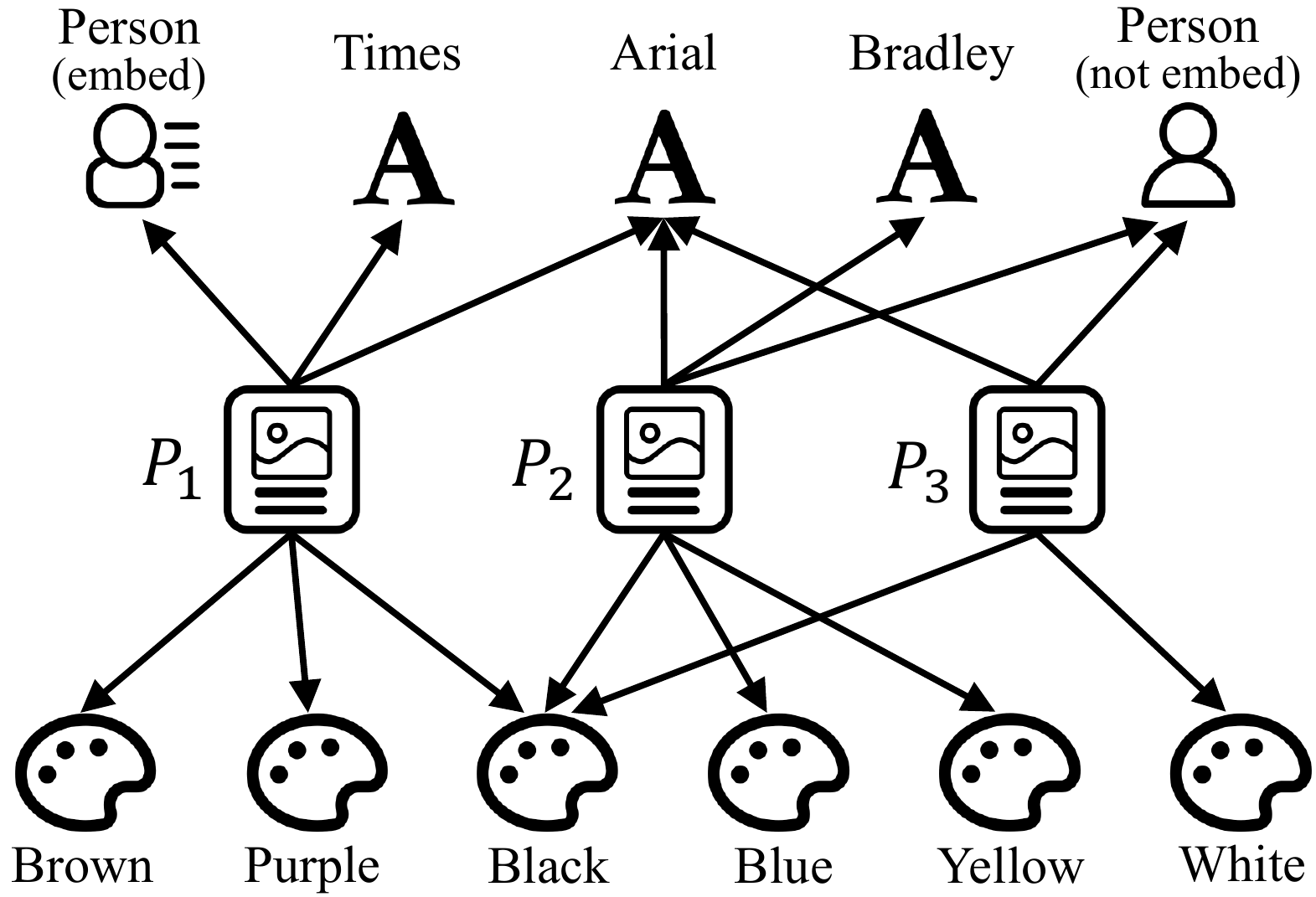} \label{fig:data_graph}
    }
    %\subfigure[Promising results in a fragment of the Amazon graph]{\includegraphics[width=0.71\linewidth]{figures/pattern_matching_data_graph_0213.pdf} \label{fig:intro_data_graph}} \\
    \topcaption{Motivating example. Figures (c) and (d) are graphs representing the posters in (a) and (b), respectively. {\color{black}Nodes are marked with their labels. An edge from nodes u to v indicates that the poster u has a design element v.}}% where a node's shape denotes its book category and an edge from a to b indicates that people are highly likely to buy b after buying a.}
    \label{fig:pattern_matching_example}
    \vspace{-1.8em}
\end{figure}

In general, it is of practical need to develop a mechanism to quantify the cases of approximate simulation to remedy the impacts of the ``yes-or-no'' semantics. {\color{black}Such quantification can not only open up a host of possibilities for using simulation but also make the results of simulation more effective and robust. Although the simulation variants differ in certain properties, they are actually derived from a common foundation, namely the simulation relation \cite{DBLP:conf/ijcai/Milner71}. Consequently, instead of developing a quantification technique independently and individually for each variant, it is more desirable to devise a general framework that works for all simulation variants. Aside from the obvious benefits of less redundancy, developing a unified framework requires a systematic study of the properties of the different simulation variants. Not only has this not been done before, doing so may help to inspire new variants.} 

\stitle{Our Contributions.} We propose the \emph{fractional $\chi$-simulation framework} that quantifies the extent of simulation and its variants in the range of $[0,1]$. Our main contributions are as follows.%We propose the \emph{fractional $\chi$-simulation framework} that returns a score within $[0,1]$ to quantify the approximate simulation for different simulation variant $\chi$. Our main contributions are as follows.

{\color{black}\ssstitle{(1) A unified definition of $\chi$-simulation.} From a systematic study of the properties of simulation and its variants, we distill the base definition of simulation and its variants into a unified definition called $\chi$-simulation. Further, we discover and name a new simulation variant - \emph{bijective simulation}. Theoretically, bijective simulation is akin to the well-known Weisfeiler-Lehman isomorphism test \cite{DBLP:journals/jmlr/ShervashidzeSLMB11} (\refsec{discussion}). Practically, its fractional form (contribution 2) is more effective than the existing models regarding node similarity measurement, as detailed in \refsec{case_studies}.}

{\color{black}\ssstitle{(2) A general framework $\Sim_{\chi}$ for computing fractional $\chi$-simulation.} To quantify the degree to which one node simulates another by a $\chi$-simulation, {\color{black} we propose the concept of \emph{fractional $\chi$-simulation} and identify a list of properties that a \emph{fractional $\chi$-simulation} measure should satisfy.} Then, we present a general computation framework, namely $\Sim_{\chi}$, which can be configured to compute fractional $\chi$-simulation for all $\chi$-simulations with the properties satisfied. %quantify all $\chi$-simulations. 
$\Sim_{\chi}$ is an iterative framework that computes the fractional $\chi$-simulation scores for all pairs of nodes over two graphs. Furthermore, we show the relations of $\Sim_{\chi}$ to several well-known concepts, including node similarity measures (i.e., \kw{SimRank} \cite{DBLP:conf/kdd/JehW02} and \kw{RoleSim} \cite{DBLP:conf/kdd/JinLH11}) and an approximate variant of bisimulation (i.e., $k$-bisimulation \cite{DBLP:books/cu/12/AcetoIS12,DBLP:conf/bncod/LuoLFBHW13,DBLP:conf/cikm/LuoFHWB13,DBLP:conf/sac/HeeswijkFP16}), in \refsec{discussion}}. 

\ssstitle{(3) Extensive experiments and case studies.} We perform empirical studies to exhibit that $\Sim_\chi$ is robust to parameter tuning and data errors, and is efficient to compute on real-world graphs. We further conduct three case studies to evaluate $\Sim_\chi$'s potential for subgraph pattern matching, node similarity measurement, and RDF graph alignment. {\color{black}Based on these studies, we reach the following conclusions. First, fractional $\chi$-simulation can remedy the ``yes-or-no" semantics of $\chi$-simulation, and it significantly improves the effectiveness of $\chi$-simulation in the related applications, e.g., simulation in subgraph pattern matching. Second, fractional bijective simulation (proposed in this paper) is a highly effective way of measuring node similarity. Finally, the $\Sim_\chi$ framework provides a flexible way to study the effectiveness of different simulation variants, and thus can be used as a tool to help identify the best variant for a specific application.}

\section{Simulation and Its Variants} \label{sec:preliminary}
\stitle{Data Model.} {\color{black}%In this paper, we consider 
Consider a node-labeled directed graph $G=(V, E, \lfunc)$, where $V(G)$ and $E(G)$ denote the node set and edge set, respectively (or $V$ and $E$ when the context is clear). $\Sigma$ is a set of string labels, and $\lfunc: V \rightarrow \Sigma$ is a labeling function that maps each node $u$ to a label $\lfunc(u) \in \Sigma$. $N_G^+(u) = \{u' | (u, u') \in E(G)\}$ denotes node $u$'s out-neighbors and, likewise, $N_G^-(u) = \{u' | (u', u) \in E(G)\}$ denotes its in-neighbors. Let $d_G^+(u) = |N_G^+(u)|$ and $d_G^-(u) = |N_G^-(u)|$ be the out- and in-degrees of node $u$, and let $d_G$, $D^+_G$ and $D^-_G$ denote the average degree, maximum out-degree and maximum in-degree of $G$, respectively. A summary of the notations used throughout this paper appears in \reftab{Notations}.}

\begin{table}[h]
	\centering
 	\small
 	\color{black}
 	\topcaption{Table of Notations} \label{tab:Notations}
 	\renewcommand{\arraystretch}{1.1}
    % \scalebox{0.86}{
	\begin{tabular}{|c|c|} \hline
		\textbf{Notation}&\textbf{Description}\\ \hline
		$G = (V, E, \lfunc)$ & a node-labeled directed graph  \\ \hline
		$V(G)/E(G)$ & the node/edge set of graph $G$  \\ \hline
		$\lfunc(\cdot)$ & a labeling function \\ \hline 
% 		$(s,t)$ & an edge from node $s$ to node $t$ \\ \hline
        $N_G^+(u)/N_G^-(u)$ & the out-neighbors/in-neighbors of node $u$ in $G$\\ \hline
        $d_G^+(u)/d_G^-(u)$ & the out-degree/in-degree of node $u$ in $G$ \\ \hline
		$d_G$ & the average degree of $G$ \\ \hline
		$D^+_G/D^-_G$ & the maximum out-degree/in-degree of $G$ \\ \hline
	\end{tabular}
	%}
% 	\vspace{-2em}
\end{table}

\stitle{Simulation Variants.} {\color{black}The first step in developing a unified definition of simulation and its variants is to formally define simulation as the foundation of all its variants.}
%We first formally define simulation as the foundation of all variants.
\vspace{-0.3em}
\begin{definition} \label{def:simulation}
    \textsc{(Simulation)} Given the graphs $G_1=(V_1, E_1, \lfunc_1)$ and $G_2=(V_2, E_2, \lfunc_2)$\footnote{$G_1 = G_2$ is allowed in this paper.}, a binary relation $R \subseteq V_1 \times V_2$ is a simulation if, for $\forall (u,v) \in R$, it satisfies that:
    \begin{enumerate}[(1)] \setlength{\itemsep}{0cm}
    \item $\lfunc_1(u) = \lfunc_2(v)$,
    \item $\!\forall u' \in N_{G_1}^+(u)$, $\!\exists v' \in N_{G_2}^+(v)\!$ such that (s.t.) $(u', v') \in R$,
    \item $\forall u'' \in N_{G_1}^-(u)$, $\exists v'' \in N_{G_2}^-(v)$ s.t. $(u'', v'') \in R$. 
    \end{enumerate}
\end{definition}
\vspace{-0.5em}
% For a concise and clear presentation, we always write $u$ as a node from $V_1$ and $v$ as a node from $V_2$ in this paper. 
For clarity, $u$ is always a node from $V_1$, and $v$ is always a node from $V_2$ in this paper. 

% Based on \refdef{simulation}, several variants of simulation are defined with more constraints. We consider three most notable ones in the literature, namely \emph{strong simulation} \cite{DBLP:journals/pvldb/MaCFHW11}, \emph{degree-preserving simulation} \cite{DBLP:journals/pvldb/SongGCW14} and \emph{bisimulation} \cite{milner1989communication} (others are surveyed in \refsec{related_work}). {\color{black}Specifically, node $u$ is simulated by node $v$ regarding strong simulation \cite{DBLP:journals/pvldb/MaCFHW11} if there exists a simulation relation $R$ between $G_1$ and $G_2[v,\delta]$ with $(u,v) \in R$, in which $G_2[v, \delta]$ is the induced subgraph that includes all nodes whose shortest distance to $v$ is not larger than the diameter $\delta$ of graph $G_1$. Therefore, strong simulation is performing simulation (\refdef{simulation}) by nature,} and will not be further discussed. We next use a unified definition of $\chi$-simulation to summarize all the studied variants.

{\color{black}The variants of simulation are based on \refdef{simulation} but have additional constraints. \refdef{simulation_variants} below provides a summary of several common simulation variants. However, one exceptional variant, \emph{strong simulation} \cite{DBLP:journals/pvldb/MaCFHW11}, must be discussed first. Strong simulation is designed for subgraph pattern matching. In brief, strong simulation exists between the query graph $Q$ and data graph $G$ if a subgraph $G[v,\delta_Q]$ of $G$ satisfies the following criteria:  (1) a simulation relation $R$ exists between $Q$ and $G[v,\delta_Q]$; and (2) $R$ contains node $v$ and all nodes in $Q$. Note that the subgraph $G[v,\delta_Q]$ is an induced subgraph that includes all nodes whose shortest distances to $v$ in $G$ are not larger than the diameter $\delta_Q$ of $Q$. In essence, strong simulation essentially performs simulation (\refdef{simulation}) multiple times, and so does not need to be specifically defined or further discussed.

\refdef{simulation_variants}, which follows, shows how $\chi$-simulation summarizes the base definition of simulation but also considers its variants. The definition below includes two notable ones in  \emph{degree-preserving simulation} \cite{DBLP:journals/pvldb/SongGCW14} and \emph{bisimulation} \cite{milner1989communication}.
}

\begin{definition} \label{def:simulation_variants}
    \textsc{($\chi$-simulation)} A simulation relation $R$ by \refdef{simulation} is further a $\chi$-simulation relation, which corresponds to 
        \begin{itemize}[noitemsep]
        \item \textbf{Simulation} ($\chi = \simu$): no extra constraint;
        % \item \textbf{Degree-preserving simulation} ($\chi = \dpsim$): if $(u, v) \in R$, (1) $\forall u' \in N^+(u)$, there exists an \textbf{injective} function $\lambda_1: N^+(u) \to N^+(v)$, s.t. $(u', \lambda_1(u')) \in R$; and (2) $\forall u'' \in N^-(u)$, there exists an \textbf{injective} function $\lambda_2: N^-(u) \to N^-(v)$, s.t. $(u'', \lambda_2(u'')) \in R$;
        {\color{black}
        \item \textbf{Degree-preserving simulation} ($\chi = \dpsim$): if $(u, v) \in R$, (1) there exists an \textbf{injective} function $\lambda_1: N_{G_1}^+(u) \to N_{G_2}^+(v)$, s.t. $\forall u' \in N_{G_1}^+(u)$, $(u', \lambda_1(u')) \in R$; and (2) there exists an \textbf{injective} function $\lambda_2: N_{G_1}^-(u) \to N_{G_2}^-(v)$, s.t. $\forall u'' \in N_{G_1}^-(u)$, $(u'', \lambda_2(u'')) \in R$;
        \item \textbf{Bisimulation} ($\chi$ = \bisim): if $(u, v) \in R$, (1) $\forall v' \in N^+(v)$, $\exists u' \in N^+(u)$ s.t. $(u', v') \in R$; and (2) $\forall v'' \in N^-(v)$, $\exists u'' \in N^-(u)$ s.t. $(u'', v'') \in R$. 
        % \item \textbf{Bisimulation} ($\chi$ = \bisim): $R^{-1}$ is a simulation relation as well, in which $R^{-1} = \{(v,u) | \forall (u, v) \in R\}$. 
        }
    \end{itemize}
    
    {\color{black}Node $u$ is $\chi$-simulated by node $v$ (or $v$ $\chi$-simulates $u$), denoted as $u \mysim^{\chi} v$, if there is a $\chi$-simulation relation $R$ with $(u,v) \in R$. Specifically, if $u \mysim^{\chi} v$ implies $v \mysim^{\chi} u$ (i.e., $\chi = \bisim$), we {\color{black}may use} $u \sim^{\chi} v$ directly.
    }
\end{definition}
% \vspace{-0.8em}

\begin{example} \label{ex:simulation_variants}
% Recall that in \refex{simulation_definition_examples}, $u$ is simulated by nodes $v_2$, $v_3$ and $v_4$ in \reffig{example_graphs}. However, $u$ cannot be $\dpsim$-simulated by $v_2$ because the two hexagon neighbors of $u$ must be mapped to the same hexagon neighbor of $v_2$, which contradicts the requirement of ``injective function''; $u$ cannot be $\bisim$-simulated by $v_3$, as the square neighbor of $v_3$ fails to simulate any neighbor of $u$.
Recall that in \refex{simulation_definition_examples}, $u$ is simulated by nodes $v_2$, $v_3$ and $v_4$ in \reffig{example_graphs}. However, $u$ cannot be $\dpsim$-simulated by $v_2$. This is because $u$ has two hexagonal neighbors and $v_2$ does not, which contradicts the requirement of ``injective function''; Analogously, $u$ cannot be $\bisim$-simulated by $v_3$, since $v_3$'s square neighbor fails to simulate any neighbor of $u$.
\end{example}
% \vspace{-0.5em}

% As a binary relation, it is natural to consider the properties of reflexivity, transitivity and symmetry. We are also interested in a property called \emph{injective neighbor mapping (IN-mapping)} inspired by the constraints of $\dpsim$-simulation. While all $\chi$-simulation trivially satisfy reflexivity (when two graphs are the same) and transitivity, we discuss in detail: (1) IN-mapping, that is $\forall (u,v) \in R$, two different neighbors (in and out) of $u$ cannot be mapped to the same neighbor of $v$; {\color{black}(2) symmetry, that is if $u \overset{\chi}\sim v$ then $v \overset{\chi}\sim u$. We summarize the properties of exiting simulation variants in \reffig{properties_of_variants}(a).}

{\color{black}Inspired by the constraints of $\dpsim$- and $\bisim$-simulations, we find that {a $\chi$-simulation} may have the following properties: (1) \emph{injective neighbor mapping} (or IN-mapping for short), i.e., $\forall (u,v) \in R$, two different neighbors (either in or out) of $u$ cannot be mapped to the same neighbor of $v$; and (2) \emph{converse invariant}, i.e., where $R^{-1} = \{(v,u) | \forall (u, v) \in R\}$ is a $\chi$-simulation if $R$ is a $\chi$-simulation. %Note that if $u\mysim^{\chi} v$ via a $\chi$-simulation with converse invariance, then $v\mysim^{\chi} u$ must hold. 
By \refdef{simulation_variants}, $\dpsim$-simulation has the property of IN-mapping, while $\bisim$-simulation has converse invariant. The properties of the exiting simulation variants are listed in \reffig{properties_of_variants}(a).}
% \vspace{-0.5em}
\begin{remark} 
\color{black}
Given a $\chi$-simulation with the property of converse invariant, if $u\mysim^{\chi} v$, then $v\mysim^{\chi} u$ must hold. Therefore, in \refdef{simulation_variants}, we have $u \mysim^{\bisim} v$ implies $v \mysim^{\bisim} u$.
\end{remark}
% \vspace{-0.5em}

\begin{figure}
    \centering
	\includegraphics[width=0.7\linewidth]{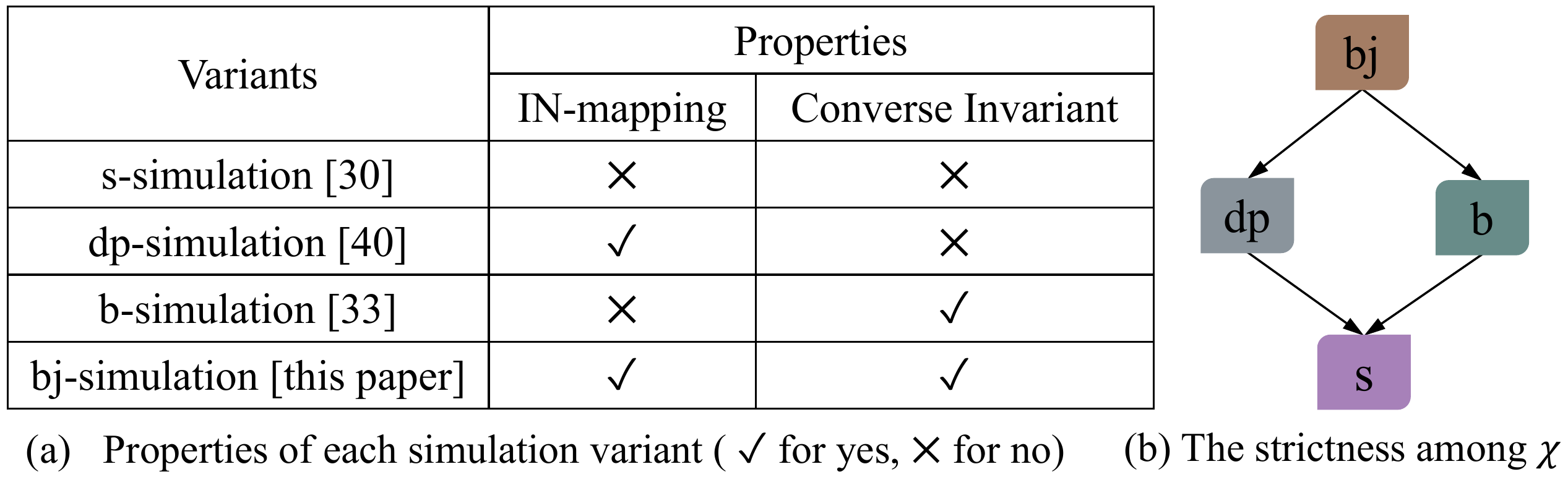}
    \topcaption{The summarization of all simulation variants.} \label{fig:properties_of_variants}
    % \vspace{-1.5em}
\end{figure}

\stitle{A New Variant: Bijective Simulation.} In compiling \reffig{properties_of_variants}(a), we realize that no simulation variant had both IN-mapping and converse invariant. This motivated us to define one. Called \emph{bijective simulation}, our definition follows. 

% \vspace{-0.3em}
{\color{black}
\begin{definition} \label{def:bijectivesimulation}
    % \textsc{(Bijective Simulation)} A simulation relation $R \subseteq V_1 \times V_2$ is a bijective simulation ($\chi = \bjsim$), if it further satisfies that, if $(u, v) \in R$, (1) there exists a \textbf{bijective} function $\lambda_1: N^+(u) \to N^+(v)$, s.t. $\forall u' \in N^+(u)$, $(u', \lambda_1(u')) \in R$; and (2) there exists a \textbf{bijective} function $\lambda_2: N^-(u) \to N^-(v)$, s.t. $\forall u'' \in N^-(u)$, $(u'', \lambda_2(u'')) \in R$. We incorporate bijective simulation in the $\chi$-simulation (\refdef{simulation_variants}) by letting $\chi = \bjsim$.
    \textsc{(Bijective Simulation)} A simulation relation $R \subseteq V_1 \times V_2$ is a bijective simulation ($\chi = \bjsim$), if $R$ is a degree-preserving simulation and the functions $\lambda_1$ and $\lambda_2$, as defined in \refdef{simulation_variants}, are further to be surjective (i.e., $\lambda_1$ and $\lambda_2$ are bijective). Bijective simulation is considered in the $\chi$-simulation (\refdef{simulation_variants}) by letting $\chi = \bjsim$.
\end{definition}}
% \vspace{-0.5em}

Compared to $\dpsim$-simulation, $\bjsim$-simulation requires that the mapping functions of the neighbors to be bijective. In other words, each pair of neighbors in a $\bjsim$-simulation must be mapped one by one. It's not hard to verify that $\bjsim$-simulation has the properties of both IN-mapping and {\color{black}converse invariance.}

\reffig{properties_of_variants}(b) shows the strictness among the simulation variants, where a ``more-strict'' edge from a $\chi_1$- to a $\chi_2$-simulation means that the $\chi_1$-simulation must also be a $\chi_2$-simulation. Such strictness among the variants can also be inferred from \reffig{example_graphs}. More specifically, given $u \mysim^{\bjsim} v_4$, it holds that $u \mysim^{\chi} v_4$, $\forall \chi \in \{\simu, \bisim, \dpsim\}$.

% {\color{black}
% \begin{remark}
% Note that The bijective simulation is less strict than the isomorphism, that is if node $u$ is isomorphism to node $v$, $u$ is $\bjsim$-simulated by $v$ but not vice versa. We prove in \refthm{bijection_and_wl_test} of \refsec{discussion} that the bijective simulation is as strict as the well-known Weisfeiler-Lehman isomorphism test.
% \end{remark}
% }
\stitle{Summary.} In this paper, we consider \emph{all together four simulation variants}: simulation ($\chi = \simu$), degree-preserving simulation (\dpsim), bisimulation (\bisim), and bijective simulation (\bjsim). With a systematic study of existing simulation variants, we have discovered bijective simulation as a new variant. We believe that our work will further inspire more variants. %We present the strictness of the studied simulation variants in \reffig{properties_of_variants}(b), where a ``more-strict'' edge from $\chi_1$ to $\chi_2$-simulation means that a $\chi_1$-simulation is also a $\chi_2$-simulation. %For example, in \reffig{example_graphs}, {\color{black}as $u \mysim^{\bjsim} v_4$, it holds that $u \mysim^{\chi} v_4$, $\forall \chi \in \{\simu, \bisim, \dpsim\}$.}
% Such strictness can also be inferred from \reffig{example_graphs}. Specifically, as $u \mysim^{\bjsim} v_4$, it holds that $u \mysim^{\chi} v_4$, $\forall \chi \in \{\simu, \bisim, \dpsim\}$.

Hereafter, we may omit the $\chi$ in $\chi$-simulation, referring simply to simulation. To avoid ambiguity, we call the simulation relation in \refdef{simulation} as \emph{simple simulation}.

\section{Fractional Simulation} \label{sec:fractionalsimulation}

{\color{black}To quantify the degree to which one node simulates the other node, we now set out the properties fractional $\chi$-simu-lation should satisfy and the framework for its computation.}

\subsection{The Properties of Fractional Simulation}
\label{fractional_simulation_properties}

% \vspace{-0.3em}
\begin{definition} \label{def:fractional_simulation_properties}
    \textsc{(Fractional $\chi$-Simulation)} Given graphs $G_1=(V_1, E_1, \lfunc_1)$ and $G_2=(V_2, E_2, \lfunc_2)$, and two nodes $u \in V_1$ and $v \in V_2$, the fractional $\chi$-simulation of $u$ and $v$ quantifies the degree to which $u$ is approximately $\chi$-simulated by $v$, denoted as $\Sim_\chi(u, v)$. $\Sim_\chi(u, v)$ should satisfy:
    \begin{enumerate}[P1.] \setlength{\itemsep}{0cm}
	\item Range: $0 \leq \Sim_\chi(u, v) \leq 1$;
% 	\item Simulation definiteness: $u \overset{\chi}\sim v$ \textbf{if and only if} $\Sim_\chi(u, v) = 1$;
	{\color{black}\item Simulation definiteness: $u$ is $\chi$-simulated by $v$, i.e., $u {\mysim}^{\chi} v$, \textbf{if and only if} $\Sim_\chi(u, v) = 1$;
% 	\vspace{10em}
    % \item Simulation transitivity: Given a third node $w$\footnote{$w$ can be a node from $G_1$, $G_2$, or a third graph.}, if $\Sim_\chi(u,v) = 1$ and $\Sim_\chi(v,w) = 1$, then $\Sim_\chi(u,w) = 1$;
    % \item Simulation transitivity: Given a third node $w$\footnote{$w$ can be a node from $G_1$, $G_2$, or a third graph.}, if $u \overset{\bjsim}\sim v$, then $\Sim_\chi(u,w) = \Sim_\chi(v,w)$, and $\Sim_\chi(w, u) = \Sim_\chi(w,v)$;
    % \item $\chi$-conditional simulation transitivity:
    % if $\chi$-simulation satisfies IN-mapping, given $(u, v)$ such that $u \overset{\chi}{\sim} v$, $d^+(u) = d^+(v)$ and $d^-(u) = d^-(v)$, then for any third node $w$, $\Sim_\chi(u, w) = \Sim_\chi(v, w)$ and $\Sim_\chi(w, u) = \Sim_\chi(w, v)$;
% 	\item $\chi$-conditional symmetry: if $\chi$-simulation is symmetric, then $\Sim_\chi(u, v)$ is symmetric as well, i.e. $\Sim_\chi(u, v) = \Sim_\chi(v, u)$.}
	\item $\chi$-conditional symmetry: if the $\chi$-simulation has the property of converse invariant (i.e., $u \mysim^{\chi} v$ implies $v \mysim^{\chi} u$), then $\Sim_\chi(u, v)$ should be symmetric, i.e., $\Sim_\chi(u, v) = \Sim_\chi(v, u)$.}
	\end{enumerate} 
	
	A computation scheme $\Sim_\chi$ is \textbf{well-defined} for fractional $\chi$-simulation, if for $\forall (u, v) \in V_1 \times V_2$, $\Sim_\chi(u, v)$ satisfies all three of the above properties.
\end{definition}
% \vspace{-0.5em}

%Property 1 is a common practice that constrains the minimum degree (0) and maximum degree (1) of simulation. 
Property 1 is a common practice. Property 2 bridges the fractional simulation and the corresponding simulation variant. The sufficient condition reflects the fact that $u$ being $\chi$-simulated by $v$ stands for the maximum degree of their simulation, while the necessary condition (only if) makes fractional simulation imply the case of simulation. {\color{black}Property 3 means the variants with converse invariance (i.e., bisimulation and bijective simulation) can be used as similarity measures.}

\subsection{Framework to Compute Fractional Simulation} \label{sec:fsim_framework}
{\color{black}We propose the $\Sim_\chi$ framework to compute the fractional $\chi$-simulation scores for all pairs of nodes across two graphs. The $\Sim_\chi$ is a non-trivial framework because it needs to account for the properties of all simulation variants as well as convergence in general. Note that hereafter, we use $\Sim_\chi$ interchangeably to indicate the framework and a $\chi$-simulation value.

Recall from \refdef{simulation_variants} that a node $u$ is $\chi$-simulated by node $v$ if they have the same label, and their neighbors are $\chi$-simulated accordingly. Thus, we have divided the computation of $\Sim_\chi(u, v)$ into three parts as follows: }

% We introduce the $\Sim_\chi$ framework to compute the $\Sim_\chi$ scores for all pairs of nodes\footnote{In the rest of this paper, we will use $\Sim_\chi$ to interchangeably indicate both a $\chi$-simulation computation framework and a $\chi$-simulation score.}. Recall from \refdef{simulation_variants} that a node $u$ is $\chi$-simulated by node $v$ if they have the same label, and their neighbors are $\chi$-simulated accordingly. Thus, we have divided the computation of $\Sim_\chi(u, v)$ into three parts as follows: 
% \vspace{-0.2em}
% \begin{equation} \label{eq:fractionalsimulation}
% \small
% \begin{split}
%     {\Sim_\chi}(u,v)&=  \underbrace{\outparam \text{ }\Sim_\chi(N^+(u), N^+(v))}_{\textbf{score by out-neighbors}} +\underbrace{\inparam \text{ }\Sim_\chi(N^-(u),N^-(v))}_{\textbf{score by in-neighbors}} \\ &+ \underbrace{\lparam \text{ }\inifunc(u,v)}_{\textbf{score by node label}},
% \end{split}
% % \vspace{-0.7em}
% \end{equation}
% \vspace{-0.7em}
\begin{equation} \label{eq:fractionalsimulation}
\small
\begin{split}
    {\Sim_\chi}(u,v) = \underbrace{\outparam \text{ }\Sim_\chi(N_{G_1}^+(u), N_{G_2}^+(v))}_{\textbf{score by out-neighbors}} +\underbrace{\inparam \text{ }\Sim_\chi(N_{G_1}^-(u),N_{G_2}^-(v))}_{\textbf{score by in-neighbors}} + \underbrace{\lparam \text{ }\inifunc(u,v)}_{\textbf{score by node label}},
\end{split}
% \vspace{-1em}
\end{equation}
where $\!\Sim_\chi(N_{G_1}^+(u), N_{G_2}^+(v))\!$ and $\Sim_\chi(N_{G_1}^-(u), N_{G_2}^-(v))$ denote the scores contributed by the out- and in-neighbors of $u$ and $v$ respectively. $\outparam$ and $\inparam$ are weighting factors that satisfy $0 \leq \outparam < 1$, $0 \leq \inparam < 1$ and $0 < \outparam + \inparam < 1$; and $\inifunc(\cdot)$ is a label function that evaluates the similarity of two nodes' labels. Specifically, if there is no prior knowledge about the labels, $\inifunc(\cdot)$ can be derived by a wide variety of string similarity functions, such as an indicator function, normalized edit distance, Jaro-Winkler similarity, etc. Alternatively, the user could specify/learn the similarities of the label semantics. Since the latter case is beyond the scope of this paper, in the following, we assume no prior knowledge about the labels. 

In \refeq{fractionalsimulation}, we need to compute the $\chi$-simulation score between two node sets $S_1$ and $S_2$ (the respective neighbors of each node pair). To do so, we derive:
% \vspace{-0.2em}
\begin{equation} \label{eq:set_score}
\small
    \Sim_\chi(S_1, S_2) = \frac{\sum_{(x, y) \in \mathcal{M}_{\chi}(S_1, S_2)} {\Sim_\chi}(x,y)}{\normopr_{\chi}(S_1, S_2)},
    % \vspace{-0.8em}
\end{equation}
where $\normopr_{\chi}$ denotes the normalizing operator that returns a positive integer w.r.t. $S_1$ and $S_2$. $\mathcal{M}_{\chi}$ denotes the mapping operator, which returns a set of node pairs defined as:
\begin{align*}
\small
    \mathcal{M}_{\chi}(S_1, S_2; f_\chi) = \{(x, y) \;|\; x \in X, y = f_\chi(x) \in Y\}, 
    % \vspace{-0.8em}
\end{align*}
where $X \subseteq S_1 \cup S_2$ and $Y \subseteq S_1 \cup S_2$. $f_\chi: X \to Y$ is a \emph{function} that is subject to certain constraints regarding the simulation variant $\chi$. These constraints include the domain and codomain of $f_\chi$, and the properties that $f_\chi$ should satisfy (e.g., that $f_\chi$ is an injective function). Note that, for clear presentation, $f_\chi$ is always omitted from the mapping operator. How $\mathcal{M}_{\chi}$ and $\normopr_{\chi}$ are configured to deploy different simulation variants for the framework is demonstrated in \refsec{configurations_of_all_variants}.

\begin{table}
    \centering
    {\color{black}
    \topcaption{Results of whether $u$ is simulated by $v_i$ ($i \in \{1,2,3,4\}$) in \reffig{example_graphs} regarding each simulation variant ($\checkmark$ for yes, $\times$ for no) and the corresponding fractional scores (in bracket)} \label{tab:fracsim_scores}
    % \scalebox{0.9}{
% 	\renewcommand{\arraystretch}{1.1}
	    \begin{tabular}{|c|c|c|c|c|} \hline
        Variants & $(u,v_1)$ & $(u,v_2)$ & $(u,v_3)$ & $(u,v_4)$ \\ \hline
        $\simu$-simulation & $\times$ (0.85) & $\checkmark$ (1.00) & $\checkmark$ (1.00) & $\checkmark$ (1.00)\\ \hline
        \makecell{$\dpsim$-simulation} & $\times$ (0.72) & $\times$ (0.85) & $\checkmark$ (1.00) & $\checkmark$ (1.00)\\ \hline
        $\bisim$-simulation & $\times$ (0.78) & $\checkmark$ (1.00) & $\times$ (0.93) & $\checkmark$ (1.00) \\ \hline
        $\bjsim$-simulation & $\times$ (0.72) & $\times$ (0.81) & $\times$ (0.94) & $\checkmark$ (1.00) \\ \hline
    \end{tabular}} %}
    % \vspace{-1.5em}
\end{table}

% \vspace{-0.3em}
\begin{example}
\color{black}
\reftab{fracsim_scores} shows the $\Sim_{\chi}$ scores for some of the node pairs in \reffig{example_graphs} based on the definition of fractional $\chi$-simulation (\refdef{fractional_simulation_properties}) and the $\Sim_{\chi}$ framework (\refeq{fractionalsimulation}). We can observe that: (1) a pair $(u,v)$ where $u$ is not but very closely simulated by $v$ has a high $\Sim_{\chi}$ score, e.g., $\Sim_\bjsim(u,v_3)$; (2) when $u$ is $\chi$-simulated by $v$, $\Sim_{\chi}(u,v)$ reaches a maximum value of 1, e.g., $\Sim_\bisim(u,v_4)$, which conforms with the well-definiteness of $\Sim_{\chi}$.
\end{example}
% \vspace{-0.6em}

According to \refeq{fractionalsimulation}, the $\!\Sim_\chi\!$ score between two nodes depends on the $\Sim_\chi$ scores of their neighbors. This naturally leads to an iterative computation scheme. This iterative process is detailed in the next section along with how to guarantee its convergence.

\subsection{Iterative Computation} \label{sec:iterative_computation}
%To better explain the iterative computation, 
Consider ${\Sim}_\chi^{k}(u,v)$, which denotes the $\chi$-simulation score of nodes $u$ and $v$ in the $k$-[th] iteration ($k \geq 1$), the mapping operator $\mathcal{M}_\chi^k$ and the normalizing operator $\Omega_\chi^k$  applied in the given iteration. 

\stitle{Initialization.} As all simulation variants require an equivalence of node labels (\refdef{simulation} and \refdef{simulation_variants}), The $\Sim_{\chi}$ score is initially set to ${\Sim}_\chi^{0}(u,v) = \inifunc(u, v)$ by default unless otherwise specified. When using such initialization, $\inifunc(u,v) = 1$ must be further constrained if, and only if, $\lfunc_1(u) = \lfunc_2(v)$, in order to guarantee that $\Sim_{\chi}$ is well-defined (\refdef{fractional_simulation_properties}).

\stitle{Iterative Update.} According to \refeq{fractionalsimulation} and \refeq{set_score}, the simulation score in the $k$-[th] iteration for a node pair $(u,v)$ regarding $\chi$ is updated via the scores of previous iteration as:
\begin{equation} \label{eq:fractional_simulation_computation}
\small
\begin{split}
    {\Sim}_{\chi}^{k}(u, v) &=  \frac{\outparam\sum_{(x, y) \in \mathcal{M}^k_{\chi}(N_{G_1}^+(u), N_{G_2}^+(v))} {\Sim}_{\chi}^{k-1}(x,y)}{\normopr^k_{\chi}(N_{G_1}^+(u), N_{G_2}^+(v))}
    + \frac{\inparam\sum_{(x, y) \in \mathcal{M}^k_{\chi}(N_{G_1}^-(u), N_{G_2}^-(v))} {\Sim}_{\chi}^{k-1}(x,y)}{\normopr^k_{\chi}(N_{G_1}^-(u), N_{G_2}^-(v))} \\
    &+ \lparam \inifunc(u,v)
\end{split}
\end{equation}
% \vspace{-1.5em}

%Refer to \refdef{fractional_simulation_properties}, the initialization of fractional simulation should satisfies that ${\Sim}^{0}(u,v) = 1$ if and only if $\lfunc(u)=\lfunc(v)$ so as to ensure the properties of fractional simulation. 
\stitle{Convergence.} Below we show what conditions the mapping and normalizing operators should satisfy to guarantee \refeq{fractional_simulation_computation} converges. {\color{black}Specifically, the computation is considered to converge if $|\Sim_{\chi}^{k+1}(u,v) - \Sim_{\chi}^{k}(u,v)| < \epsilon$ for $\forall (u, v) \in V_1 \times V_2$, in which $\epsilon$ is a small positive value. Note that the simulation subscript $\chi$ is omitted in the following theorem as it applies to all simulation variants.}

% \begin{theorem} \label{thm:convergence}
%     When $\mathcal{M}$ at iteration $k$ gives a matching to maximize the summation of $\Sim^{k-1}(x,y)$ values with $(x, y) \in \mathcal{M}$, while the matching size $|\mathcal{M}|$ and normalizing value $\normopr$ for a node pair $(u,v)$ do not vary with iteration $k$, then the computation of fractional simulation always converges if $|\mathcal{M}| \leq \normopr$ and $w^+ + w^- < 1$ hold.
% \end{theorem}

% \vspace{-0.5em}
\begin{theorem} \label{thm:convergence}
   The computation in \refeq{fractional_simulation_computation} is guaranteed to converge if in every iteration $k$, the following conditions are satisfied for any two node sets $S_1$ and $S_2$ in the mapping and normalizing operators: 
   \begin{enumerate}[(C1)] \setlength{\itemsep}{0cm}
       \item $|\mathcal{M}^{k+1}(S_1, S_2)| = |\mathcal{M}^k(S_1, S_2)|$, and $\Omega^{k+1}(S_1, S_2) = \Omega^{k}(S_1, S_2)$. %This means that the mapping sizes between any two node sets do not vary between iterations.
       \item $|\mathcal{M}^k(S_1, S_2)| \leq \Omega^k(S_1, S_2)$.
       \item Subject to the function $f$, $\mathcal{M}^k(S_1, S_2)$ returns node pairs such that
       \begin{align*}
       \sum_{(x, y) \in \mathcal{M}^k(S_1, S_2)} {\Sim}^{k-1}(x,y) \text{ is maximized.}
       \end{align*}
   \end{enumerate}
\end{theorem}
% \vspace{-0.5em}

\begin{proof}
    Let $\delta^k(u,v) = |\Sim^k(u,v) - \Sim^{k-1}(u,v)|$ and $\Delta^k = \max_{(u,v)}\delta^k(u,v)$. To prove this theorem, we must show that $\Delta^k$ decreases monotonically, i.e., $\Delta^{k+1} < \Delta^{k}$.
    
Let $W^k(S_1, S_2) = \sum_{(x,y) \in \mathcal{M}^k(S_1,S_2)} \Sim^{k-1}(x,y)$. As the size of the mapping operator and the value of normalizing operator between $S_1$ and $S_2$ do not vary with $k$ (C1), we simply write $|\mathcal{M}(S_1, S_2)|$ and $\Omega(S_1, S_2)$ by dropping the superscript. Then, we have
    \begin{equation*} \label{eq:two_eq1}
        \small
        \begin{split}
            % W^{k+1}(S_1, S_2) &= \sum_{(x,y)\in \mathcal{M}^{k+1}(S_1, S_2)} \Sim^{k}(x,y) \\
           W^{k+1}(S_1, S_2) & \geq \sum_{(x,y)\in \mathcal{M}^{k}(S_1, S_2)} \Sim^{k}(x,y) \text{ (by C3)} \\
            % & \geq \sum_{(x,y)\in \mathcal{M}^{k}(S_1, S_2)} (\Sim^{k-1}(x,y) - \Delta^k) \\
            & \geq W^{k}(S_1, S_2) - |\mathcal{M}(S_1, S_2)|\Delta^k \text{ (by C1) }
        \end{split}
    \end{equation*}
    % \vspace{-0.2em}
    Similarly, {\small $ W^{k}(S_1, S_2) \geq W^{k+1}(S_1, S_2) - |\mathcal{M}(S_1, S_2)|\Delta^k$} can be derived, and we immediately have, 
    % \begin{equation*} 
    %     W^{k}(S_1, S_2) \geq W^{k+1}(S_1, S_2) - |\mathcal{M}(S_1, S_2)|\Delta^k \label{eq:two_eq2}
    % \end{equation*}
    % \vspace{-0.2em}
    %where (\textbf{C3}) and (\textbf{C1}) in \refeq{two_eq1} indicate that the derivation bases on condition 3 and condition 1. We omit the derivation of \refeq{two_eq2} as it is similar to that of \refeq{two_eq1}. 
    % We immediately have,
    % \vspace{-0.5em}
    \begin{equation} \small \label{eq:two_eq_conclusion}
    \small
        |W^{k+1}(S_1, S_2)-W^{k}(S_1, S_2)| \leq \normopr(N_1,N_2)\Delta^k \text{ (by C2) }
    \end{equation}
    % \vspace{-1.5em}
    
    %We further denote $W^k_{+} = W(\mathcal{M}^k(N^+(u),N^+(v)))$ and $W^k_{-} = W(\mathcal{M}^k(N^-(u),N^-(v))$ considering $\Sim^{k-1}(\cdot)$, respectively. 
    Then,
    % \vspace{-0.2em}
    \begin{equation} \label{eq:convergence_proof}
    \small
    \begin{split}
        \delta^{k+1}(u,v) &\leq (\outparam + \inparam)\Delta^k \text{ (by \refeq{two_eq_conclusion}) }\\
                          & < \Delta^k  \text{  (by $\outparam + \inparam < 1$) }
    \end{split}
    \end{equation}
    % \vspace{-0.8em}
    Thus, $\Delta^{k+1}< \Delta^k$, and the computation converges.
\end{proof}

% \begin{corollary} \label{coro:converge_speed}
% The computation in \refeq{fractional_simulation_computation} converges within $\lceil\log_{(\outparam + \inparam)}\epsilon\rceil$ iterations, in which the computation is considered to converge if $|\Sim^k(u,v) - \Sim^{k-1}(u,v)| < \epsilon$ for $\forall (u, v) \in V_1 \times V_2$.
% \end{corollary}
% \begin{corollary} \label{coro:converge_speed}
% The computation in \refeq{fractional_simulation_computation} converges within $\lceil\log_{(\outparam + \inparam)}\epsilon\rceil$ iterations, in which the computation is considered to converge if $|\Sim^k(u,v) - \Sim^{k-1}(u,v)| < \epsilon$ for $\forall (u, v) \in V_1 \times V_2$.
% \end{corollary}
\begin{corollary} \label{coro:converge_speed}
The computation in \refeq{fractional_simulation_computation} converges within $\lceil\log_{(\outparam + \inparam)}\epsilon\rceil$ iterations.
\end{corollary}

\begin{proof}
According to \refeq{convergence_proof}, we have $\Delta^{k+1}\leq (\outparam + \inparam)\Delta^k$. As $\Delta^{0}$ cannot exceed 1, the theorem holds.
\end{proof}

{
% {\color{red} We discuss the needs of the three conditions in \refthm{convergence}. C1 requires that $\normopr_\chi(S_1, S_2))$ and $|\mappingopr_\chi(S_1, S_2)|$ remain fixed among iterations, so that the fractional score in two consecutive iterations can be compared while proving its convergence. C2 constrains the number of node pairs in $\mappingopr_\chi(S_1, S_2))$ to be be less than the value of $\normopr_\chi(S_1, S_2)$ in order to guarantee the range property in \refdef{fractional_simulation_properties}. C3 requires that $\mappingopr_\chi$ must include the pairs of neighbors that maximize the sum of their $\Sim_{\chi}$ scores in the previous iteration, which has the intuition of maximizing the contributions of the neighbors and is essential to satisfy the property of simulation definiteness in \refdef{fractional_simulation_properties}. Such a mapping operator is accordingly called a \emph{maximum mapping operator}, and will be applied by default in the following. } 
{\color{black} We discuss the needs of the three conditions in \refthm{convergence}. Given node sets $S_1$ and $S_2$, C1 requires that the value of the normalizing operator $\normopr_\chi(S_1, S_2)$ and the number of node pairs in $\mappingopr_\chi(S_1, S_2)$ (i.e., $|\mappingopr_\chi(S_1, S_2)|$) remain unchanged throughout the iterations. C2 requires that $|\mappingopr_\chi(S_1, S_2))|$ should be less than $\normopr_\chi(S_1, S_2)$ to guarantee the range property in \refdef{fractional_simulation_properties}. C3 requires that $\mappingopr_\chi$ should include the pairs of neighbors that maximize the sum of their $\Sim_{\chi}$ scores in previous iteration. Intuitively, C3 maximizes the contributions of neighbors and is essential to satisfy simulation definiteness (property 2 in \refdef{fractional_simulation_properties}). Such a mapping operator is accordingly called a \emph{maximum mapping operator}, and will be applied by default in the following. } 
%According to \refthm{convergence}, the value of the normalizing operator $\normopr_\chi$ and the number of node pairs returned by $\mappingopr_\chi$ regarding any two node sets should not change with the iterations (C1). Specifically, the number of node pairs in $\mappingopr_\chi$ should be less than the value of $\normopr_\chi$ (C2) so as to satisfy the range property in \refdef{fractional_simulation_properties}. In addition, the mapping operator $\mappingopr_\chi$ must include the pairs that maximize the sum of the $\Sim_{\chi}$ scores in the previous iteration (C3). Such a mapping operator is accordingly called a \emph{maximum mapping operator}, which is essential to satisfy the property of simulation definiteness in \refdef{fractional_simulation_properties} and will be applied by default in the following.
}

\subsection{Computation Algorithm} \label{sec:computing_algorithm}
\refalg{fractional_simulation_computation} outlines the process for computing $\Sim_{\chi}$. The computation begins by initializing a hash map \mapcurr to maintain the initial $\Sim_{\chi}$ scores of candidate node pairs (Line~\ref{initialization}). Note that not all $|V_1|\times|V_2|$ node pairs need to be maintained, which is explained in \refsec{configurations_of_all_variants}.
Then, the scores of the node pairs in \mapcurr are updated iteratively until convergence (Lines~\ref{update_start}-\ref{update_end}). In Line~\ref{return}, the hash map is returned with the results .

% \refalg{fractional_simulation_computation} outlines the process for computing $\Sim_{\chi}$. We first initialize a hash map \mapcurr to maintain the initial $\Sim_{\chi}$ scores of candidate node pairs (line~\ref{initialization}). Note that we do not need to maintain all $|V_1|\times|V_2|$ node pairs, as will be illustrated in \refsec{configurations_of_all_variants}.
% Then, we iteratively update the scores of node pairs in \mapcurr (line~\ref{update_start} - line~\ref{update_end}) until converged. Finally, we return the hash map with the results (line~\ref{return}).

\stitle{Parallelization.} The most time-consuming part of \refalg{fractional_simulation_computation} is running the iterative update in Lines~\ref{update_start} through \ref{update_end}. This motivated us to consider accelerating the computation with para-llelization by using multiple threads to compute different node pairs simultaneously. In this implementation, the simulation scores of the previous iteration are maintained in $H_{\kw{p}}$, which means computing the node pairs in Lines~\ref{update-out-neighbor} and \ref{update-in-neighbor} is independent of each other, and can be completed in parallel without any conflicts. 
We simply round-robin the node pairs in $H_{\kw{c}}$ to distribute the load to all available threads, which achieves satisfactory scalability in the experiment (\reffig{time_vary_thread}). 

% \vspace{-1em}
\begin{algorithm}
	\SetAlgoVlined
	\SetFuncSty{textsf}
	\SetArgSty{textsf}
	\small
	\caption{The algorithm of computing $\Sim_{\chi}$} \label{alg:fractional_simulation_computation}
	\small
	\Input {Graphs $G_1 = (V_1, E_1, \lfunc_1)$, $G_2 = (V_2, E_2, \lfunc_2)$, weighting factors ${w}^+, {w}^-$.}
% 	\Input {Graphs $G_1 = (V_1, E_1, \lfunc_1)$, $G_2 = (V_2, E_2, \lfunc_2)$, weighting factors ${w}^+, {w}^-, {w}^*$, and threshold $\theta$, approximating factor $\alpha$.}
	\Output {$\Sim_{\chi}$ Scores.}
    \State{$H_{\kw{c}} \leftarrow$ \textbf{Initializing($G_1$, $G_2$, ${w}^+$, ${w}^-$)};} \label{initialization} \\
    \State{$H_{\kw{p}} \leftarrow H_{\kw{c}}$;} \\
    \While{not converged}{ \label{update_start}
    	\ForEach{$(u,v) \in H_{\kw{c}}$} { 
	        \State{$H_{\kw{c}}[(u,v)] \leftarrow (1-w^+-w^-)\inifunc(u, v)$;} \label{update-label-sim} \\
		    \ForEach{$(x,y) \in \mathcal{M}_\chi (N_{G_1}^+(u), N_{G_2}^+(v))$} {
		        \State{$H_{\kw{c}}[(u,v)] \leftarrow H_{\kw{c}}[(u,v)] + \frac{w^+H_{\kw{p}}[(x,y)]}{\normopr_{\chi}(N_{G_1}^+(u), N_{G_2}^+(v))}$;} \label{update-out-neighbor}\\ 
		    }
		    \ForEach{$(x',y') \in \mathcal{M}_\chi (N_{G_1}^-(u), N_{G_2}^-(v))$} {
		        \State{$H_{\kw{c}}[(u,v)] \leftarrow H_{\kw{c}}[(u,v)] + \frac{w^-H_{\kw{p}}[(x',y')]}{\normopr_{\chi}(N_{G_1}^-(u), N_{G_2}^-(v))}$;} \label{update-in-neighbor}\\
		    }
	    }
	    \State{$H_{\kw{p}} \leftarrow H_{\kw{c}}$;} \label{update_end}\\
	   % \State{$k \leftarrow k + 1$;} \\
    }
	\State{\Return $H_{\kw{c}}$.} \label{return}
\end{algorithm}
% \vspace{-1.8em}
%Note that at each iteration, the $\Sim_{\chi}$ value of a node pair $(u,v)$ relies on their label similarity (Line~\ref{update-label-sim}) and the $\Sim_{\chi}$ scores of their respective neighbors at previous iteration (Line~\ref{update-out-neighbor} and Line~\ref{update-in-neighbor}). Fortunately, the label similarity can be computed based on the graphs, and the scores of neighbors can be fetched from $H_{\kw{p}}$, which makes the computations of node pairs independent of each other, and can be paralleled without any conflicts. We simply round-robin the node pairs in the hash map to distribute the load to all available threads, which achieves a satisfactory scalability in the experiment (\reffig{time_vary_thread}). 

\stitle{Upper-Bound Updating.} According to the range property (\refdef{fractional_simulation_properties}) and the computation in \refeq{fractional_simulation_computation}, there exists an upper-bound on the $\Sim_{\chi}$ value of each node pair, which is computed via:
\begin{equation} \label{eq:upper_bound_label_constrained_mapping}
% \vspace{-0.5em}
\small
\begin{split}
    {\Sim}_\chi(u,v) & \leq \overline{\Sim}_\chi(u,v) \\
    & = \lambda^+(u, v) + \lambda^-(u, v) + (1-w^+-w^-) \mathcal{L}(u,v),
    %&= \frac{w^+|\mathcal{M}_\chi(N^+(u), N^+(v))|}{\normopr^+_{\chi}(N^+(u),N^+(v))} \\
    %& + \frac{w^-|\mathcal{M}_\chi(N^-(u),N^-(v))|}{\normopr^-_{\chi}(N^-(u),N^-(v))} \\ 
    %&+ (1-w^+-w^-) \mathcal{L}(u,v)
\end{split}
% \vspace{-0.5em}
\end{equation}
where $\lambda^\kw{s} = \frac{w^\kw{s}|\mathcal{M}_\chi(N_{G_1}^\kw{s}(u), N_{G_2}^\kw{s}(v))|}{\normopr^\kw{s}_{\chi}(N_{G_1}^\kw{s}(u),N_{G_2}^\kw{s}(v))}$, for $\kw{s} \in \{+, -\}$. Accordingly, if the upper bound of a certain node pair $(u, v)$ is relatively small (smaller than a given threshold $\beta$), it is expected to make a limited contribution to the scores of others. Thus, we can skip computing (and maintaining) $\Sim_\chi(u, v)$, and use an approximated value $\alpha\overline{\Sim}_\chi(u,v)$ ($0 < \alpha < 1$ is a given small constant) instead when needed. The implementation of upper-bound updating based on \refalg{fractional_simulation_computation} is as follows: (1) in Line~\ref{initialization}, $H_{c}$ only maintains the node pairs that are guaranteed to be larger than $\beta$; (2) in Lines~\ref{update-out-neighbor} and \ref{update-in-neighbor}, if $(x, y)$ (or $(x',y')$) is not in $H_{p}$, use $\alpha\overline{\Sim}_\chi(x,y)$ (or $\alpha\overline{\Sim}_\chi(x',y')$) instead.

%Note that the upper bound value $\overline{\Sim}_\chi(u,v)$ in \refeq{upper_bound_label_constrained_mapping} is negatively correlated with $\theta$ in the label constraint (\refrem{label_constraint_mapping}), that is $\overline{\Sim}_\chi(u,v)$ gets smaller when $\theta$ gets larger. A larger $\theta$ will potentially make both $|\mathcal{M}_\chi(N^+(u), N^+(v))|$ and $|\mathcal{M}_\chi(N^+(u), N^+(v))|$ smaller due to a smaller number of matchable nodes. Through prior experiment, we find out that such upper-bound optimization does not work notably when $\theta < 1.0$, thus we only use it in the ``Vary thread'' experiment (when $\theta = 1.0$) in \refsec{efficiency}. As we did for tuning $\theta$ in \refsec{sensitivity_analysis} and \refsec{efficiency}, we have experimented to select $\alpha = 1$ and $\beta = 0.5$ for the upper-bound optimization, which achieves the best trade-off between sensitivity (correlation score $>0.9$ compared with non-optimized) and performance.
% \input{sec4.3_speeding_up_computation.tex}
% \vspace{0.5em}
\section{Configure Framework to Quantify Different Simulation Variants} \label{sec:configurations_of_all_variants}
In this section, we show how to configure the mapping and normalizing operators in \refeq{set_score}, such that the computation of $\Sim_\chi$ converges, and $\Sim_\chi$ remains well-defined (\refdef{fractional_simulation_properties}) for all simulation variants. %We first detail the configurations for $\simu$-simulation, then summarize the configurations for all the other variants. We finally discuss the relations of $\Sim_\chi$ to several notable works.

\subsection{Configurations of Simple Simulation} \label{sec:simple_simulation}
\stitle{Fractional \simu-simulation.} %Refer to \refdef{simulation}, simple simulation (\simu-simulation)  only considers labels of the nodes and their out-neighbors, thus we set $\inparam = 0$ in \refeq{fractionalsimulation}. 
Given two node sets $S_1$ and $S_2$, $\mathcal{M}_\simu$ and $\normopr_\simu$ are the operators for implementing fractional $\simu$-simulation according to \refdef{fractional_simulation_properties} as follows: 
% \vspace{-1em}
\begin{equation} \label{eq:match_simulation}
\small
    \mathcal{M}_{\simu}(S_1, S_2)= \{(x, y)| \forall x \in S_1, y = f_\simu(x) \in S_2\},
    % \vspace{-0.2em}
\end{equation}
where $f_\simu: S_1 \to S_2$ is a function subject to the label constraint $\inifunc(x, f_\simu(x)) \geq \theta$, and
\begin{equation} \label{eq:denom_simulation}
\small
    \normopr_{\simu}(S_1, S_2)= |S_1|.
    % \vspace{-0.2em}
\end{equation}

% \vspace{-0.5em}
\begin{remark} \label{rem:label_constraint_mapping} \textbf{Label-constrained Mapping.}
Analogous to the initialization of $\Sim_\chi$ (\refsec{iterative_computation}), a label constraint is added when mapping neighbors. $\theta$ is a constant given by the user to control the strictness of the mapping. When $\theta = 0$, the nodes can be mapped arbitrarily. When $\theta = 1$, only nodes of the same label can be mapped. It is obvious that a larger $\theta$ leads to faster computation. As a practical guide to setting $\theta$, \refsec{sensitivity_analysis} includes a sensitivity analysis of $\theta$ and \refsec{efficiency} provides an efficiency analysis. In the following, the label constraint is applied in the mapping operator by default and is thus omitted from the notations for clear presentation. 
%we empirically study the influence of $\theta$ regarding both sensitivity (\refsec{sensitivity_analysis}) and efficiency (\refsec{efficiency}). 
\end{remark}
% \vspace{-0.5em}

\comment{
\begin{figure}
    \centering
    \includegraphics[width=0.78\linewidth]{figures/function_of_label_constrained_mapping_aline_2.pdf} 
    \topcaption{When computing $\Sim_\simu(u,v_3)$, the searching space of the green pentagon node (a neighbor of $u$ outlined with red) regarding the $v_3$'s neighbor set.} \label{fig:function_of_label_constrained_mapping}
\end{figure}

{\color{red}According to \refdef{simulation}, two simulated nodes must have the same node label. Therefore, we add the constraint $\inifunc(x, f_\simu(x)) \geq \theta$ in the mapping operator in \refeq{match_simulation}. We use the following example to illustrate the affect of $\inifunc(x, f_\simu(x)) \geq \theta$.}

\begin{example} \label{ex:label_constrained_mapping}
In \reffig{function_of_label_constrained_mapping}, when calculating $\Sim_\simu(u, v_3)$, we need to calculate the mapped node for each node in $u's$ neighbors. Take the green pentagon neighbor (outlined with red) of $u$ for an example, if $\theta = 0$, the searching space is all the neighbors of $v_3$ (light shading). When $\theta$ is set to 0.5, the search space will shrink to the darker shading area, which contains only half of the neighbors.
\end{example}
}

% \vspace{-1em}
%From \refeq{match_simulation}, the following constraints $\Phi_\simu$ are inferred: (1) the codomain and domain of $f$ are $N_1$ and $N_2$, respectively; (2) $f$ is a total function. 

{\stitle{Convergence.} \color{black}It is obvious that $|M_\simu(S_1, S_2)| \leq |S_1| = \Omega_\simu(S_1, S_2)$, which satisfies C1 and C2 in \refthm{convergence}. As mentioned earlier, C3 is applied by default. Therefore, the convergence of $\Sim_\simu$ is guaranteed.}

\stitle{Well-Definiteness.} 
\refthm{frac_S_simulation} shows that $\Sim_\simu$ is well-defined for fractional $\simu$-simulation according to \refdef{fractional_simulation_properties}. 

\begin{theorem} \label{thm:frac_S_simulation}
    $\Sim_{\simu}$ is well-defined for fractional \simu-simulation.
\end{theorem}
% \vspace{-0.8em}
\begin{proof}
We prove that $\Sim_{\simu}$ satisfies all the properties in \refdef{fractional_simulation_properties}. P1 is easy to verify. It is unnecessary to verify P3 as \simu-simulation {\color{black}has no converse invariant}. Below we prove P2. For brevity, we only consider out-neighbors in the proof, and the proof with in-neighbors is similar.

We first prove that if $\Sim_{\simu}(u,v)=1$, $u \mysim^{\simu} v$. Based on \refeq{fractionalsimulation}, we have $\lfunc_1(u) = \lfunc_2(v)$, and we add $(u,v)$ into $R$ (initialized as $\emptyset$). $\forall (x, y) \in \mathcal{M}_{\simu}$, $\Sim_{\simu}(x,y)=1$ and $\lfunc_1(x) = \lfunc_2(y)$. Then, we add these nodes pairs into $R$, i.e. $R = R \bigcup \mathcal{M}_{\simu}$. %New node pairs can be added recursively as every node pair in $R$ has score $1$. 
New node pairs can be added recursively. The process will terminate as $|R|$ increases and cannot exceed $|V_1|\times|V_2|$. One can verify that $R$ is a simulation. 

We next prove that, for $\forall (u,v) \in R$, $\Sim_{\simu}^k(u,v) = 1$ for any $k$, where $R$ is a simulation relation. Based on \refdef{simulation}, we define the mapping $\mathcal{M}$ between $N_{G_1}^+(u)$ and $N_{G_2}^+(v)$ as $\mathcal{M} = \{(u',v') | \forall u' \in {N}_{G_1}^+(u), (u',v') \in R\}$. The case of $k=0$ is easy to verify. Assume the theorem holds at $k-1$. For a node pair $(u,v) \in R$, any $(u',v') \in \mathcal{M}$ defined above satisfies $\Sim_{\simu}^{k-1}(u',v') = 1$. Clearly, $\mathcal{M}$ is a mapping operator defined in \refeq{match_simulation}. Thus, $\Sim_{\simu}^{k}(u,v) = 1$.
\end{proof}
% \vspace{-1em}
\stitle{Computation.} The mapping operator $\mappingopr_{\simu}$ (\refeq{match_simulation}) constrains that $\forall (x, y) \in \mappingopr_{\simu}$, $\inifunc(x,y) \geq \theta$. As a result, the nodes pairs with $\inifunc(\cdot) < \theta$ will never contribute to the computation. Thus, only the node pairs with $\inifunc(\cdot) \geq \theta$ need to be maintained(Line~\ref{initialization} in \refalg{fractional_simulation_computation}), which helps to reduce both the time and space complexity. %The pseudo-code of initializing can be found in \refalg{fractional_simulation_initial} in \refsec{initializating}.

% \stitle{Cost Analysis.} The time cost to compute $\Sim_{\simu}(u,v)$ in a single iteration is dominated by the mapping operator. According to \refeq{match_simulation}, for $\forall x \in S_1$, we simply search $y \in S_2$ that has the maximum fractional score, which consumes $O(|S_1||S_2|)$ time. Therefore, the time cost of calculating the $\Sim_{\simu}$ scores is $O(k|H| D^+_{G_1} d^+_{G_2})$ with $k$ as the number of iterations and $|H|$ is the size of the candidate node pairs ($|H| \leq |V_1|\times|V_2|$). The space cost is $O({|H|})$ as we need to maintain a map to store the $\Sim_{\chi}$ scores of the previous iteration.

% \stitle{Cost Analysis.} The time cost to compute $\Sim_{\simu}(u,v)$ is dominated by the mapping operator. According to \refeq{match_simulation}, for $\forall x \in S_1$, we simply search $y \in S_2$ to maximize $\Sim^{k-1}_\chi(x, y)$, which consumes $O(|S_1||S_2|)$ time. Therefore, the time cost of computing the $\Sim_{\simu}$ scores for $\forall (u,v) \in H$ is $\sum_{(u, v) \in H} (d_{G_1}^+(u) \cdot d_{G_2}^+(v) + d_{G_1}^-(u) \cdot d_{G_2}^-(v)) \leq |H| (D^+_{G_1} D^+_{G_2} + D^-_{G_1} D^-_{G_2})$, where $D^+_{G_1}$, $D^+_{G_2}$ are the maximum out-degree, and $D^-_{G_1}$, $D^-_{G_2}$ are the maximum in-degree of $G_1$ and $G_2$. As a result, the time complexity of computing $\Sim_{\simu}$ is $O(k|H| (D^+_{G_1} D^+_{G_2} + D^-_{G_1} D^-_{G_2})$ with $|H|\leq |V_1| \times |V_2|$ and $k$ as the number of iterations. The space cost is $O({|H|})$ as we need to maintain a map to store the $\Sim_{\simu}$ scores of the previous iteration.
\stitle{Cost Analysis.} The time cost to compute $\Sim_{\simu}(u,v)$ is dominated by the mapping operator. According to \refeq{match_simulation}, for $\forall x \in S_1$, we simply search $y \in S_2$ to maximize $\Sim^{k-1}_\chi(x, y)$, which takes $O(|S_1||S_2|)$ time. Therefore, the time complexity of computing $\Sim_{\simu}$ is $O(k|H|(D^+_{G_1} D^+_{G_2} + D^-_{G_1} D^-_{G_2})$ with $|H|\leq |V_1| \times |V_2|$ and $k$ as the number of iterations. The space cost is $O({|H|})$, as the map of $\Sim_{\simu}$ scores for the previous iteration needs to be stored.
%Therefore, the time cost of computing the $\Sim_{\simu}$ scores for node pairs in $H$ is $\sum_{(u, v) \in H} (d_{G_1}^+(u) \cdot d_{G_2}^+(v) + d_{G_1}^-(u) \cdot d_{G_2}^-(v)) \leq |H| (D^+_{G_1} D^+_{G_2} + D^-_{G_1} D^-_{G_2})$. As a result, the time complexity of computing $\Sim_{\simu}$ is $O(k|H| (D^+_{G_1} D^+_{G_2} + D^-_{G_1} D^-_{G_2})$ with $|H|\leq |V_1| \times |V_2|$ and $k$ as the number of iterations. The space cost is $O({|H|})$ as we need to maintain a map to store the $\Sim_{\simu}$ scores of the previous iteration.

\subsection{Configurations for All Simulation Variants}
%Based on the configurations for $\Sim_{\simu}$, 
% We now configure the mapping and normalizing operators of other simulation variants (\refsec{preliminary}). Note that this is non-trivial, as the configurations need to meet with the properties of the variant, and meanwhile preserve the convergence (\refthm{convergence}) and well-definiteness (\refdef{fractional_simulation_properties}) of the computation.

% , e.g. {\color{black}bisimulation with converse invariant} and degree-preserving simulation with IN-mapping, and meanwhile preserve the convergence (\refthm{convergence}) and well-definiteness (\refdef{fractional_simulation_properties}) of the computation. 

% \reftab{configurations} summarizes the configurations of each simulation variant. 
\reftab{configurations} summarizes the configurations of each simulation variant. With the given configurations, $\Sim_{\chi}$ is well-defined for $\forall \chi \in \{\simu,\dpsim,\bisim,\bjsim\}$. We only provide the proofs of the well-definiteness for $\Sim_\simu$ (asymmetric, \refthm{frac_S_simulation}) and $\Sim_\bjsim$ (symmetric, \refthm{bijectivesimulation}). The proofs for the other variants are similar and thus are omitted due to space limitations.

\begin{table}[h]
    \small
    \centering
    \topcaption{Configurations of the fractional $\chi$-simulation framework ($\Sim_{\chi}$) to quantify the studied simulation variants.}\label{tab:configurations}
    \scalebox{0.86}{
    \renewcommand{\arraystretch}{1.8}
	\begin{tabular}{|c|c|c|c|} \hline
    \textbf{$\Sim_{\chi}$} & $\normopr_{\chi}(S_1,S_2)$ & $\mappingopr_{\chi}(S_1,S_2)$ & \textbf{Function Constraints (label constraint implied)} \\ \hline
    $\Sim_{\simu}$ & {$|S_1|$} & {$\{(x, y)| \forall x \in S_1, y = f_\simu(x) \in S_2\}$} & {\makecell{$f_\simu(x): S_1 \to S_2$%, and $\forall x \in S_1$, \inifunc(x, f_\simu(x)) \geq \theta$
    }} \\ \hline
    $\Sim_{\dpsim}$ & {$|S_1|$} & {\makecell{$\{(x, y)| \forall x \in S'_1, y = f_\dpsim(x) \in S_2\}$, \\ where $S'_1 \subseteq S_1$ with $|S'_1| = \min(|S_1|, |S_2|)$}} & {\makecell{$f_\dpsim: S'_1 \to S_2$ is an injective function%, and $\forall x \in S'_1, \inifunc(x, f_\dpsim(x)) \geq \theta$
    }} \\ \hline
    $\Sim_{\bisim}$ & {$|S_1| + |S_2|$} & {\makecell{$\{(x, y)| \forall x \in S, y = f_\bisim(x) \in S\}$, where $S = S_1 \cup S_2$}} & {\makecell{$
    f_\bisim(x) \in \begin{cases}
    S_2, \text{ if } x \in S_1, \\
    S_1, \text{ if } x \in S_2
    \end{cases}$% and $\forall x \in S_1 \cup S_2, \inifunc(x, f_\bisim(x)) \geq \theta$
    }} \\ \hline
    $\Sim_{\bjsim}$ & {$\sqrt{|S_1|\times|S_2|}$} & {\makecell{$\{(x, y)| \forall x \in S_m, y = f_\bjsim(x) \in S_M\}$, in which if $|S_1| \leq |S_2|$, \\ $S_{m} = S_1$ and $S_{M} = S_2$; otherwise, $S_{m} = S_2$ and $S_{M} = S_1$}} & {\makecell{$f_\bjsim(x): S_m \to S_M$ is an injective function%, and $\forall x \in S_m, \inifunc(x, f_\bjsim(x)) \geq \theta$
    }} \\ \hline
    %$\strongdsim$-simulation & \multicolumn{3}{c|}{same as \simu-simulation} \\ \hline
    \end{tabular}}
\end{table}

% \vspace{-0.5em}
\begin{theorem} \label{thm:bijectivesimulation}
$\!\Sim_{\bjsim}\!$ is well-defined for fractional {\bjsim}-simulation. 
\end{theorem}
% \vspace{-1em}
\begin{proof}
% Proofs of P1 and P2 are similar to those of $\Sim_{\simu}$ (\refthm{frac_S_simulation}). We then prove that $\Sim_{\bjsim}$ satisfies P3, as $\bjsim$-simulation {\color{black}has the property of converse invariant.} 

% We prove $\Sim_{\bjsim}^k(u,v) = \Sim_{\bjsim}^k(v,u)$ for any $k$ values. The case of $k=0$ is true as the initialization function are symmetric. Suppose $\Sim_{\bjsim}^{k-1}(u,v)$ is symmetric, the symmetry of $\Sim_{\bjsim}^{k}(u,v)$ can be immediately proved as $\normopr_{\bjsim}$ is symmetric as well. As a result, we have $\Sim_{\bjsim}(u,v) = \Sim_{\bjsim}(v,u)$.
Proofs of P1 and P2 are similar to those of $\Sim_{\simu}$ (\refthm{frac_S_simulation}). Proof of P3, i.e., $\Sim_{\bjsim}(u,v) = \Sim_{\bjsim}(v,u)$, is given by mathematical induction.  

As the initialization function is symmetric, we have $\Sim_{\bjsim}^0(u,v) =\Sim_{\bjsim}^0(v,u)$. Suppose that $\Sim_{\bjsim}^{k-1}(u,v)$ is symmetric, the symmetry of $\Sim_{\bjsim}^{k}(u,v)$ can be immediately proved as $\normopr_{\bjsim}$ is symmetric as well. As a result, we have $\Sim_{\bjsim}(u,v) = \Sim_{\bjsim}(v,u)$.
\end{proof}

\stitle{Cost Analysis.} According to \refalg{fractional_simulation_computation}, the space complexity is $O(|H|)$, where $|H|\leq |V_1| \times |V_2|$. The time complexity of computing $\Sim_{\bisim}$ is the same as $\Sim_{\simu}$. For computing $\Sim_{\dpsim}$ and $\Sim_{\bjsim}$, the Hungarian algorithm needs to be applied to implement the mapping operators due to the presence of injection. Using a popular greedy approximate of Hungarian \cite{avis1983survey}, $\mappingopr_\dpsim(S_1, S_2)$ and $\mappingopr_\bjsim(S_1, S_2)$ can be solved in a time complexity of $O(|S_1||S_2| \log(|S_1||S_2|))$. {\color{black}As a whole, the time cost of computing $\Sim_{\dpsim}$ and $\Sim_{\bjsim}$ is $O(k{|H|}({D^+_{G_1}}{D^+_{G_2}}\cdot\log {D^+_{G_1}}{D^+_{G_2}}+ {D^-_{G_1}}{D^-_{G_2}}\cdot\log {D^-_{G_1}}{D^-_{G_2}}))$.}
\subsection{Discussions} \label{sec:discussion}
%Speaking of an iterative framework that computes a score within $[0,1]$ for all node pairs, it naturally comes to mind the widely-used node similarity measures, i.e., \kw{SimRank} and \kw{RoleSim}. We then discuss the relations of $\Sim_\chi$ to node similarity measures. {\color{black}In addition, we discuss the relations of $\Sim_{\chi}$ to $k$-bisimulation (a variant of bisimulation) and graph isomorphism, respectively.}
%we find out a connection between $\Sim_{\bjsim}$, or more specifically the bijective simulation, and the Weisfeiler-Lehman isomorphism test.
{\color{black}$\Sim_\chi$ is closely related to several well-known concepts, including node similarity measures (i.e., \kw{SimRank} and \kw{RoleSim}), $k$-bisimulation (a variant of bisimulation) and graph isomorphism. In this subsection, we discuss their relations to $\Sim_\chi$.}

\stitle{Relations to Similarity Measures.} The $\Sim_{\chi}$ framework (\refeq{fractional_simulation_computation}) can be configured to compute \kw{SimRank} \cite{DBLP:conf/kdd/JehW02} and \kw{Rolesim} \cite{DBLP:conf/kdd/JinLH11}. As both the algorithms are applied to a single unlabeled graph, we let $G_1 = G_2$, and the graph be label-free.   %Note that SimRank is a similarity measure on computing unlabelled directed graphs, and it can only compute similarities between nodes in the same graph

%\textbf{SimRank} \cite{DBLP:conf/kdd/JehW02} is a structural similarity measure on computing unlabelled directed graph, and it can only compute node pairs in the same graph. Thus, we constrain the computation of $\Sim_{\chi}$ on one graph, i.e. $G_1 = G_2$ and set $\mathcal{L}(u,v) = 0$ for $\forall (u,v) \in V_1 \times V_1$. 
To configure $\Sim_\chi$ for \kw{SimRank}, if $u = v$, we set $\Sim^0_\chi(u, v)$ to 1 in the initialization step, and 0 otherwise. In the update step, we set $w^+ = 0$, $\mathcal{M}(S_1,S_2) = S_1 \times S_2$, $\normopr(S_1,S_2) = |S_1||S_2|$ and $\mathcal{L}(u, v) = 0$ in \refeq{fractional_simulation_computation}. It is clear that with such configurations, $\Sim_{\chi}$ computes \kw{SimRank} scores for all node pairs in a manner following \cite{DBLP:conf/kdd/JehW02}. Note that the convergence of $\Sim_{\chi}$ is guaranteed, as the mapping and normalizing operators satisfy all conditions in \refthm{convergence}.

\kw{RoleSim} \cite{DBLP:conf/kdd/JinLH11} computes structural similarity with automorphic confirmation (i.e. the similarity of two isomorphic nodes is 1) on an undirected graph. Thus, we let the out-neighbors of each node maintain its undirected neighbors, and leave the in-neighbors empty. In the initialization step, we set $\Sim^0_\chi(u, v) = \frac{\min(d^+(u), d^+(v))}{\max(d^+(u), d^+(v))}$ for all node pairs following \cite{DBLP:conf/kdd/JinLH11}. In the update step, we set $w^- = 0$ and $\mathcal{L}(u, v) = 1$ for each node pair, and follow the settings of mapping and normalizing operators of bijective simulation in \refeq{fractional_simulation_computation}. With such configurations, one can verify according to \cite{DBLP:conf/kdd/JinLH11} that $\Sim_\chi$ is computing axiomatic role similarity. %As a result, we conclude that $\kw{RoleSim}$ is a special case of $\Sim_\bjsim$ via proper initialization for unlabeled graphs.  

%and set $\mathcal{L}(u,v) = 1$ for all node pairs. In order to compute RoleSim scores in \cite{DBLP:conf/kdd/JinLH11}, we further set $\mathcal{M}(N(u),N(v)) = \mathcal{M}_{\bjsim}$ (\reftab{configurations}), and $\normopr(N(u),N(v)) = max(|N(u)|,|N(v)|)$ in \refeq{fractional_simulation_computation}, and change the initialization functions to those of RoleSim. According to \refrem{normalizing_bijective_simulation}, such configurations actually lead to a well-defined $\Sim_{\bjsim}$ measure. As a result, we can safely conclude that RoleSim is a special case of $\Sim_{\bjsim}$ for unlabelled undirected graphs.

% Analogously, the convergence of $\Sim_{\chi}$ is guaranteed as well, as $\mathcal{M}_{\bjsim}$ is by default set as maximum and the mapping and normalizing operators satisfy the all the conditions in \refthm{convergence}.

{\color{black}\stitle{Relation to $k$-bisimulation.} $k$-bisimulation \cite{DBLP:books/cu/12/AcetoIS12,DBLP:conf/bncod/LuoLFBHW13,DBLP:conf/cikm/LuoFHWB13,DBLP:conf/sac/HeeswijkFP16} is a type of approximate bisimulation. Given a graph $G(V,E,\lfunc)$ and an integer $k\geq 0$, node $u$ is simulated by node $v$ via $k$-bisimulation \cite{DBLP:conf/bncod/LuoLFBHW13} (i.e., $u$ and $v$ are $k$-bisimilar) if, and only if, the following conditions hold: (1) $\lfunc(u) = \lfunc(v)$; (2) if $k > 0$, for $\forall u' \in N_G^+(u)$, there exists $v' \in N_G^+(v)$ s.t. $u'$ and $v'$ are [k-1]-bisimilar; and (3) if $k > 0$, for $\forall v' \in N_G^+(v)$, there exists $u' \in N_G^+(u)$ s.t. $v'$ and $u'$ are [k-1]-bisimilar. An iterative framework is proposed by \cite{DBLP:conf/bncod/LuoLFBHW13} to compute $k$-bisimulation, in which each node $u$ is assigned with a signature ${sig}_k(u)$ based on its node label and neighbors' signatures. Node $u$ is simulated by node $v$ via $k$-bisimulation if and only if ${sig}_k(u) = {sig}_k(v)$ \cite{DBLP:conf/bncod/LuoLFBHW13}. We show in \refthm{k-bisimulation} that our $\Sim_{\chi}$ can be configured to compute $k$-bisimulation. As $k$-bisimulation in \cite{DBLP:conf/bncod/LuoLFBHW13} uses one single graph and only considers out-neighbors, we set $G_1 =G_2$ and $w^- = 0$ for $\Sim_{\chi}$. Recall that $\Sim_{\bisim}^k(u,v)$, computed by \refeq{fractional_simulation_computation}, is the $\bisim$-simulation score of nodes $u$ and $v$ in the $k$-[th] iteration,
% \vspace{-0.4em}
\begin{theorem} \label{thm:k-bisimulation}
Given a graph $G$ and an integer $k$, node $u$ is simulated by node $v$ via $k$-bisimulation if and only if {$\Sim_{\bisim}^k(u,v) = 1$.}%, where $\Sim_{\bisim}^k(\cdot)$ denotes the $\bisim$-simulation score in the $k$-[th] iteration}.
\end{theorem}
% \vspace{-1em}
\begin{proof}
The case when $k=0$ is easy to verify. Assume the theorem is true at $k-1$, we show that the theorem also holds at $k$. On the one hand, if $u$ is simulated by $v$ via $k$-bisimulation, i.e., ${sig}_{k}(u) = {sig}_{k}(v)$, one can verify that $\!\mathcal{M} = \{(u',v')|{sig}_{k-1}(u') = {sig}_{k-1}(v') \wedge v' \in N_G^+(v), \forall u' \in N_G^+(u)\} \bigcup \{(v'',u'')|{sig}_{k-1}(v'') = {sig}_{k-1}(u'') \wedge u'' \in N_G^+(u), \forall v'' \in N_G^+(v)\}\!$ is a matching of $\Sim_{\bisim}$. Based on the assumption, we have $\Sim_{\bisim}^k(u,v) = 1$. On the other hand, if $\Sim_{\bisim}^k(u,v) = 1$, for $\forall u' \in N_G^+(u)$, there exists $v' \in N_G^+(v)$ such that $\Sim_{\bisim}^{k-1}(u',v') = 1$, which means ${sig}_{k-1}(u') = {sig}_{k-1}(v')$. Similarly, $\forall v'' \in N_G^+(v)$, there exists $u'' \in N_G^+(u)$ with ${sig}_{k-1}(u'') = {sig}_{k-1}(v'')$. Thus, the set of signature values in $u$'s neighborhood is the same as that in $v$'s neighborhood. Then, we have ${sig}_{k}(u) = {sig}_{k}(v)$. %As a result, the theorem holds.
% one can verify that $\!\mathcal{M} = \{(u',v')|{sig}_{k-1}(u') = {sig}_{k-1}(v') \wedge v' \in N_G^+(v), \forall u' \in N_G^+(u)\} \bigcup \{(v'',u'')|{sig}_{k-1}(v'') = {sig}_{k-1}(u'') \wedge u'' \in N_G^+(u), \forall v'' \in N_G^+(v)\}\!$ is a matching of $\Sim_{\bisim}$. And also, if $\Sim_{\bisim}^k(u,v) = 1$, then $\forall u' \in N_G^+(u), \exists v' \in N_G^+(v)$ such that $\Sim_{\bisim}^{k-1}(u',v') = 1$ and ${sig}_{k-1}(u') = {sig}_{k-1}(v')$
% As a result, we have $\!\Sim_{\bisim}^k(u,v) = 1\!$ and the theorem holds.
\end{proof}
% \vspace{-1em}

\stitle{Relation to isomorphism.} The graph isomorphism test asks for whether two graphs are topologically identical, and node $u$ of $G_1$ is \emph{isomorphic} to node $v$ of $G_2$ if there exists an isomorphism between $G_1$ and $G_2$ mapping $u$ to $v$. {\color{black} Graph isomorphism is a challenging problem, and there is no polynomial-time solution yet \cite{DBLP:conf/stoc/Babai16}. The Weisfeiler-Lehman isomorphism test (the WL test) \cite{DBLP:journals/jmlr/ShervashidzeSLMB11} is a widely used solution to test whether two graphs are isomorphic. The WL test can be solved in polynomial time, but it is necessary but not sufficient for isomorphism, that is two graphs that are isomorphic must pass the WL test but not vice versa. 
%Apart from some corner cases, the Weisfeiler-Lehman isomorphism test (short as WL test) \cite{DBLP:journals/jmlr/ShervashidzeSLMB11} can distinguish a broad class of graphs \cite{DBLP:conf/focs/BabaiK79} and is widely used as it is solvable in polynomial time. 
%Weisfeiler-Lehman isomorphism test (short as WL test) \cite{DBLP:journals/jmlr/ShervashidzeSLMB11} is a polynomial-time solution to test whether two graphs are isomorphic without false negative (false positives are rare but possible).
%Apart from some corner cases\footnote{There exist some non-isomorphic graphs that can pass the WL test.}, the Weisfeiler-Lehman isomorphism test (or WL test) \cite{DBLP:journals/jmlr/ShervashidzeSLMB11} is an effective solution that can distinguish a broad class of graphs \cite{DBLP:conf/focs/BabaiK79} and is widely used as it is solvable in polynomial time.} 
%Analogous to WL test, the bijective simulation is also a necessary but insufficient condition for the isomorphism, that is if node $u$ is isomorphic to node $v$, $u$ is $\bjsim$-simulated by $v$ but not vice versa. 
Like the WL test, bijective simulation is also necessary but not sufficient for isomorphism. We next show that it is as powerful as the WL test in theory.
}

% Analogous to WL test, the bijective simulation is a necessary but insufficient condition for the isomorphism, that is if node $u$ is isomorphic to node $v$, $u$ is $\bjsim$-simulated by $v$ but not vice versa. 

% Given two undirected (the graph model is accordingly adapted as \kw{RoleSim}) labeled graphs $G_1(V_1, E_1, \lfunc_1)$ and $G_2(V_2, E_2, \lfunc_2)$, the WL test iteratively labels each node $u \in V_1$ (resp. $v \in V_2$) as $s(u)$ (resp. $s(v)$) by aggregating its node label and labels of its neighbors. The algorithm decides node $u$ is isomorphic to node $v$ with high probability if $s(u) = s(v)$ when the algorithm converged \footnote{Note that the algorithm is not guaranteed to converge.}  
% and if it converged, w

The WL test \cite{DBLP:journals/jmlr/ShervashidzeSLMB11} is applied to undirected labeled graphs, and the graph model is accordingly adapted as \kw{RoleSim}. We assume both graphs are connected, as otherwise each pair of connected components can be independently tested. Given graphs $G_1$ and $G_2$, the WL test iteratively labels each node $u \in V_1$ (resp. $v \in V_2$) as $s(u)$ (resp. $s(v)$). The algorithm decides that node $u$ is isomorphic to node $v$ if $s(u) = s(v)$ when the algorithm converges\footnote{The algorithm is not guaranteed to converge.} The following theorem reveals the connection between WL test and bijective simulation.} 

% The WL test decides two graphs are non-isomorphic if at some iteration the labels of nodes in $G_1$ are different from those of $G_2$, and it deems 

% and if it converged\footnote{Note that the algorithm is not guaranteed to converge.}, we have $s(u) = s(v)$, $\forall (u, v) \in V_1 \times V_2$, if there exists an isomorphism $f: V_1 \to V_2$ such that $f(u) = v$. The following theorem reveals the connection between WL test and bijective simulation. 

% The WL test \cite{DBLP:journals/jmlr/ShervashidzeSLMB11} is applied to undirected labeled graphs and thus the graph model is accordingly adapted as \kw{RoleSim}. We assume both graphs are connected, as otherwise each pair of connected components can be independently tested. The algorithm iteratively labels a node $u$ as $s(u)$, and if it converged\footnote{Note that the algorithm is not guaranteed to converge.}, we have $s(u) = s(v)$, $\forall (u, v) \in V_1 \times V_2$, if there exists an isomorphism $f: V_1 \to V_2$ such that $f(u) = v$. The following theorem reveals the connection between WL test and bijective simulation. 
% \vspace{-0.6em}
\begin{theorem}\label{thm:bijection_and_wl_test}
Given graphs $G_1$ and $G_2$, and a node pair $(u, v) \in V_1 \times V_2$, and assume the WL test converges, we have $s(u) = s(v)$ \textbf{if and only if} $\Sim_\bjsim(u, v) = 1$, namely $u {\sim}^\bjsim v$. 
\end{theorem}
% \vspace{-1em}
\begin{proof}
%To better understand this proof, one can refer to algorithm 1 in \cite{DBLP:journals/jmlr/ShervashidzeSLMB11}. 
%We omit the process of label compression and relabeling in \cite{DBLP:journals/jmlr/ShervashidzeSLMB11}, and use the string directly as the label for clarity. 
Let $s^k(u)$ and $s^k(v)$ be the label of $u$ and $v$ at the $k$-[th] iteration during WL test. We first prove that for any $k$, if $s^k(u) = s^k(v)$, $\Sim_{\bjsim}^k(u,v) = 1$. The case of $k=0$ is easy to verify. Suppose the theorem is true at $k-1$. At the $k$-[th] iteration, we have $s^k(u)$ = $s^{k-1}(u) \sqcup_{u'\in N(u)} {s^{k-1}(u')}$ and $s^k(v)$ = $s^{k-1}(v) \sqcup_{v'\in N(v)} {s^{k-1}(v')}$, where $\sqcup$ denotes label concatenation. If $s^k(u) = s^k(v)$, there exists a bijective function $\!\lambda_1: N_{G_1}(u) \to N_{G_2}(v)\!$ s.t. $s^{k-1}(u') = s^{k-1}(\lambda_1(u'))$. Based on the assumption, we have $\Sim_{\bjsim}^{k}(u,v) = 1$. %Based on \refthm{bijectivesimulation}, we have if $s(u) = s(v)$, $u \overset{\bjsim}{\sim} v$ holds.

Next, we prove if $u {\sim}^{\bjsim} v$, $s(u) = s(v)$. It is easy to verify the case of $k=0$. Assume that if $\Sim_{\bjsim}^{k-1}(u,v) = 1$, $\!s^{k-1}(u) = s^{k-1}(v)\!$ holds. At the $k$-[th] iteration, if $\Sim_{\bjsim}^{k}(u,v) = 1$, there exits a bijective function $\!\lambda_2: N_{G_1}(u) \to N_{G_2}(v)\!$ s.t. for $\forall u' \in N_{G_1}(u)$, $\Sim_{\bjsim}^{k-1}(u',\lambda_2(u')) = 1$. Thus, we can derive $s^{k-1}(u') = s^{k-1}(\lambda_2(u'))$ and $s^k(u) = s^k(v)$. 
\end{proof}
% \vspace{-1em}

\begin{remark}
\color{black}
Note that there is no clear relation between bijective simulation and graph homomorphism. To be specific, bijective simulation cannot derive homomorphism, and homomorphism cannot derive bijective simulation either.
\end{remark}
% \vspace{-1em}

% Next, we prove that if $u \overset{\bjsim}{\sim} v$, $s(u) = s(v)$ is true. Analogously, we prove that if $\Sim_{\bjsim}^k(u,v) = 1$, $s^k(u) = s^k(v)$ holds at any iteration $k \geq 0$ by using mathematical induction.

% \textit{Induction Basis.} At the 0-[th] iteration, if $\Sim_{\bjsim}^0(u,v) = 1$, $\lfunc(u)=\lfunc(v)$. Accordingly, we have $s^0(u) = s^0(v)$. 

% \textit{Induction Step.} Suppose that if $\Sim_{\bjsim}^{k-1}(u,v) = 1$, $s^{k-1}(u) = s^{k-1}(v)$ holds. At the $k$-[th] iteration, if $\Sim_{\bjsim}^{k}(u,v) = 1$, there exits a bijective function $\lambda_2: N(u) \to N(v)$ s.t. $\Sim_{\bjsim}^{k-1}(u',\lambda_2(u')) = 1$, for $\forall u' \in N(u)$. Denote $v' = \lambda_2(u')$, we have $s^{k-1}(u') = s^{k-1}(v')$ based on the assumption. Hence, we can derive $s^k(u) = s^k(v)$. 
% Based on the above analysis, the theorem holds.

% \input{sec5.3_bijective_simulation.tex}
% \input{sec5.4_selection_of_simulation_variants.tex}
% \input{sec7.1_effectiveness.tex}
\section{Experimental Evaluation} \label{sec:experiment}
%We exhibit the experimental results in this section. 
% Our experiments in this section focus on the following perspectives: (1) $\Sim_\chi$ is sufficiently robust against parameter tuning and data errors (\refsec{sensitivity_analysis}); (2) $\Sim_\chi$ is efficient to compute on real-world graphs (\refsec{efficiency}); (3) $\Sim_\chi$ is effective on real-world applications (\refsec{case_studies}). 

\subsection{Setup} \label{sec:setup}
% \vspace{-0.5em}

\stitle{Datasets.} We used eight publicly available real-world datasets. \reftab{graph_statistics} provides their descriptive statistics, including the number of nodes $|V|$, the number of edges $|E|$, the number of labels $|\Sigma|$, the average degree $d_G$, the maximum out-degree $D^+_G$ and the maximum in-degree $D^-_G$. 

\begin{table}[h]
    \small
    \centering
    \topcaption{Dataset Statistics and Sources} \label{tab:graph_statistics}
    % \scalebox{0.86}{
    % % \color{black}
    % \setlength{\tabcolsep}{1.1mm}{
    \begin{tabular}{|c||c|c|c|c|c|c|c|} \hline
      \textbf{Datasets} & $|E|$ & $|V|$ & $|\Sigma|$ & $d_G$ & $D^+_G$ & $D^-_G$ & Source\\ \hline \hline
      Yeast& 7,182 & 2,361 & 13 & 3 & 60 & 47 & \cite{kunegis2013konect} \\ \hline
      Cora& 91,500 & 23,166 & 70 & 4 & 104 & 376 & \cite{kunegis2013konect}\\ \hline
      Wiki& 119,882 & 4,592 & 120 & 26 & 294 & 1,551 & \cite{kunegis2013konect}\\ \hline
      JDK& 150,985 & 6,434 & 41 & 23 & 375 & 32,507 & \cite{kunegis2013konect}\\ \hline
      NELL& 154,213 & 75,492 & 269 & 2 & 1,011 & 1,909 & \cite{xiong2017deeppath} \\ \hline
      GP & 298,564 & 144,879 & 8 & 2 & 191 & 18,553 & \cite{harding2017iuphar} \\ \hline
      Amazon& 1,788,725 & 554,790 & 82 & 3 & 5 & 549 & \cite{snapnets} \\ \hline
      ACMCit & 9,671,895 & 1,462,947 & 72K & 7 & 809 & 938,039 & \cite{DBLP:conf/kdd/TangZYLZS08} \\ \hline
    \end{tabular}%}}
    % \vspace{-1.6em}
\end{table}

\stitle{Experimental Settings.} Without loss of generality, we assume that in-neighbors and out-neighbors contribute equally to the $\Sim_{\chi}$ computation. Thus, $w^+ = w^-$ in all experiments. Algorithms were terminated when the values changed by less than 0.01 of their previous values. \textit{Note that when we applied $\Sim_{\chi}$ to one single graph, we actually computed the $\Sim_{\chi}$ scores from the graph to itself.} $\Sim_{\chi} \{\theta = a\}$ and $\Sim_{\chi} \{\ub\}$ denote the computation of $\Sim_{\chi}$ uses the optimizations of label-constrained mapping (setting $\theta = a$) and upper-bound updating (\refsec{computing_algorithm}), respectively. The two optimizations can be meanwhile used as $\Sim_{\chi} \{\ub, \theta=a\}$. We use $\theta=0$ by default, which will be omitted for simplicity thereafter. %and we 
%accordingly use $\Sim_{\chi} \{\theta = x\}$ and $\Sim_{\chi} \{\ub\}$ to denote $\Sim_{\chi}$ with the respective optimizations, in which $x$ is used to specify the $\theta$ value in the label-constrained mapping. Specifically, we use $\Sim_{\chi}$ for short if $\theta$ is set to 0. 

We implemented $\Sim_{\chi}$ in C++. All experiments were conducted on a platform comprising two Intel(R) Xeon(R) CPU E5-2698 v4 @ 2.20GHz (each with 20 cores) and 512GB memory.

% \begin{table} 
%     \small
%     \centering
%     \topcaption{Graph Statistics} \label{tab:graph_statistics}
%     \scalebox{0.86}{
%     % \color{black}
%     \begin{tabular}{|c||c|c|c|c|c|c|} \hline
%       \textbf{Datasets} & $|E|$ & $|V|$ & $|\Sigma|$ & $d_G$ & $D^+_G$ & $D^-_G$\\ \hline \hline
%       Yeast& 7,182 & 2,361 & 13 & 3 & 60 & 47\\ \hline
%       Cora& 91,500 & 23,166 & 70 & 4 & 104 & 376\\ \hline
%       Wiki& 119,882 & 4,592 & 120 & 26 & 294 & 1,551\\ \hline
%       JDK& 150,985 & 6,434 & 41 & 23 & 375 & 32,507\\ \hline
%       NELL& 154,213 & 75,492 & 269 & 2 & 1,011 & 1,909\\ \hline
%       GP & 298,564 & 144,879 & 8 & 2 & 191 & 18,553\\ \hline
%       Amazon& 1,788,725 & 554,790 & 82 & 3 & 5 & 549\\ \hline
%       ACMCit & 9,671,895 & 1,462,947 & 72K & 7 & 809 & 938,039\\ \hline
%     \end{tabular}}
%     \vspace{-1.6em}
% \end{table}

\subsection{Sensitivity Analysis} \label{sec:sensitivity_analysis}
Our first test was a sensitivity analysis to examine $\Sim_\chi$'s robustness to parameter tuning and data errors. Following \cite{DBLP:conf/kdd/JinLH11}, we calculated Pearson's correlation coefficients. The larger the coefficient, the more correlated the evaluated subjects. Note that the patterns were similar across datasets. Hence, only the results for NELL are reported.
% Our first test was a sensitivity analysis to examine $\Sim_\chi$'s robustness to parameter tuning and data errors. Following \cite{DBLP:conf/kdd/JinLH11}, we computed Pearson's correlation coefficients, which gives a score within $[-1.0,1.0]$ to evaluate the degree of correlations of the scores computed with different settings. The larger the correlation value, the more correlated the evaluated subjects. Note that we only present the results on the NELL graph, while the results on the other graphs are similar.
% We first perform the sensitivity analysis to study whether $\!\Sim_\chi\!$ is robust to parameter tuning and data errors. We follow \cite{DBLP:conf/kdd/JinLH11} to adopt the Pearson correlation measurement, which gives a score within $[-1.0,1.0]$ to evaluate the degree of correlations of the scores computed with different settings (the larger, the better). Note that we only present the results on the NELL graph, while the results on the other graphs are similar.

\stitle{Sensitivity of Framework Parameters.} We performed the sensitivity analysis against three parameter settings: (1) the initialization function $\inifunc(\cdot)$ presented in \refsec{iterative_computation}; (2) the threshold $\theta$ for the label-constrained mapping outlined in \refrem{label_constraint_mapping}; and (3) the weighting factors outlined in \refsec{iterative_computation}. 

\sstitle{\color{black}Varying $\inifunc(\cdot)$.} In this analysis, we computed and cross-compared the $\Sim_{\chi}$ scores using the three different initialization functions: indicator function $\inifunc_I(\cdot)$, normalized edit distance $\inifunc_E(\cdot)$, and Jaro-Winkler similarity $\inifunc_{J}(\cdot)$. The results are shown in \reftab{compare_initialization}. The Pearson's coefficients for all pairs of initialization functions are very high ($>0.92$), which indicates that $\Sim_{\chi}$ is not sensitive to initialization functions. Hence, going forward, we used $\inifunc_{J}(\cdot)$ as the initialization function unless specified otherwise.
%In addition, the computati-on of $\Sim_\simu$ is the most robust while varyinging $\inifunc(\cdot)$, followed by $\Sim_\dpsim$, $\Sim_\bisim$ and $\Sim_\bjsim$.
\begin{table}[h]
    \centering
    \small
    \topcaption{Pearson's correlation coefficients when comparing initialization functions.} \label{tab:compare_initialization}
    % \scalebox{0.9}{
        \begin{tabular}{|c|c|c|c|c|}  \hline
            $\Sim_{\chi}$ & $\Sim_{\simu}$ & $\Sim_{\dpsim}$ & $\Sim_{\bisim}$ & $\Sim_{\bjsim}$ \\ \hline
            $\inifunc_{I}$- $\inifunc_{E}$ & 0.990 & 0.982 & 0.979 & 0.969 \\ \hline
            $\inifunc_{I}$-$ \inifunc_{J}$ & 0.967 & 0.950 & 0.937 & 0.922 \\ \hline
            $\inifunc_{J}$- $\inifunc_{E}$ & 0.985 & 0.977 & 0.975 & 0.962 \\ \hline
    \end{tabular}%}
    % \vspace{-1.8em}
\end{table}

\sstitle{\color{black}Varying $\theta$.} For this analysis, we varied $\theta$ from 0 to 1 in steps of 0.2, and calculated the Pearson's coefficient against the baseline case of $\theta = 0$ (with $w^+$ and $w^-$ set to 0.4). The results in \reffig{correlation_vary_theta} clearly show that the coefficients decrease as $\theta$ increases. %{\color{black}The reason is that when $\theta$ gets larger, there are more node pairs that are meant to be mapped when $\theta = 0$, cannot be mapped anymore.} 
{\color{black}This is reasonable as node pairs with $\mathcal{L}(\cdot) < \theta$ will not be considered by the mapping operator. Also, more node pairs are pruned as $\theta$ grows.} However, the coefficients are still very high ($>0.8$) for all variants, even when $\theta = 1$, which indicates that $\Sim_{\chi}$ is not sensitive to $\theta$. 

\sstitle{\color{black}Varying $w^*$.} To examine the influence of the weighting parameters, we varied $w^*$ from 0.1 to 1, where $w^* = 1-w^+-w^-$.
Recall that $\theta = 1$ constrains mapping only the same-label nodes (\refrem{label_constraint_mapping}). As $w^*$ is label-relevant, we computed the coefficients of $\Sim_\chi$ (vs. $\Sim_\chi\{\theta = 1\}$) by varying $w^*$. The results, reported in \reffig{vary_w}, show that the coefficients increase as $w^*$ increases and at $w^* > 0.6$, the coefficient is already almost 1. This is expected because a larger $w^*$ mitigates the impact of the label-constrained mapping. At a more reasonable setting of $w^* = 0.2$, the coefficients sit at around 0.85, which indicates that $\Sim_\chi\{\theta = 1\}$ aligns with $\Sim_\chi$ well. Hence, we set $w^* = 0.2$ by default in subsequent tests.

\begin{figure}[h]
    \centering
    % \captionsetup[subfigure]{aboveskip=0.01cm,belowskip = 0.00cm}
    \subfigure[varying $\theta$]{
      	\centering
		\includegraphics[width=0.38\linewidth]{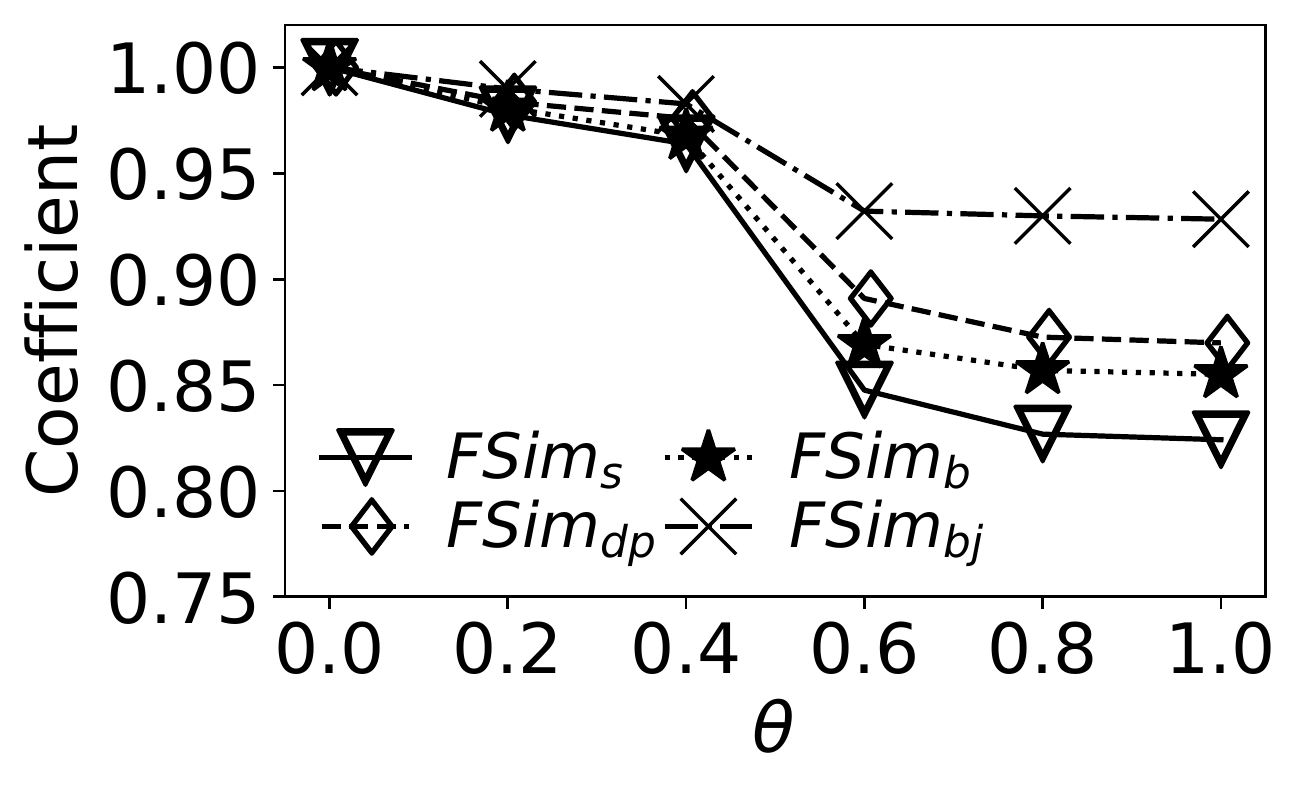}
        \label{fig:correlation_vary_theta}
      }
    %  \hspace{-0.5em}
    \subfigure[varying $w^*$]{
      \includegraphics[width=0.38\linewidth]{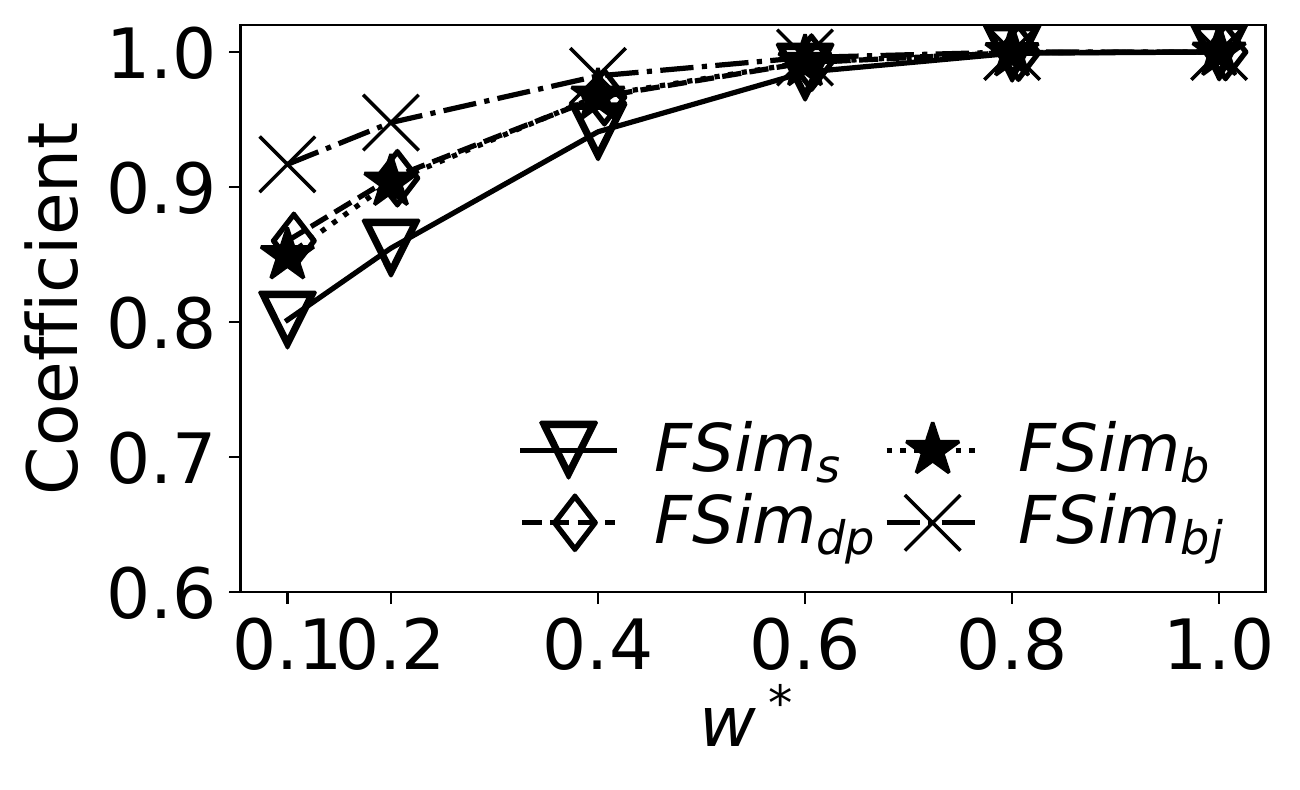} \label{fig:correlation_vary_w}
      }
      \topcaption{Pearson's correlation coefficients when varying $\theta$ and $w^*$}
      \label{fig:vary_w}
    %   \vspace{-1.3em} 
\end{figure}

\stitle{Robustness against Data Errors.} \reffig{correlation_vary_data_errors} plots the robustness of $\Sim_{\bjsim}$ against data errors, i.e., structural errors (with edges added/removed) and label errors (with certain labels missing), from one extreme ($\theta = 0$) to the other ($\theta = 1$) as an example of how all simulation variants performed. It is expected that the coefficients decrease as the error level increases. Yet, the coefficients remained high even at the 20\% error level ($>0.7$ for both cases). This shows that $\Sim_{\chi}$ is robust to data errors, which conforms with one of the reasons why we initially thought to propose fractional simulation. 

\begin{figure}[h]
    \centering
    \subfigure[varying structural errors]{
      \includegraphics[width=0.38\linewidth]{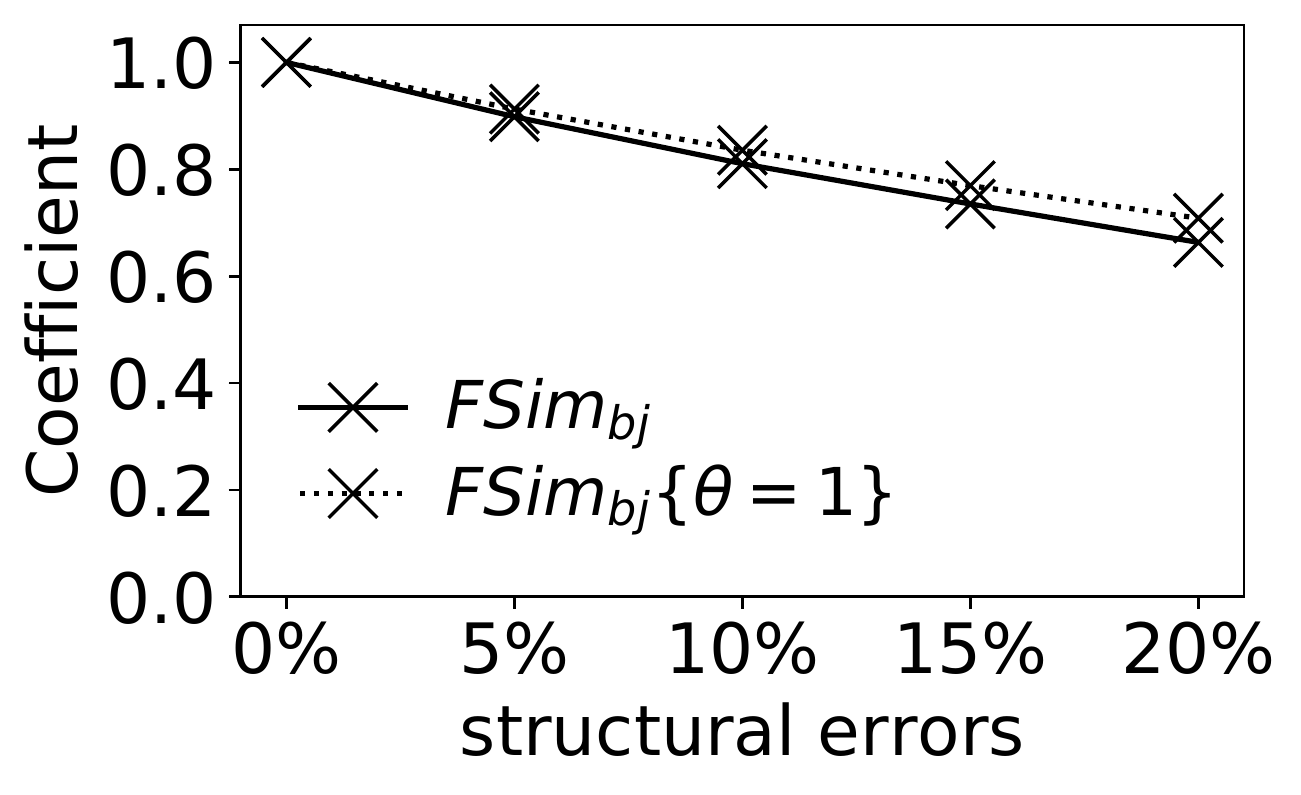} \label{fig:vary_structural_noise}
      }
    %   \hspace{-0.5em}
      \subfigure[varying label errors]{
      \includegraphics[width=0.38\linewidth]{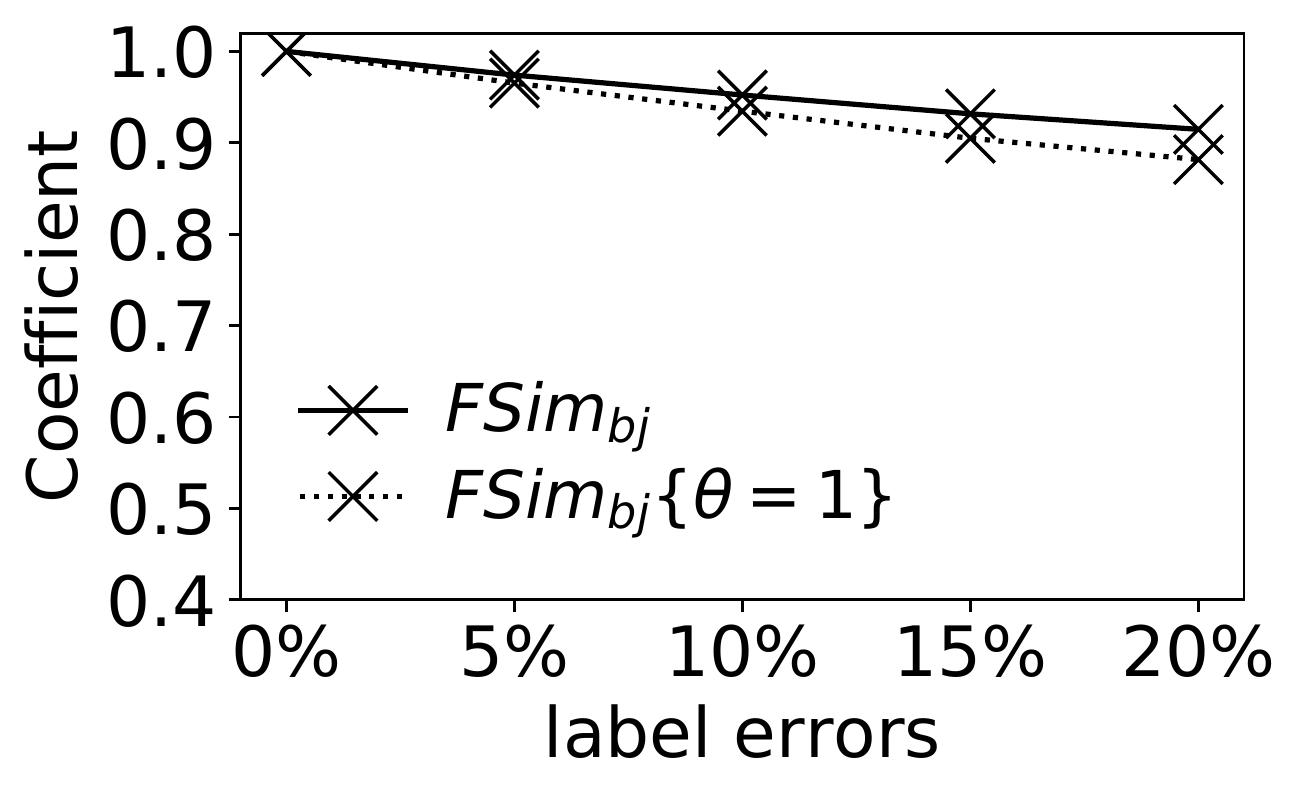} \label{fig:vary_label_noise}
      }
    %   \vspace{-0.2em}
      \topcaption{Pearson's correlation coefficients when varying the ratio of data errors}
      \label{fig:correlation_vary_data_errors}
    % \vspace{-1.8em}
\end{figure}

\stitle{Sensitivity of Upper-bound Updating.} To assess the influence of upper-bound updating (\refsec{computing_algorithm}), we varied $\alpha$ (the approximate ratio) from 0 to 0.5 and $\beta$ (the threshold) from 0 to 1 in steps of $0.1$. Again, the results for all simulation variants were similar, so only the results for $\Sim_{\bjsim}\{\ub\}$ (vs. $\Sim_{\bjsim}$) and $\!\Sim_\bjsim\{\ub, \theta = 1\!\}$ (vs. $\Sim_\bjsim\{\theta = 1\}$) are shown.

\sstitle{\color{black}Varying $\beta$.} \reffig{vary_beta} shows the coefficients while varying $\beta$ from 0 to 0.5 with $\alpha$ fixed to 0.2. It is clear that the coefficients decrease as $\beta$ increases. This is reasonable as more node pairs are pruned, and the scores become less precise as $\beta$ gets larger. Note that when $\beta \geq 0.3$, the decreasing trend becomes smoother for $\Sim_{\bjsim}\{\ub, \theta = 1\}$. Observe that even at $\beta = 0.5$, the coefficients are still very high ($>0.9$), which indicates that the validity of upper-bound updating is not sensitive to $\beta$. We thus set $\beta = 0.5$ going forward to utilize as much pruning power as possible.

\begin{figure}[h]
	\centering
	\subfigure[varying $\beta$]{
		\centering
		\includegraphics[width=0.38\linewidth]{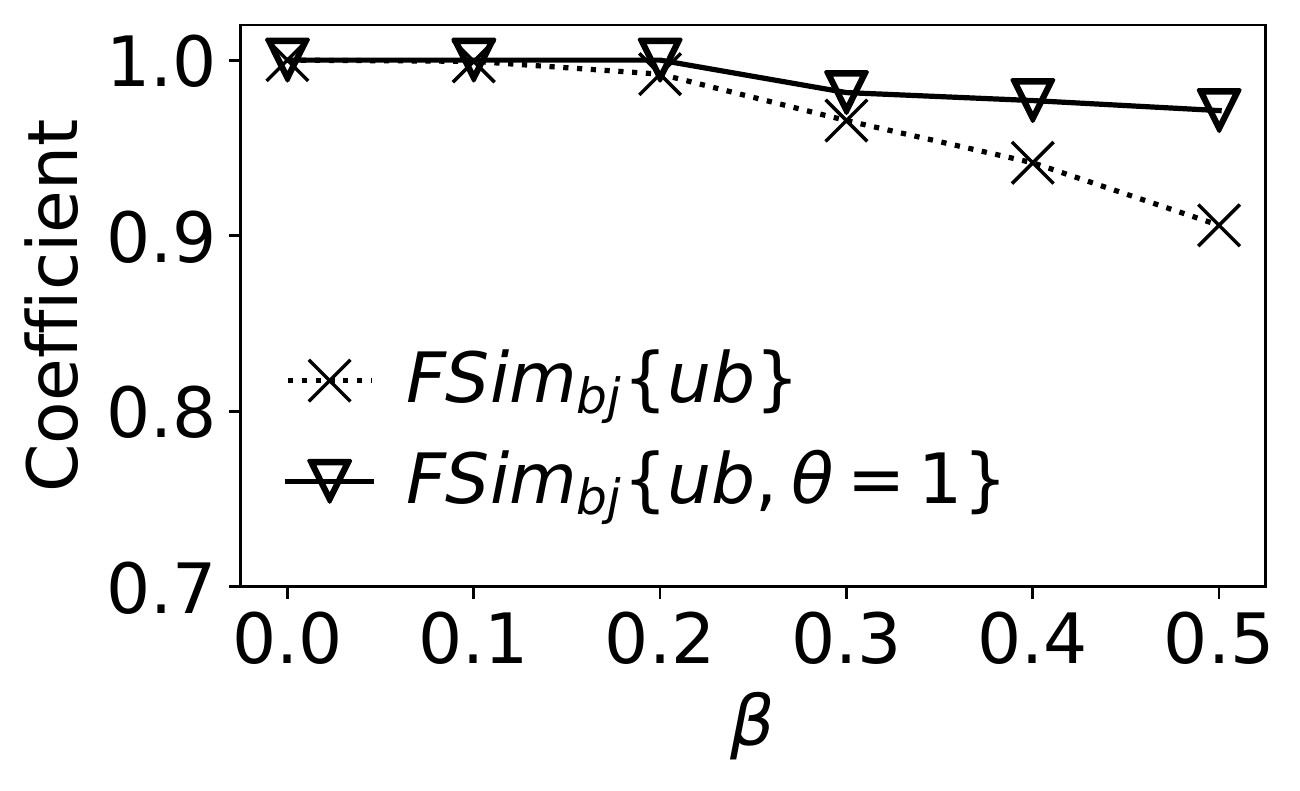}
		\label{fig:vary_beta}
	}
% 	\hspace{-0.5em}
	\subfigure[varying $\alpha$]{
		\centering
		\includegraphics[width=0.38\linewidth]{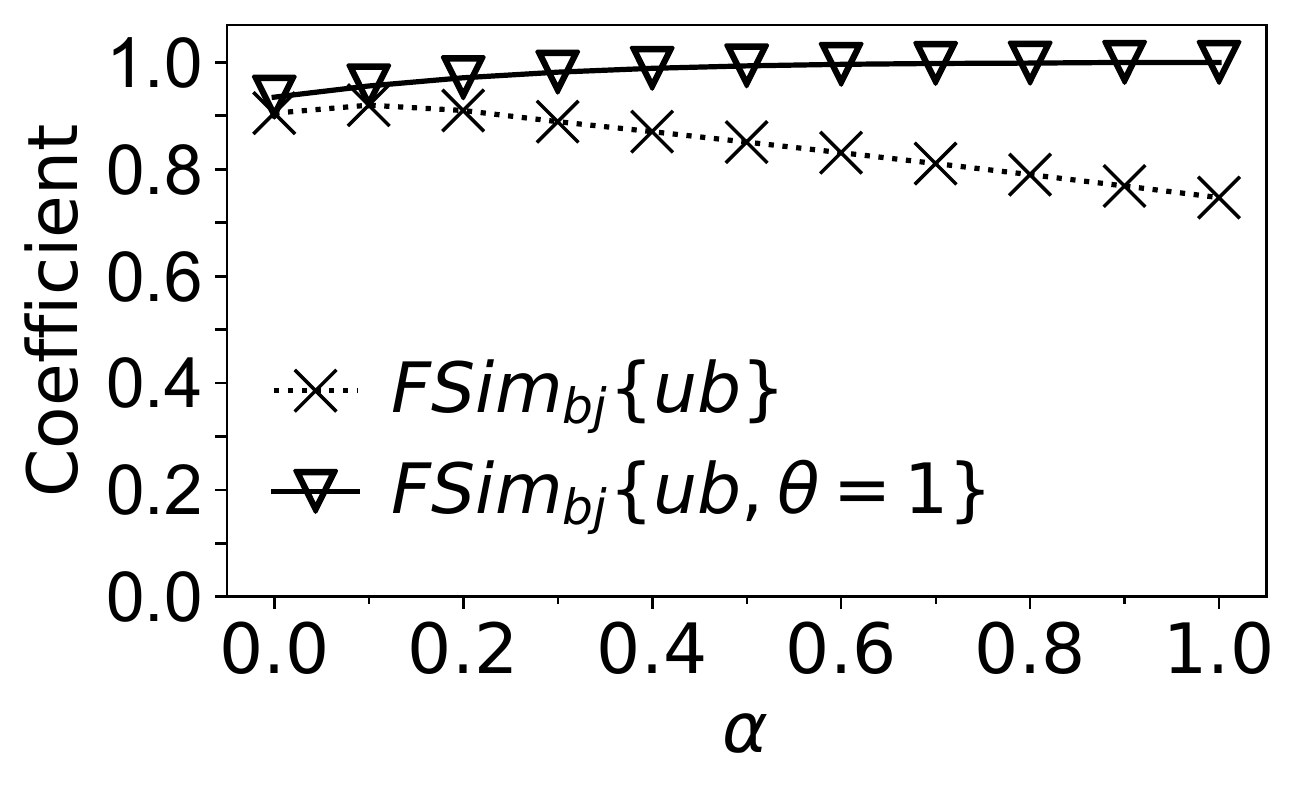}
		\label{fig:vary_alpha}
	}
	\topcaption{Pearson's correlation coefficients when varying $\alpha$ and $\beta$} \label{fig:sensitivity_upper_bound}
% 	\vspace{-1.5em}
\end{figure}

\sstitle{\color{black}Varying $\alpha$.} \reffig{vary_alpha} shows the coefficients when varying $\alpha$ from 0.0 to 1.0. We made two observations here. First, the coefficients of $\Sim_\bjsim\{\ub\}$ initially increase, then decrease as $\alpha$ gets larger. A possible reason is that $\alpha = 0$ and $\alpha = 1$ are at each extreme of the setting range, but the most appropriate setting lies somewhere in between. Second, the coefficients for $\Sim_\bjsim\{\ub, \theta = 1\}$ increase as $\alpha$ increases. Potentially, the true scores of pruned node pairs are larger than $1-w^+-w^-$, and thus a larger $\alpha$ is preferred. Note that when $\alpha = 0$, {\color{black}i.e., when ignoring the pruned node pairs,} the coefficients for both $\Sim_\bjsim\{\ub\}$ and $\Sim_\bjsim\{\ub, \theta = 1\}$ were above 0.9; hence, $\alpha = 0$ became our default.

% \sstitle{\color{black}Varying $\alpha$.} \reffig{vary_alpha} shows the coefficients when varying $\alpha$ from 0.0 to 1.0. On the one hand, the correlation values of $\Sim_\bjsim\{\ub\}$ first increase and then decrease when $\alpha$ gets larger. This is because $\alpha = 0$ and $\alpha = 1$ stand for two extremes of setting scores for the pruned node pairs, and the most appropriate setting must lie in between. On the other hand, the correlations of $\Sim_\bjsim\{\ub, \theta = 1\}$ increases when $\alpha$ gets larger. Potentially, the true scores of pruned node pairs are larger than $1-w^+-w^-$, and thus a larger $\alpha$ is preferred. Note that when $\alpha = 0$, the correlations of both $\Sim_\bjsim\{\ub\}$ and $\Sim_\bjsim\{\ub, \theta = 1\}$ are above 0.9, we set $\alpha = 0$ {\color{black}(i.e., ignore pruned node pairs)} by default in the following.
\subsection{Efficiency}  \label{sec:efficiency}
%  In this section, we experiment the efficiency of $\Sim_{\chi}$. We first use the NELL graph by default while varying $\theta$ for all simulation variants. Then we give the running time of $\Sim_{\chi}$ on different datasets with different optimizations. Finally, we present the results of parallelization and scalability.

\stitle{{\color{black}Varying $\theta$.}} With NELL as a representative of all tests, \reffig{time_vary_theta} shows the running time of $\Sim_{\chi}$ while varying $\theta$ from 0 to 1. The experimental results show that $\Sim_{\chi}$ runs faster as $\theta$ increases, which is expected since a larger $\theta$ contributes to less candidate pairs to compute, as shown in \reffig{node_pairs_vary_theta}. We then compared the running time of different simulation variants under certain $\theta$ value. It is not surprising that $\Sim_\dpsim$ and $\Sim_\bjsim$ ran slower than the other two variants, as they contain a costly maximum-matching operation (cost analysis in \refsec{configurations_of_all_variants}). $\Sim_\bisim$ ran slower than $\Sim_\simu$ because the mapping operator of $\Sim_\bisim$ considers both neighbors of a node pair(\reftab{configurations}). At $\theta \geq 0.6$, the difference in running time for all variants was already very small. Considering the sensitivity analysis in \reffig{correlation_vary_theta} as well as these results, $\theta = 1$ seems a reasonable setting that renders both good coefficients and performance.

\begin{figure}[h]
\centering
	\subfigure[running time]{
		\centering
        \includegraphics[width=0.38\linewidth]{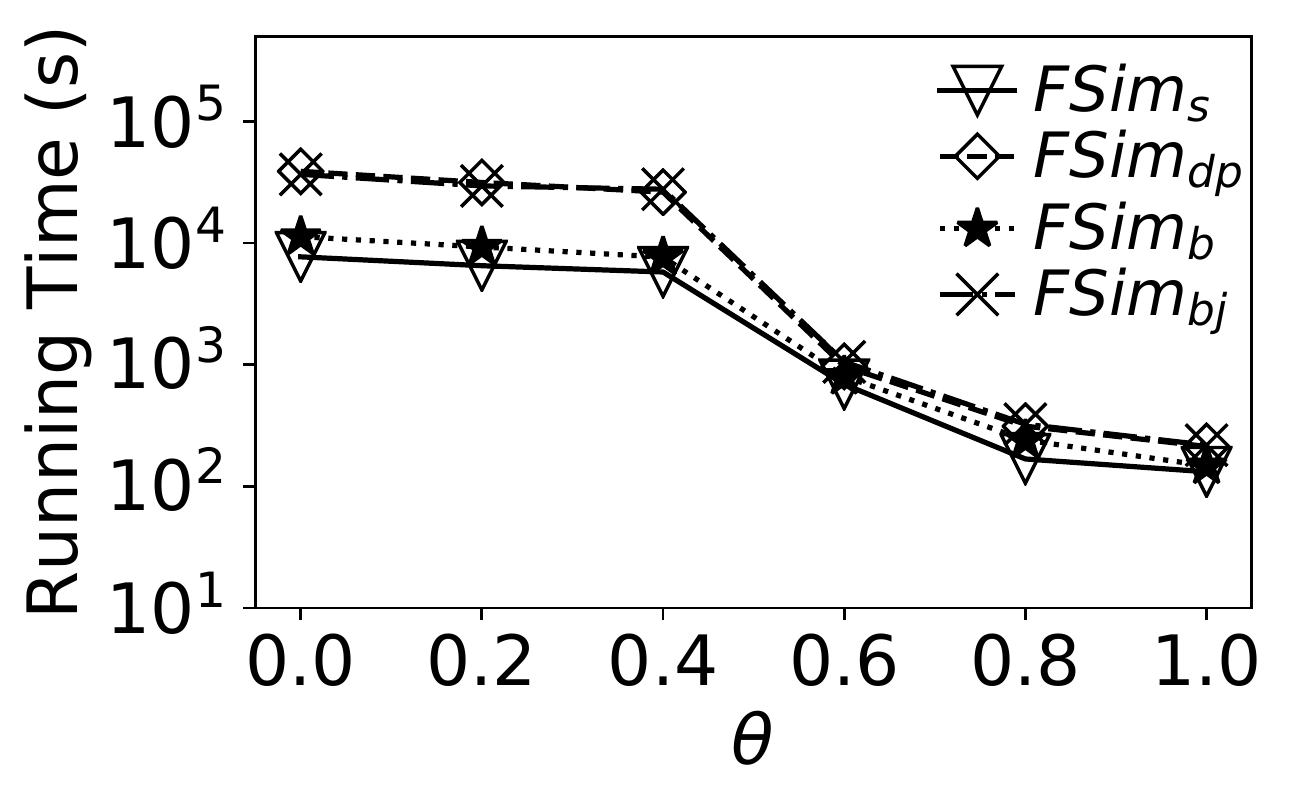}
        % \topcaption{Running time of $\Sim_{\chi}$ while varying $\theta$}
        \label{fig:time_vary_theta}
	}
% 	\hspace{-0.5em}
	\subfigure[number of node pairs]{
		\centering
        \includegraphics[width=0.38\linewidth]{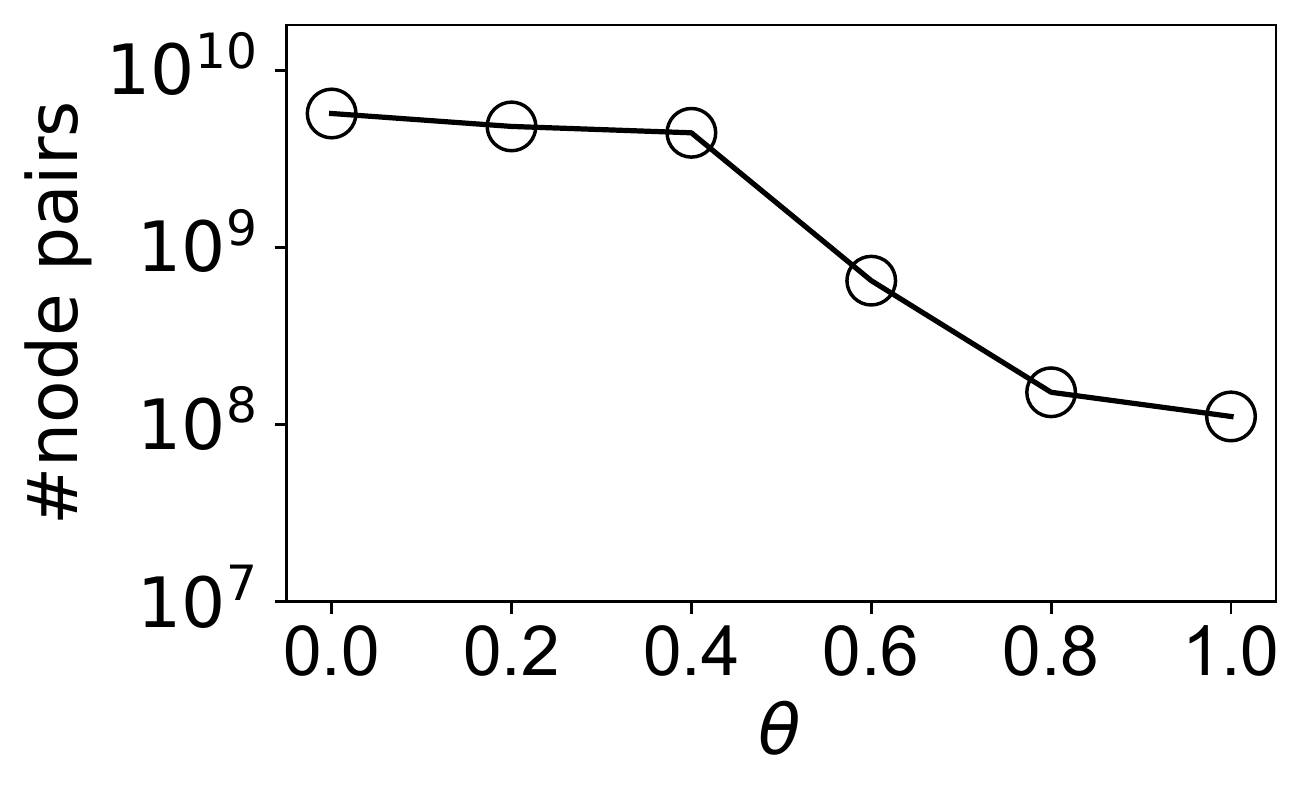}
        % \topcaption{Scalability of $\Sim_{\chi}$}
        \label{fig:node_pairs_vary_theta}
	}
% 	\subfigure[varying thread]{
% 		\centering
%         \includegraphics[width=0.46\linewidth]{figures/time_vary_thread_new.pdf}
%         % \topcaption{Scalability of $\Sim_{\chi}$}
%         \label{fig:time_vary_thread}
% 	}
	\topcaption{Running time of $\Sim_{\chi}$, $\chi \in \{\simu,\bisim,\dpsim,\bjsim\}$, while varying $\theta$}
% 	\vspace{-1.8em}
\end{figure}

\stitle{\color{black}Varying the Datasets.} \reffig{vary_opt} reports the running time of $\Sim_{\bjsim}$, the most costly simulation variant, with different optimizations on all datasets. Additionally, experiments that resulted in out-of-memory errors have been omitted. From these tests, we made the observations: (1) the upper-bound updating alone contributed about 5$\times$ the performance gain compared to $\Sim_{\bjsim}\{\ub\}$ with $\Sim_{\bjsim}$. (2) Label-constrained mapping is the most effective optimization, making $\!\Sim_{\bjsim}\{\theta = 1\}\!$ faster than  $\!\Sim_{\bjsim}\!$ by up to 3 orders of magnitude. Applying both label-constrained mapping and upper-bound updating, $\Sim_{\bjsim}\{\ub, \theta = 1\}$ was the only algorithm that could complete the executions on all datasets in time, including the two largest ones, Amazon and ACMCit.

\begin{figure}[h]
    \centering
    \includegraphics[width=0.8\linewidth]{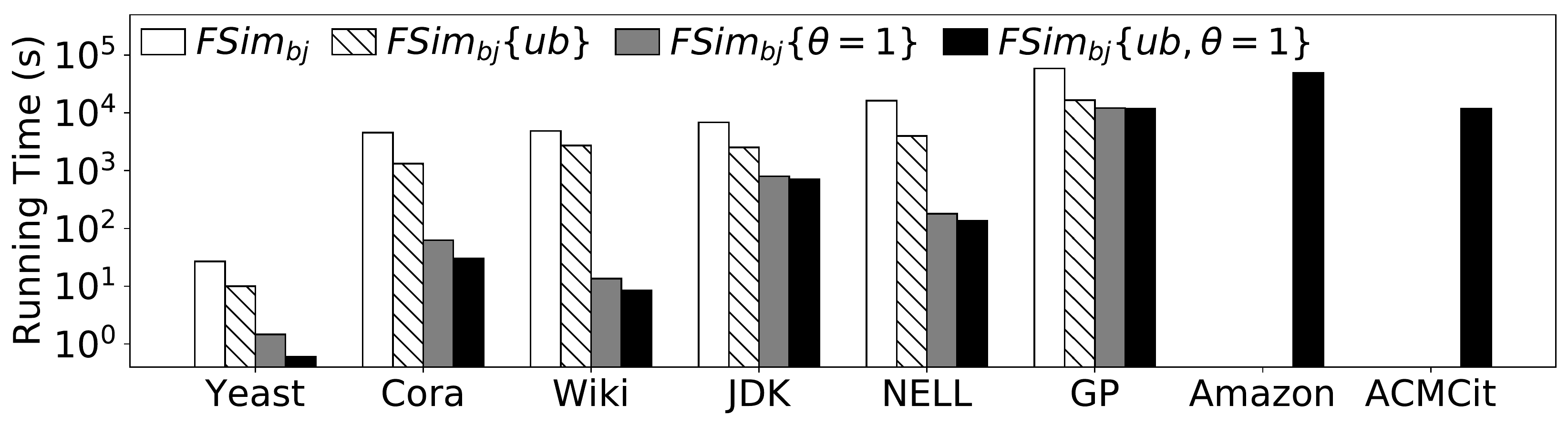}
    % \vspace{-0.5em}
    \topcaption{Running time of $\Sim_{\bjsim}$ on all datasets with different optimizations}  \label{fig:vary_opt}
    % \vspace{-1.3em}
\end{figure}

\stitle{Parallelization and Scalability.} {\color{black}We studied the scalability of $\Sim_{\chi}$ with parallelization on two representative datasets, i.e., NELL and ACMCit (with more than 1 million nodes). The results for $\Sim_{\bjsim}\{\ub,\theta = 1\}$ follow.}
%{\color{black}We discuss in \refsec{computing_algorithm} the computation of $\Sim_{\chi}$ can be easily paralleled, and thus we study the scalability of its parallelization on two representative datasets, i.e., NELL and ACMCit (more than 1 million nodes). We only report the results of $\Sim_{\bjsim}\{\ub,\theta = 1\}$.}%, and the trends of other simulation variants are similar and thus are omitted due to short of space.

\sstitle{\color{black}Varying the Number of Threads.} \reffig{time_vary_thread} shows the running time of $\Sim_{\bjsim}\{\ub,\theta = 1\}$ by varying the number of threads from 1 to 32. We observe that both curves demonstrate reasonable decreasing trends as the number of threads increases. The benefits from 1 to 8 threads are substantial. After 8, the reward ratio flattens due to the cost of thread scheduling. %The deceasing trends are sharper at the start from 1 to 8 threads and eventually become smoother due to the cost of thread scheduling. 
{\color{black}Specifically, when setting $t=32$, parallelization can speed up the computation by 15 to 17 times of magnitude.}
%Note that the computation is faster on the larger ACMCit dataset. The reason is that it contains more labels than Amazon, which leads to better performance with label-constrained mapping.

\begin{figure}
	\centering
	\subfigure[varying the number of threads]{
		\centering
		\includegraphics[width=0.46\linewidth]{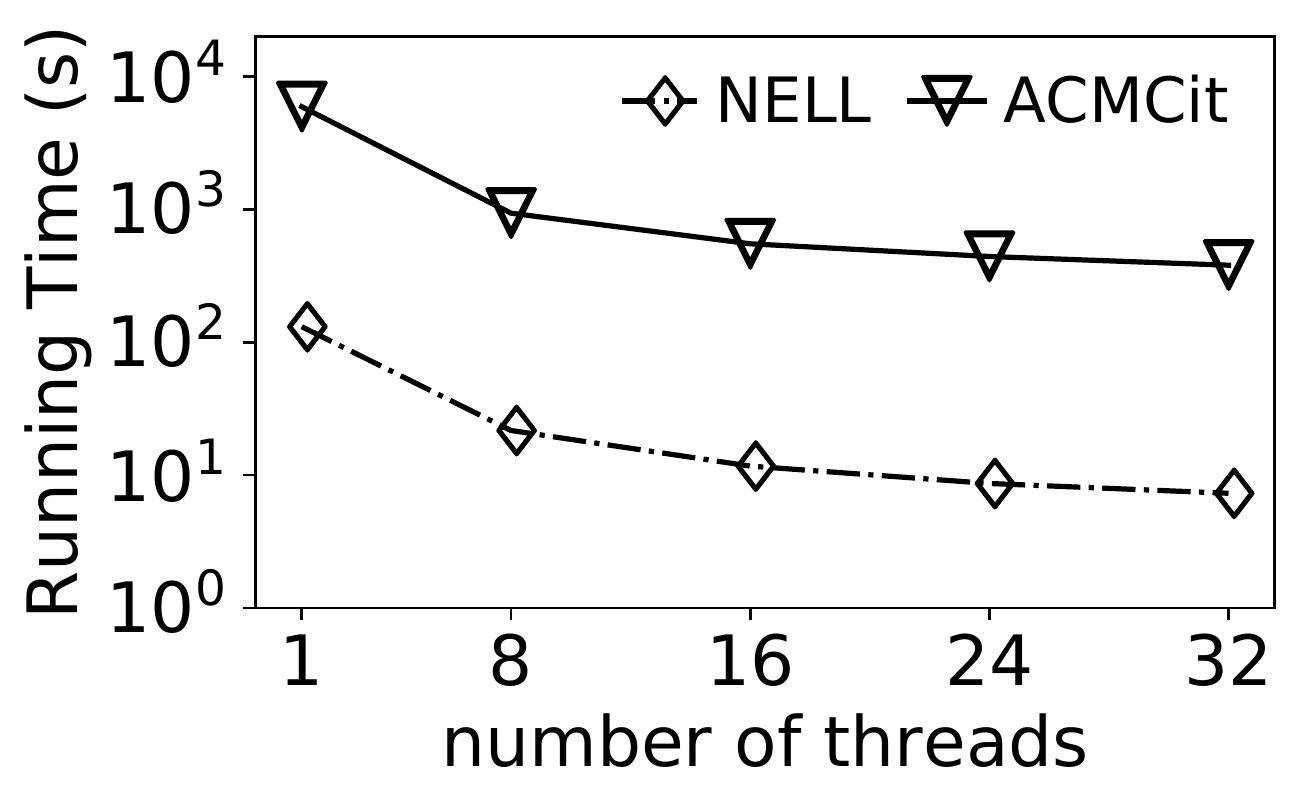}
		\label{fig:time_vary_thread}
	}
% 	\hspace{-0.5em}
	\subfigure[{varying density}]{
		\centering
		\includegraphics[width=0.46\linewidth]{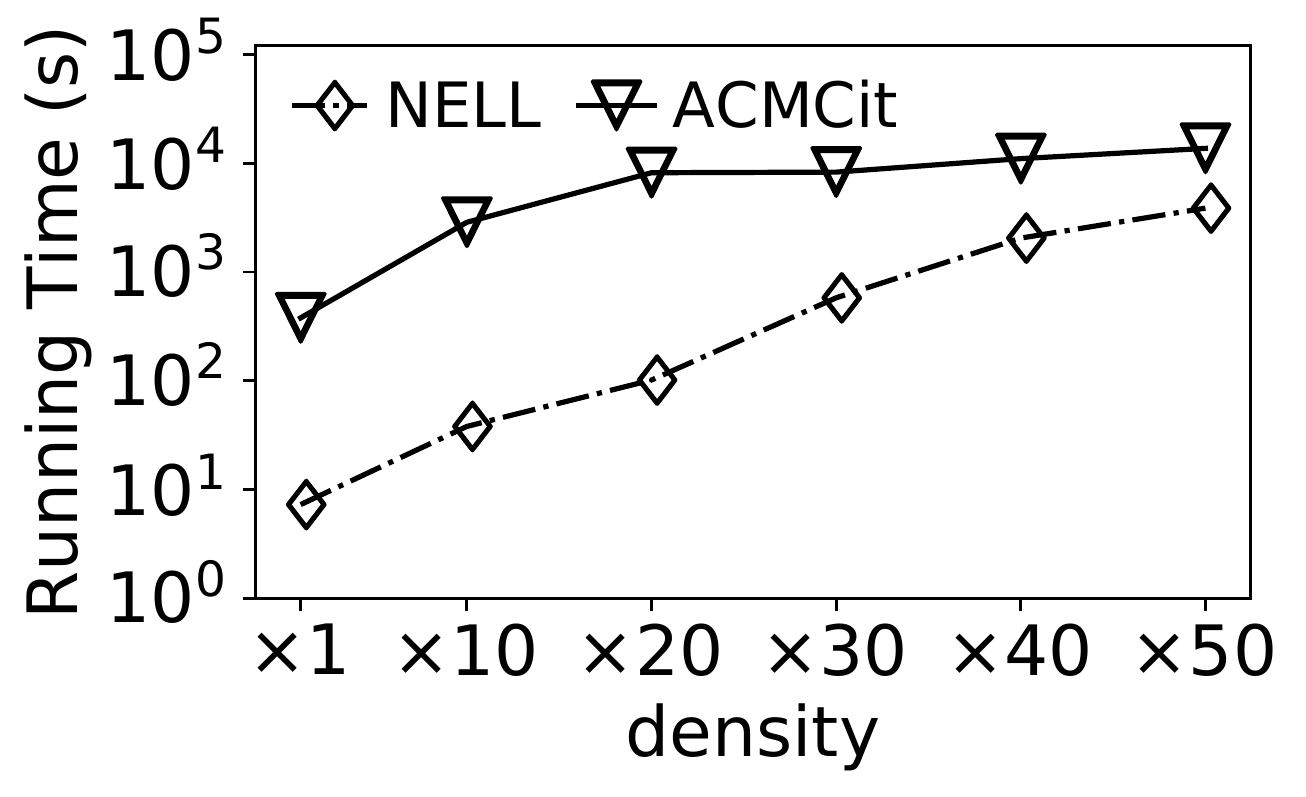}
		\label{fig:time_vary_density}
	}
% 	\vspace{-1.2em}
	\topcaption{\color{black}Parallelization and Scalability} \label{fig:parallelization_and_scalability}
% 	\vspace{-1.8em}
\end{figure}

\sstitle{\color{black}Varying Density.} {\color{black}\reffig{time_vary_density} reports the running time of $\Sim_{\bjsim}\{\ub,\theta = 1\}$ (with 32 threads) while varying the density of the datasets from $\times 10$ to $\times 50$ by randomly adding edges. Unsurprisingly, the running times of both grew longer as the graphs became denser. However, although increased density means greater computational complexity in theory, it also means each node has more neighbors by expectation. Hence, the upper bound in \refeq{upper_bound_label_constrained_mapping} may become smaller, which, in turn, contributes to greater upper-bound pruning power. This may offset some of the increase in computation complexity. Note that $\Sim_{\chi}$ finished within reasonable time on the ACMCit with $50\times$ more edges, indicating that it is scalable to the graphs with hundreds of millions of edges.}
%Having larger density, on the one hand, means larger computation complexity in theory. On the other hand, it causes each node to have more neighbors by expectation, so that the upper bound in \refeq{upper_bound_label_constrained_mapping} may become smaller, which in turn contributes to larger upper-bound pruning power. As a result, the increasing in complexity may be compensated.  Note that the computation of $\Sim_{\chi}$ finishes within reasonable time on the ACMCit with $50\times$ more edges, indicating that $\Sim_{\chi}$ is scalable to the graphs with hundreds of millions of edges.

%. Note that we create the datasets of density ratio from $\times 10$ to $\times 50$ from the original datasets by randomly adding edges. 

\subsection{Case Studies} \label{sec:case_studies}
In this subsection, we used three case studies to exhibit the potential of $\Sim_\chi$ in the applications of pattern matching, node similarity measurement and RDF graph alignment. We will demonstrate the following strengths of our framework.
%We study the effectiveness of $\Sim_\chi$ by applying it to three typical graph applications, namely pattern matching, RDF graph alignment and node similarity measurement. We will demonstrate the following strengths of our framework:%the applications of pattern matching, RDF graph alignment and node similarity measurement. %
% \vspace{-0.5em}
\begin{enumerate}[S1.] \setlength{\itemsep}{0cm}
    \item Our $\Sim_\chi$ framework quantifies the degree of simulation, which remedies the coarse ``yes-or-no'' semantics of simulation, significantly improves the effectiveness, and expand the scope of applying simulation.
    \item When multiple simulation variants are suitable for a certain application, the $\Sim_\chi$ framework provides a flexible way to experiment with all suitable variants, so as to determine the one that performs the best. 
    % \item Our $\Sim_\chi$ framework quantifies the degree of simulation, which helps produce promising results that would be otherwise missed out by the ``yes-or-no'' semantics of simulation.
    % \item When multiple simulation variants are suitable for a certain application, it is flexible to configure the $\Sim_\chi$ framework to experiment with all suitable variants, so as to pick out the one that performs the best. %, namely $\Sim_{\chi}$, stands out as it admits multiple configurations to fit into these variants.
    % %\item The fractional simulation score regarding the new variant, bijective simulation, is a promising candidate of node similarity measure on labeled graphs.
    % \item The $\Sim_\bjsim$ score is a promising candidate of node similarity measure on labeled graphs.
\end{enumerate}
% \vspace{-0.2em}

% \begin{table}
%     \centering
%     \small
% 	\topcaption{The results of pattern matching. $\mathcal{R}_x = \{\mathcal{G}_x^1, \mathcal{G}_x^2, \mathcal{G}_x^3\}$ is the top-3 result set given by $\Sim_{\simu}$ for $Q_x$ ($x \in \{1,2\}$).} \label{tab:pattern_matching}
%     \scalebox{0.86}{
%     \begin{tabular}{|c|c|} \hline
%     Pattern Graph & \makecell{Results of $\Sim_{\simu}$ (top 3 matches) and \\ Results of strong simulation (in dashed box)} \\ \hline
%     \includegraphics[width=0.22\linewidth]{figures/query1_new.pdf} & \includegraphics[width=0.75\linewidth]{figures/query1_all_new.pdf} \\ %\hline
%     \includegraphics[ width=0.19\linewidth]{figures/query2_new.pdf} & \includegraphics[width=0.75\linewidth]{figures/query2_all_new.pdf} \\ \hline
%     \end{tabular}}
%     \vspace{-1.8em}
% \end{table}

% \begin{figure}
% 	\centering
% 	\subfigure[$Q_1$]{
% 		\centering
% 		\includegraphics[width=0.25\linewidth]{figures/pattern_matching_query1.pdf}
% 		\label{fig:pattern_matching_query1}
% 	}
% 	\hspace{0.1em}
% 	\subfigure[$Q_2$]{
% 		\centering
% 		\includegraphics[width=0.23\linewidth]{figures/pattern_matching_query2.pdf}
% 		\label{fig:pattern_matching_query2}
% 	}
% % 	\hspace{-1.2em}
%     \hspace{0.5em}
% 	\subfigure[Result of $\Sim_{\chi}$]{
% 		\centering
% 		\includegraphics[width=0.24\linewidth]{figures/pattern_matching_fsim.pdf}
% 		\label{fig:pattern_matching_fsim}
% 	}
% % 	\vspace{-1.2em}
% 	\topcaption{Real-life matches on the Amazon co-purchasing data.} \label{fig:real-life-match}
% 	\vspace{-1.6em}
% \end{figure}

Before simply driving into the case studies, the first question to answer is: which simulation variant should be used for a given application? We discuss the answer intuitively. Subgraph pattern matching is essentially asymmetric (matching the pattern graph to the data graph but not the other way around), and thus $\Sim_{\simu}$ and $\Sim_{\dpsim}$ are appropriate choices. Node similarity measurement and graph alignment require symmetry, and hence $\Sim_{\bisim}$ and $\Sim_{\bjsim}$ are applied. The codes of all the baselines were provided by the respective authors. $\mathcal{L}(\cdot)$ was used indicator function since the semantics of node labels in the studied data were clear and without ambiguity. 

\stitle{Pattern Matching.} In this case study, we first considered strong simulation (exact simulation by nature, \cite{DBLP:journals/pvldb/MaCFHW11}) and $\dpsim$-simulation \cite{DBLP:journals/pvldb/SongGCW14} as two baselines, and compared them with $\Sim_{\simu}$ and $\Sim_{\dpsim}$ to illustrate how $\Sim_{\chi}$ facilitates pattern matching. \reffig{real-life-match} shows two example matches on the Amazon graph (see \reftab{graph_statistics} for graph statistics). When answering query $Q_1$, strong simulation (and $\dpsim$-simulation) returns $G_1$, pictured in \reffig{pattern_matching_g1}, which is also the top-1 result of $\Sim_\simu$ (and $\Sim_{\dpsim}$). Clearly, a simulation relation exists between $Q_1$ and $G_1$, and $\Sim_{\chi}$ captures $G_1$ with the highest score because of simulation definiteness (\refdef{fractional_simulation_properties}). $Q_2$ adds two extra nodes with new labels to $Q_1$ but, with this modification, both strong simulation and $\dpsim$-simulation fail to return a result while $\Sim_{\chi}$ returns $G_2$ (strength S1), which closely matches $Q_2$ by missing only an edge.

\begin{figure}
	\centering
	\subfigure[$Q_1$]{
		\centering
		\includegraphics[width=0.18\linewidth]{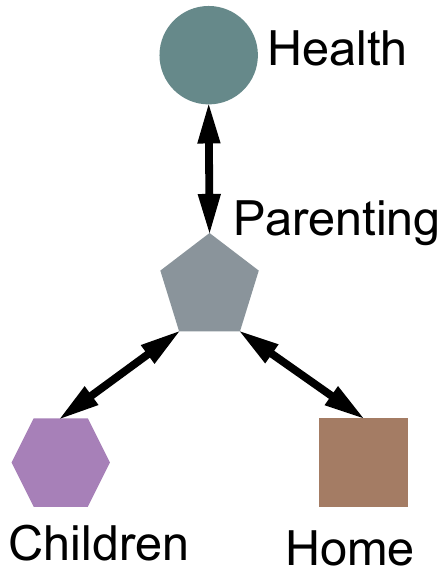}
		\label{fig:pattern_matching_q1}
	}
% 	\hspace{-0.8em}
	\subfigure[$G_1$]{
		\centering
		\includegraphics[width=0.17\linewidth]{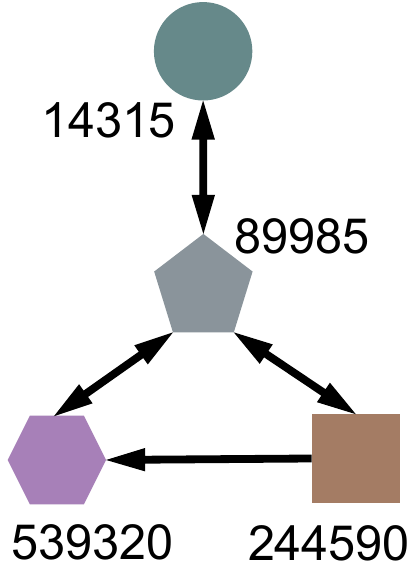}
		\label{fig:pattern_matching_g1}
	}
	\subfigure[$Q_2$]{
		\centering
		\includegraphics[width=0.17\linewidth]{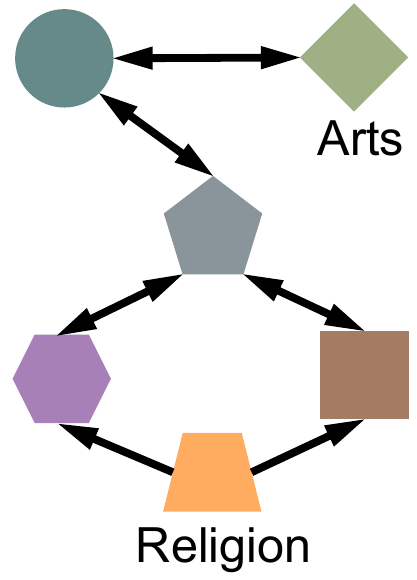}
		\label{fig:pattern_matching_q2}
	}
% 	\hspace{-0.8em}
	\subfigure[$G_2$]{
		\centering
		\includegraphics[width=0.18\linewidth]{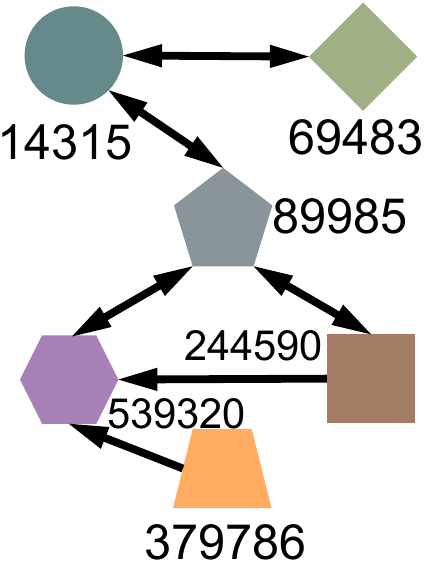}
		\label{fig:pattern_matching_g2}
	}
% 	\vspace{-1.2em}
	\topcaption{Real-life matches on the Amazon graph. $G_1$ and $G_2$ are top-1 matches of $\Sim_{\chi}$ for answering queries $Q_1$ and $Q_2$ respectively. Nodes in match $G$ are marked by their item ids, while nodes in query $Q$ are marked by their labels. Nodes with the same shape have the same label.} \label{fig:real-life-match}
% 	\vspace{-1.8em}
\end{figure}
% As a complete study, we 
% also compare the results of $\Sim_{\simu}$ and $\Sim_{\dpsim}$ with existing approximate pattern matching algorithms. There are two main-stream solutions in this field. The first one matches the nodes based on their similarity scores, and \kw{NAGA} is the representative algorithm \cite{DBLP:conf/www/DuttaN017}; the second one attempts to enumerate all matches that miss-match edges up to a given threshold, and \kw{TSpan} \cite{DBLP:conf/sigmod/ZhuLZZY12} is the representative. Therefore, in addition to strong simulation \cite{DBLP:journals/pvldb/MaCFHW11,DBLP:journals/tods/MaCFHW14} ($\dpsim$-simulation is not compared as its result is similar to strong simulation), we also use \kw{NAGA} and \kw{TSpan} as the baselines.

{\color{black}For a more complete study, we also compared the results of $\Sim_{\chi}$ with some other approximate pattern matching algorithms. The related algorithms can be summarized in two categories: (1) the edit-distance based algorithms, e.g., \kw{SAPPER} \cite{DBLP:journals/pvldb/ZhangYJ10} and \kw{TSpan} \cite{DBLP:conf/sigmod/ZhuLZZY12}, which enumerate all matches with mismatched edges up to a given threshold; and (2) the similarity-based algorithms that compute matches based on (sub)graph similarity or node similarity. To name a few, G-Ray \cite{DBLP:conf/kdd/TongFGE07} computes the ``goodness'' of a match based on node proximity. \kw{IsoRank} \cite{DBLP:journals/pnas/0001XB08}, \kw{NeMa} \cite{DBLP:journals/pvldb/KhanWAY13} and \kw{NAGA} \cite{DBLP:conf/www/DuttaN017} find matches based on node similarity. \kw{G}-\kw{Finder} \cite{DBLP:conf/bigdataconf/LiuDXT19} and \kw{SAGA} \cite{DBLP:journals/bioinformatics/TianMSSP07} design cost functions with multiple components to allow node mismatches and graph structural differences. \kw{SLQ} \cite{DBLP:journals/pvldb/YangWSY14} and $S^4$ \cite{DBLP:journals/pvldb/ZhengZPYSZ16} find matches in RDF knowledge graphs by considering the semantics of queries. More specifically, $S^4$ uses the semantic graph edit distance to integrate structure similarity and semantic similarity. Note that, in the Amazon graph, an edge from $u$ to $v$ indicates that people are highly likely to buy item $v$ after buying item $u$ \cite{DBLP:journals/pvldb/MaCFHW11}, and hence there is no complex semantics among edges. As a result, we choose \kw{TSpan}, \kw{NAGA} and \kw{G}-\kw{Finder}, the state-of-the-art algorithms in each category, as another three baselines.}

{\color{black}%We follow \kw{NAGA} \cite{DBLP:conf/www/DuttaN017} for matches generation and quality evaluation.
We followed the state-of-the-art algorithm \kw{NAGA} \cite{DBLP:conf/www/DuttaN017} for match generation and quality evaluation. Briefly, node pairs with high $\Sim_{\chi}$ scores are considered to be ``seeds'', and matches are generated by expanding the regions around the ``seeds'' subsequently.} The evaluated queries are generated randomly by extracting subgraphs from the data graph and introducing structural noises (randomly insert edges, up to 33\%) or label noises (randomly modify node labels, up to 33\%). We then evaluated different algorithms across four query scenarios: (1) queries with no noises (\kw{Exact}); (2) queries with structural noises only (\kw{Noisy}-\kw{E}); (3) queries with label noises only (\kw{Noisy}-\kw{L}); and (4) queries with both kinds of noises (\kw{Combined}). Note that the queries are extracted from the graphs, which naturally serve as the ``ground truth''. {\color{black}Given a query $Q$ and a returned match $\phi$ (we use top-1 match in this case study), the $F_1$ score is calculated by $F_1 = \frac{2\cdot P\cdot R}{(P+R)}$, where $P = \frac{|\phi_t|}{|\phi|}$, $R = \frac{|\phi_t|}{|Q|}$, $\phi_t$ is a subset of $\phi$ that includes the correctly discovered node matches in $\phi$, and $|X|$ indicates the number of nodes in the match or graph, $\forall X \in \{\phi_t,\phi,Q\}$}.

\begin{table}
    \centering
    \small
    \topcaption{Average F1 scores (\%) while answering queries in different scenarios on the Amazon dataset. \kw{TSpan}-x indicates miss-matching up to $x$ edges in \kw{TSpan}.} \label{tab:pattern-matching-fscore}
    \renewcommand{\arraystretch}{1.05}
    % \scalebox{0.8}{
    % \setlength{\tabcolsep}{1mm}{
    \color{black}
    \begin{tabular}{|c||ccccc|cc|} \hline
        \multirow{2}{*}{\makecell{Query\\Scenario}} & \multicolumn{5}{c|}{Baselines} & \multicolumn{2}{c|}{$\Sim_{\chi}$} \\ \cline{2-8}
         &\kw{NAGA} & \kw{G}-\kw{Finder} & \kw{TSpan}-1 & \kw{TSpan}-3 & \makecell{Strong \\ Simulation} & $\Sim_{\simu}$ & $\Sim_{\dpsim}$ \\ \hline \hline
        \kw{Exact} & 30.2 & \textbf{100} & \textbf{100} & \textbf{100} & \textbf{100} & \textbf{100} & \textbf{100} \\ 
        {\kw{Noisy}-\kw{E}} & 30.5 & 49.2 & 71.0 & \textbf{95.8} & 50.0 & 84.0 & 65.7 \\ 
        {\kw{Noisy}-\kw{L}} & 20.6 & 40.7 & - & - & 33.3 & \textbf{75.1} & 73.2 \\ 
        \kw{Combined} & 21.2 & 40.9 & - & - & 29.2 & \textbf{76.6} & 66.7 \\ \hline
        % \textbf{Average} & 2.4 & 25.6 & 55.0 &78.9 & 80.5 \\ \hline
    \end{tabular}
    % }}
    % \vspace{-2em}
\end{table}

% \begin{table}
%     \centering
%     \small
%     \topcaption{Average F1 score (\%) over varying pattern matching queries for the Amazon dataset} \label{tab:pattern-matching-fscore}
%     \renewcommand{\arraystretch}{1.05}
%     \scalebox{0.78}{
%     \begin{tabular}{|c||cccc|cc|} \hline
%         \multirow{2}{*}{Query Scenario} & \multicolumn{4}{c|}{Baselines} & \multicolumn{2}{c|}{$\Sim_{\chi}$} \\ \cline{2-7}
%          &\kw{NAGA} & \kw{TSpan}-1 & \kw{TSpan}-3 & \makecell{Strong \\ Simulation} & $\Sim_{\simu}$ & $\Sim_{\dpsim}$ \\ \hline \hline
%         Exact Match & 30.2 & 100 & 100 & 100 & 100 & 100 \\ 
%         Noisy Edges & 30.5 & 71.0 & 95.8 & 50.0 &76.4 & 74.7 \\ 
%         Noisy Labels & 20.6 & - & - & 33.3 &70.4 & 79.3\\ 
%         Combined & 21.2 & - & - & 29.2 &68.8 & 68.1 \\ \hline
%         % \textbf{Average} & 2.4 & 25.6 & 55.0 &78.9 & 80.5 \\ \hline
%     \end{tabular}}
%     \vspace{-2em}
% \end{table}

\reftab{pattern-matching-fscore}\footnote{The results of \kw{NAGA} are provided by the authors and we acknowledge the assistance from Dr. Sourav Dutta and Dr. Shubhangi Agarwal.} shows the F1 scores of different algorithms. The result is an average from 100 random queries of sizes ranging from 3 to 13. $\dpsim$-simulation was not compared as it is similar to strong simulation. {\color{black}As with the last results, strong simulation performed poorly against noise. In comparison, $\Sim_{\chi}$ was more robust and performed much better (strength S1).} %We then compare $\Sim_{\chi}$ with \kw{NAGA}, \kw{G}-\kw{Finder} and \kw{TSpan}. From \reftab{pattern-matching-fscore}, 
{\color{black}Additionally, $\Sim_{\simu}$ outperformed \kw{NAGA}, \kw{G}-\kw{Finder} and \kw{TSpan}-1 by a big margin on all query scenarios.} \kw{TSpan}-3 performed well in ``\kw{Exact}'' and ``\kw{Noisy}-\kw{E}'' with its highest F1 score of 95.8\% for ``\kw{Noisy}-\kw{E}''. This is because \kw{TSpan}-3 finds all matches with up to 3 mismatched edges, which is not less than the number of noisy edges in most queries. However, \kw{TSpan} favors the case with missing edges rather than nodes. Thus, it has no results for ``\kw{Noisy}-\kw{L}'' and ``\kw{Combined}''. In summary, $\Sim_{\chi}$ is qualified for approximate pattern matching (strength S1). While both $\simu$- and $\dpsim$-simulation can be configured for the application, $\Sim_{\simu}$ is more robust to noises and performs better than $\Sim_{\dpsim}$ (strength S2). %which reflects the flexibility of $\Sim_{\chi}$ to be experimented for certain application (strength S2).

\stitle{Node Similarity Measurement.} In this case study, we compared $\Sim_{\chi}$ to four state-of-the-art similarity measurement algorithms: \PCRW \cite{DBLP:journals/ml/LaoC10}, \PathSim \cite{DBLP:journals/pvldb/SunHYYW11}, \JoinSim \cite{DBLP:journals/tkde/XiongZY15} and \nSimGram \cite{DBLP:conf/kdd/ConteFGMSU18}. Following \cite{DBLP:conf/kdd/ConteFGMSU18,DBLP:journals/pvldb/SunHYYW11}, we used the DBIS dataset, which contains 60,694 authors, 72,902 papers and 464 venues. In DBIS, the venues and papers are labeled as ``V'' and ``P", respectively. The authors are labeled by their names.

We first computed the top-5 most similar venues to WWW using all algorithms. The results are shown in \reftab{top_10_venue}. Note that WWW$_1$, WWW$_2$ and WWW$_3$ all represent the WWW venue but with different node ids in DBIS, and thus they are naturally similar to WWW. Although all algorithms gave reasonable results, $\Sim_\bjsim$ was the only one to return WWW$_1$, WWW$_2$ and WWW$_3$ {\color{black}in the top-5 results. In addition, if we applied exact $\bisim$- and $\bjsim$-simulation to the task, other than ``WWW'' itself (``Yes''), all the other venues had the same score (``No''). This shows that $\Sim_{\chi}$ can be applied to the scenarios that require fine-grained evaluation, such as node similarity measurement (strength S1).}

\begin{table} [h]
    \centering
    \small
    \topcaption{The top-5 similar venues for ``WWW'' of different algorithms} \label{tab:top_10_venue}
    \renewcommand{\arraystretch}{0.92}
    % \scalebox{0.78}{
    % \setlength{\tabcolsep}{1.7mm}{
    \begin{tabular}{|c|c|c|c|c|c|c|} \hline
        Rank & \PCRW & \PathSim & \JoinSim & \nSimGram & $\Sim_{\bisim}$ & $\Sim_{\bjsim}$ \\ \hline \hline
        1 & WWW & WWW & WWW & WWW & WWW & WWW \\
        2 & SIGIR & CIKM & WWW$_1$ & CIKM & CIKM & WWW$_1$ \\
        3 & ICDE & SIGKDD & CIKM & SIGIR & ICDE & CIKM \\
        4 & VLDB & WISE & WSDM & WWW$_1$ & VLDB & WWW$_2$ \\
        5 & Hypertext & ICDM & WWW$_2$ & SIGKDD & SIGIR & WWW$_3$ \\ \hline
        % 6 & SIGKDD & WWW$_1$ & T.I.T. & TKDE & SIGMOD & SIGIR \\
        % 7 & CIKM & SIGIR & TWEB & ICDE & SIGKDD & WebDB \\
        % 8 & SIGMOD & APWeb & WISE & VLDB & TKDE & T.I.T. \\
        % 9 & TKDE & TKDE & WWW$_3$ & WISE & Hypertext & AusWeb \\
        % 10 & WWW$_1$ & EDBT & SIGIR & SIGMOD & WISE & WISE \\\hline
    \end{tabular}
    % }}
    % \vspace{-1.6em}
\end{table}

Following \cite{DBLP:conf/kdd/ConteFGMSU18,DBLP:journals/pvldb/SunHYYW11}, we further computed the top-15 most similar venues to 15 subject venues (same as \cite{DBLP:conf/kdd/ConteFGMSU18}) of each algorithm.
%(ICDE, KDD, ICDM, PKDD, PAKDD, WWW, SIGIR, SIGMOD, VLDB, PODS, EDBT, DASFAA, SDM, TREC and AP-Web) to evaluate the above algorithms. 
For each subject venue, we labeled each returned venue with a relevance score: 0 for non-relevant, 1 for some-relevant, and 2 for very-relevant, considering both the research area and venue ranking in \cite{core_ranking}. For example, the relevance score for ICDE and VLDB is 2 as both are top-tier conferences in the area of database. We then evaluated the ranking quality of the algorithms using nDCG (the larger the score, the better). 

\begin{table} [h]
    \small
    % \begin{minipage}{\linewidth}
    \centering
    \topcaption{NDCG results of node similarity algorithms} 
    \label{tab:node_similarity_results}
    % \scalebox{0.86}{
    \begin{tabular}{|c|c|c|c|c|c|} \hline
    \multicolumn{4}{|c|}{Baselines} & \multicolumn{2}{c|}{Fractional $\chi$-simulation} \\ \hline
    \PCRW & \PathSim & \JoinSim & \nSimGram & $\Sim_{\bisim}$ & $\Sim_{\bjsim}$ \\ \hline\hline
    0.684 & 0.684 & 0.689 & 0.700 & 0.699& \textbf{0.733}\\ \hline
    \end{tabular}
    % }
    %  \end{minipage}
    % \vspace{-2em}
\end{table}
% ~
% \begin{minipage}{\linewidth}
% \begin{table}
%   \centering
% % \vspace{-1em}
%      \topcaption{The F1 scores (\%) of graph alignment. $x$-bisim indicates setting $k=x$ in $k$-bisimulation.} \label{tab:graph_alignment_results}
%     %\vspace{-2.5em}
%     \scalebox{0.93}{
%     \setlength{\tabcolsep}{1.2mm}{
%     \color{black}
%     \begin{tabular}{|c||cccccc|cc|} \hline
%         \multirow{2}{*}{Graphs} & \multicolumn{6}{c|}{Baselines} & \multicolumn{2}{c|}{$\Sim_{\chi}$} \\ \cline{2-9} & $2$-bisim & $4$-bisim & Olap & \kw{\textsc{gsaNA}} & \kw{FINAL} & \kw{EWS} & $\Sim_{\bisim}$& $\Sim_{\bjsim}$ \\ \hline\hline
%         $G_1$-$G_2$ & 19.9 & 9.1 & 37.9 & 11.8 & 55.2 & 70.8 & \textbf{97.6} & 96.5 \\ 
%         $G_1$-$G_3$ & 53.0 & 10.9 & 37.6 & 14.9 & 52.7 & 65.3 & \textbf{96.9} & 95.6 \\ \hline
%     \end{tabular}}}
% % \end{minipage}
% % \begin{minipage}{\linewidth}
% % %\begin{table}
% %   \centering
% % % \vspace{-1em}
% %      \topcaption{The F1 scores (\%) of graph alignment} \label{tab:graph_alignment_results}
% %     %\vspace{-2.5em}
% %     \scalebox{0.91}{
% %     \begin{tabular}{|c||cccc|cc|} \hline
% %         \multirow{2}{*}{Graphs} & \multicolumn{4}{c|}{Baselines} & \multicolumn{2}{c|}{Fractional $\chi$-simulation} \\ \cline{2-7} & Olap & \kw{\textsc{gsaNA}} & \kw{FINAL} & \kw{EWS} & $\Sim_{\bisim}$& $\Sim_{\bjsim}$ \\ \hline\hline
% %         $GP_1-GP_2$ & 37.9 & 11.8 & 55.2 & 70.8 & \textbf{97.6} & 96.5 \\ 
% %         $GP_1-GP_3$ & 37.6 & 14.9 & 52.7 & 65.3 & \textbf{96.9} & 95.6 \\ \hline
% %     \end{tabular}}
% % \end{minipage}
% \vspace{-2em}
% \end{table}

\reftab{node_similarity_results} shows the nDCG results. Accordingly, $\Sim_\chi$ outperforms the state-of-the-art algorithms by a large margin. This indicates that $\Sim_\chi$ is qualified to measure node similarity on labeled graphs (strength S1). The result that $\Sim_\bjsim$ outperforms $\Sim_\bisim$ in both the ``WWW'' case and the general evaluation suggests $\Sim_\bjsim$ is a better candidate for similarity measurement (strength S2).

\stitle{RDF Graph Alignment.} We investigate the potential of $\Sim_{\chi}\!$ in RDF graph alignment and briefly discuss its performance below. We followed Olap \cite{DBLP:journals/pvldb/BunemanS16} (a bisimulation-based alignment algorithm) to align three different versions of biological graphs from different times, $G_1$, $G_2$ and $G_3$ \cite{harding2017iuphar}. $G_1$ has 133,195 nodes and 273,512 edges, $G_2$ has 138,651 nodes and 285,000 edges, $G_3$ includes 144,879 nodes and 298,564 edges, and all of them have 8 node labels and 23 edge labels. Note that the original URI values in these datasets do not change over time. Hence, we can use this information to identify the ground truth alignment. In addition to Olap, we also included another four state-of-the-art algorithms, {\color{black}namely $k$-bisimulation \cite{DBLP:conf/sac/HeeswijkFP16}, \kw{\textsc{gsaNA}} \cite{DBLP:conf/kdd/YasarC18}, \kw{FINAL} \cite{DBLP:conf/kdd/ZhangT16} and \kw{EWS} \cite{DBLP:journals/pvldb/KazemiHG15}. When aligning graphs with $\Sim_{\chi}$, a node $u \in V_1$ will be aligned to a node set $A_u = \argmax_{v \in V_2}\Sim_\chi(u, v)$, while with $k$-bisimulation, $u$ will be aligned to $A_u = \{v|v \in V_2 \wedge u \text{ and } v \text{ are bisimilar}\}$. The F1 score of $\Sim_{\chi}$ and $k$-bisimulation is calculated by $F1 = \sum_{u\in V_1}\frac{2\cdot P_u \cdot R_u}{|V_1|(P_u + R_u)}$, where $P_u$ (resp. $R_u$) is $\frac{1}{|A_u|}$ (resp. 1) if $A_u$ contains the ground truth, and 0 otherwise. We follow the settings in the related papers for the other baselines. }
%As baselines, we also include an embedding-based algorithm \kw{\textsc{gsaNA}} \cite{DBLP:conf/kdd/YasarC18}, a similarity-based algorithm \kw{FINAL} \cite{DBLP:conf/kdd/ZhangT16} and a PGM-based (percolation graph matching) algorithm \kw{EWS} \cite{DBLP:journals/pvldb/KazemiHG15}. 

\reftab{graph_alignment_results} reports the F1 scores of each algorithm. {\color{black}Note that we also tested the bisimulation, which resulted in 0\% F1 scores in both cases since there is no exact bisimulation relation between two graphs. $k$-bisimulation performs better than bisimulation as it, to some extent, approximates bisimulation.} From \reftab{graph_alignment_results}, our $\Sim_{\chi}$ had the highest F1 scores and thus outperformed all the other baselines. This shows that we can apply $\Sim_\bisim$ and $\Sim_\bjsim$ with high potential for graph alignment (strength S1). $\Sim_{\bisim}$ outperforms $\Sim_{\bjsim}$ and thus is a better candidate for graph alignment (strength S2).

\begin{table}
  \centering
% \vspace{-1em}
     \topcaption{The F1 scores (\%) of each algorithm when aligning two graphs. $x$-bisim indicates setting $k=x$ in $k$-bisimulation.} \label{tab:graph_alignment_results}
    %\vspace{-2.5em}
    % \scalebox{0.92}{
    % \setlength{\tabcolsep}{1.1mm}{
    \color{black}
    \begin{tabular}{|c||cccccc|cc|} \hline
        \multirow{2}{*}{Graphs} & \multicolumn{6}{c|}{Baselines} & \multicolumn{2}{c|}{$\Sim_{\chi}$} \\ \cline{2-9} & $2$-bisim & $4$-bisim & Olap & \kw{\textsc{gsaNA}} & \kw{FINAL} & \kw{EWS} & $\Sim_{\bisim}$& $\Sim_{\bjsim}$ \\ \hline\hline
        $G_1$-$G_2$ & 19.9 & 9.1 & 37.9 & 11.8 & 55.2 & 70.8 & \textbf{97.6} & 96.5 \\ 
        $G_1$-$G_3$ & 53.0 & 10.9 & 37.6 & 14.9 & 52.7 & 65.3 & \textbf{96.9} & 95.6 \\ \hline
     \end{tabular}
    % }}
% \vspace{-2.2em}
\end{table}

{\color{black}\stitle{Efficiency Evaluation.} Given the superior effectiveness of $\Sim_{\chi}$ in the above case studies, one may also be interested in its efficiency. Next, we show the running time of $\Sim_{\chi}$ (with 32 threads) and the most effective baseline in each case study. We will also report the running time of exact simulation (or its variant) if it is applied and effective in the case study. For pattern matching, $\Sim_{\chi}$ on average took 0.25s for each query. In comparison, exact simulation took around 1.2s, and \kw{TSpan}, the most effective baseline, spent more than 70s. In similarity measurement, \nSimGram took 0.03ms to compute a single node pair, while $\Sim_{\chi}$ finished the computation within 6500s for 134060$\times$134060 pairs or roughly 0.0004ms per pair. In graph alignment, $k$-bisimulation ($k=4$) spent 0.4s for the computation, and \kw{EWS} spent 1496s. {Our $\Sim_{\chi}$ 
{ran a bit slower than \kw{EWS} and took 3120s}, which is tolerable as {it is much more effective} than the other algorithms.} Note that it is not straightforward and potentially unfair to compare with all the baselines as they either focus on per-query computation (e.g., \PathSims and \JoinSim) or have been individually implemented in different languages (e.g., Olap in Python and \kw{FINAL} in Matlab).}
\section{Related Work} \label{sec:related_work}
% \vspace{-0.6em}
\stitle{Simulation and Its Variants.} 
In this paper, we focused on four simulation variants: simple simulation \cite{DBLP:conf/ijcai/Milner71, DBLP:journals/pvldb/MaCFHW11}, bisimulation \cite{milner1989communication}, degree-preserving simulation \cite{DBLP:journals/pvldb/SongGCW14} and bijective simulation. The original definition of simulation \cite{DBLP:conf/ijcai/Milner71} only considered out-neighbors, but Ma et al.'s redefinition in 2011 \cite{DBLP:journals/pvldb/MaCFHW11} takes in-neighbors into account and hence is the definition we used. Reverting to the original definition is as easy as setting $w^- = 0$ in our framework. {\color{black}Additionally, we discussed a variant of approximate bisimulation, namely $k$-bisimulation \cite{DBLP:books/cu/12/AcetoIS12,DBLP:conf/bncod/LuoLFBHW13,DBLP:conf/cikm/LuoFHWB13,DBLP:conf/sac/HeeswijkFP16}, and investigated its relation to our framework} (\refsec{discussion}). There are other variants that have not yet included in the framework, including bounded simulation \cite{DBLP:journals/pvldb/FanLMTWW10} and weak simulation \cite{ milner1989communication}. These variants consider the $k$-hop neighbors ($k\geq 1$) in addition to the immediate neighbors. As an interesting future work, we will study to incorporate them in our framework. {\color{black}There are also some algorithms that aim to compute simulation (variants) efficiently and effectively, e.g., a hash-based algorithm in \cite{DBLP:conf/sac/HeeswijkFP16}, external-memory algorithms in \cite{DBLP:conf/sigmod/HellingsFH12,DBLP:conf/cikm/LuoFHWB13}, a distributed algorithm in \cite{DBLP:conf/bncod/LuoLFBHW13} and a partition refinement algorithm in \cite{DBLP:journals/acta/Ranzato14}. However, all these algorithms compute the ``yes-or-no'' simulation (or its variant) and cannot provide fractional scores as proposed in this paper.}

% \stitle{Simulation and Its Variants.} Simulation was first defined by Robin Milner \cite{DBLP:conf/ijcai/Milner71} to determine if two programs are the realizations of the same algorithm. This definition was further extended to form other simulation variants by adding more constraints. Bisimulation \cite{milner1989communication} is based on simulation $R$ while it requires $R^{-1}$ to be a simulation as well. Dual simulation \cite{DBLP:journals/pvldb/MaCFHW11} was proposed to take nodes' in-neighbors into consideration in order to capture more topology information. Strong simulation \cite{DBLP:journals/pvldb/MaCFHW11} is dual simulation by nature specifically designed for graph pattern matching. Degree-preserving dual simulation \cite{DBLP:journals/pvldb/SongGCW14} extends the definition of dual simulation by requiring injective mapping between nodes' neighbors. In this paper, we show that all these simulation variants can be incorporated into our fractional $\chi$-simulation framework. There exist some other simulation variants, e.g., bounded simulation \cite{DBLP:journals/pvldb/FanLMTWW10} and weak simulation \cite{ milner1989communication}. Note that these variants consider nodes' $k$-hop neighbors ($k\geq 1$) instead of immediate neighbors, thus we currently do not cover these variants. As an interesting future work, we will study to incorporate them in our framework.

{\color{black} \stitle{Node Similarity Measures.}} %We then review existing node similarity measures. 
{\color{black} We have shown that $\Sim_{\bjsim}$ is qualified for node similarity measurement. Thus, we review node similarity measures on labeled graphs.} \SimRank \cite{DBLP:conf/kdd/JehW02} and \RoleSim \cite{DBLP:conf/kdd/JinLH11} are two representative measures, and their relations to our $\Sim_{\chi}$ have been discussed in \refsec{discussion}. %Another measure proposed in \cite{DBLP:journals/siamrev/BlondelGHSD04} computes node similarity in a mutually reinforcing scheme. 
As these two measures are less effective in computing node similarity on labeled graphs \cite{DBLP:conf/kdd/ConteFGMSU18,DBLP:journals/pvldb/SunHYYW11}, similarity measures \cite{DBLP:conf/kdd/ConteFGMSU18,DBLP:conf/kdd/HuangZCSML16,DBLP:journals/pvldb/SunHYYW11,DBLP:journals/tkde/XiongZY15} were proposed. \PathSim \cite{DBLP:journals/pvldb/SunHYYW11}, for instance, uses a ratio of meta-paths connecting two nodes as the measure. \JoinSim \cite{DBLP:journals/tkde/XiongZY15} is similar to \PathSim, but it satisfies the triangle inequality. \nSimGram \cite{DBLP:conf/kdd/ConteFGMSU18} computes node similarity based on q-grams instead of meta-paths to capture more topology information. {\color{black}Note that these measures cannot substitute our work as their scores are not related to simulation and thus are not suitable to quantify the extent of simulation.}

%{\color{red}Note that these measures cannot substitute our work as: 1) they are all symmetric and cannot be applied to asymmetric applications such as pattern matching; 2) their scores are not related to simulation and thus are not suitable to quantify the degree of simulation.}

{\color{black} %\stitle{Node Similarity Techniques.} We haven shown that the fractional scores computed by $\Sim_{\chi}$ can be applied to subgraph pattern matching, graph alignment and node similarity measurement. We then review related works in the following two aspects, i.e., node similarity techniques in the subgraph pattern matching and graph alignment, and node similarity measures. 
{\color{black} \stitle{Similarity-based Applications.} There are a number of works on pattern matching and graph alignment that are based on node similarity techniques. These works may differ in measuring node similarities. Specifically,} \kw{IsoRank} \cite{DBLP:journals/pnas/0001XB08} computes the similarity between two nodes based on an weighted average of their neighbors' scores. \kw{NeMa} \cite{DBLP:journals/pvldb/KhanWAY13} defines a vector that encodes the neighborhood information for each node. The distance between two nodes is then computed from these vectors. \kw{NAGA} \cite{DBLP:conf/www/DuttaN017} leverages statistical significance through chi-square measure to compute node similarity. \kw{REGAL} \cite{DBLP:conf/cikm/HeimannSSK18} measures the similarity of two nodes by taking the information of $k$-hop neighbors into account. \kw{FIRST} \cite{DBLP:conf/kdd/DuZCT17} and \kw{FINAL} \cite{DBLP:conf/kdd/ZhangT16} use a Sylvester equation to compute similarities, which encodes structural consistency and attribute consistency of two networks. {\color{black}For similar reasons, these works are also not suitable to quantify the degree of simulation.}}

\section{Conclusion} \label{sec:conclusion}
In this paper, we formally define fractional $\chi$-simulation to quantify the degree to which one node simulates another by a $\chi$-simulation. We then propose the $\Sim_{\chi}$ computation framework to realize the quantification for all $\chi$-simulations. We conduct extensive experiments to demonstrate the effectiveness and efficiency of the fractional $\chi$-simulation framework. {\color{black}Considering end-users are also interested in the top-k similarity search. In the future, we plan to devise efficient techniques to process top-k queries based on the $\Sim_{\chi}$.}
% In this paper, we formally define fractional $\chi$-simulation to quantify the degree that node $u$ is simulated approximately by node $v$ regarding the simulation variant $\chi$. We then propose the fractional $\chi$-simulation framework that admits flexible configurations to realize different simulation variants for various applications. We conduct extensive experiments and case studies to demonstrate the effectiveness of the proposed techniques and the strengths of the fractional $\chi$-simulation framework.

% \bibliographystyle{unsrt}  
% \bibliography{main}  %%% Remove comment to use the external .bib file (using bibtex).

\end{document}